\documentclass[twoside,11pt]{article}
\usepackage{latexsym}
\usepackage{epsfig}
\usepackage{amsfonts}
\usepackage{amsthm,amsmath,amsfonts,amssymb}
\numberwithin{equation}{section}

\newcommand{\supp}{\operatorname{supp}}

\newcommand{\opm}{\operatorname{op}}
\newcommand{\Op}{\operatorname{Op}}
\newcommand{\tub}{\operatorname{tub}}
\newcommand{\edge}{\operatorname{edge}}
\newcommand{\interior}{\operatorname{int}}
\newcommand{\smooth}{\operatorname{smooth}}
\newcommand{\gflat}{\operatorname{flat}}
\newcommand{\LS}{\operatorname{LS}}
\newcommand{\res}{\operatorname{Res}}
\newcommand{\M}{\operatorname{M}}
\newcommand{\G}{\operatorname{G}}

\def\dbar{{\mathchar'26\mkern-12mud}}

\begin{document}
\newtheorem{assumption}{Assumption}
\newtheorem{proposition}{Proposition}
\newtheorem{definition}{Definition}
\newtheorem{lemma}{Lemma}
\newtheorem{theorem}{Theorem}
\newtheorem{observation}{Observation}
\newtheorem{remark}{Remark}
\newtheorem{corollary}{Corollary}

\title{Explicit Green operators for quantum mechanical Hamiltonians.~II.~ Edge type singularities of the helium atom}

\author{Heinz-J{\"u}rgen Flad$^\ast$, Gohar Flad-Harutyunyan$^\ast$,
and Bert-Wolfgang Schulze$^\ddag$ \\
\ \\
$^\ast${\small Zentrum Mathematik, Technische Universit\"at M\"unchen,
Boltzmannstr 3, D-85748 Garching, Germany}\\
$^\ddag${\small Institut f\"ur Mathematik, Universit\"at Potsdam, Karl-Liebknecht-Stra{\ss}e 24-25,
D-14476 Potsdam-Golm, Germany}
}
\maketitle

\begin{abstract}
\noindent
We extend our approach of asymptotic parametrix construction for Hamiltonian operators from conical to edge-type singularities which is applicable to coalescence points of two particles of the helium atom 
and related two electron systems including the hydrogen molecule. Up to second order we have calculated the symbols
of an asymptotic parametrix of the nonrelativisic Hamiltonian of the helium atom within the Born-Oppenheimer approximation
and provide explicit formulas for the corresponding Green operators
which encode the asymptotic behaviour of the eigenfunctons near an edge.
\end{abstract}

\section{Introduction}
\subsection{Singular analysis meets quantum chemistry}
\label{samqch}
In a reductionistic approach to chemistry it is the primary goal to explain 
chemical phenomena in terms of underlying fundamental physical principles. Here a key role
has electronic structure theory, where properties of molecules and solids are derived from ``first principles'' by solving Schr\"odinger's equation for many-particle systems interacting via Coulomb potentials.
Apparent singularities of Coulomb potentials at coalescence points of particles 
are reflected by the regularity properties of these solutions. Concerning their global regularity, 
Kato showed that solutions of an $N$ particle Schr\"odinger equation belong to the Sobolev space $H^2(\mathbb{R}^{3N})$. 
Their maximal Sobolev regularity is not much higher, actually it is not difficult to see
that already for a hydrogen atom the corresponding solutions are in $H^s(\mathbb{R}^{3})$ for $s < 5/2$ only.
An alternative measure of global regularity are Sobolev spaces with mixed partial derivatives which are
of particular significance with respect to numerical analysis. Basic results for such spaces have been recently obtained 
by Yserentant \cite{Yser04,Yser05}.

Besides global regularity it is rewarding to study in detail the asymptotic behaviour of solutions
near coalescence points of particles. This approach, pioneered by Kato \cite{Kato57}, led already to a fairly
detailed picture of asymptotic properties of eigenfunctions of many-particle Hamiltonians, 
mainly due to the work of M. and T. Hoffmann-Ostenhof and coworkers \cite{HOS81,HO292,HO2S94,FHO2S05,FHO2S09}. 
Their most recent result concerning two-particle singularities has important consequences concerning the present work and will
be discussed below. 

In our research we pursue the strategy to apply a general operator calculus from singular analysis \cite{ES97,Schulze98,HS08} 
in order to get a detailed picture concerning the asymptotic behaviour near coalescence points of particles for
both the original Schr\"odinger equation and other many-particle models commonly used in quantum chemistry 
and solid state physics. 
The basic idea is to construct an asymptotic parametrix for a Hamiltonian which encodes
all the required asymptotic information. Details concerning the general concept of an asymptotic parametrix has been presented elsewhere \cite{FHS16},
and a first application to the Hamiltonian of the hydrogen atom was given in \cite{FHSS10}. In the present work,
we want to establish an explicit construction of an asymptotic local parametrix for the Hamiltonian of two-electron systems,
in particular the helium series and hydrogen molecule, near coalescence points of two particles, i.e., two electrons or an electron and a nucleus. 
Despite their simplicity, two-electron systems represent
an important benchmark problem in quantum chemistry. Furthermore,
two-electron subsystems represent the dominant contribution
in many-particle models, like coupled cluster theory.
Let us just mention that the results of the present work can be applied to these models
with minor modifications. For further details and first applications we refer to \cite{FHS15}. There are
other approaches in singular analysis which  have been applied to electronic structure theory as well, see e.g.~the work of Mazzeo, Nistor and collaborators, cf. \cite{ACM14}, \cite{ACN10}, \cite{HNS08}.

In the following we consider the stationary nonrelativistic Schr\"odinger equation within the Born-Oppenheimer approximation, i.e.,
\begin{equation}
 \biggl( - \tfrac{1}{2} \bigl( \Delta_1 + \Delta_2 \bigr) - \frac{Z}{|\boldsymbol{x}_1|} - \frac{Z}{|\boldsymbol{x}_2|}
+ \frac{1}{|\boldsymbol{x}_1 - \boldsymbol{x}_2|} \biggr) \Psi (\boldsymbol{x}_1,\boldsymbol{x}_2 ) = E \Psi ( \boldsymbol{x}_1, \boldsymbol{x}_2 )
\label{Schroedinger}
\end{equation}
for the helium atom ($Z=2$) and isoelectronic negatively ($Z=1$) or positively ($Z>2$) charged ions.
Possible modifications of the formalism outlined below for other two electron systems, i.e., the hydrogen molecule, are straightforward and will be 
discussed in the text where appropriate.
Depending on the total spin of the electrons, eigenfunctions $\Psi ( \boldsymbol{x}_1, \boldsymbol{x}_2 ) \in H^2(\mathbb{R}^6)$ must be
symmetric (singlet state) or antisymmetric (triplet state). Very accurate approximate solutions of (\ref{Schroedinger}) 
have been reported in the literature \cite{Hyll29,Pek58,FP66,FHM84}. Moreover, Fock outlined a recursive approach \cite{Fock54} which 
has been conjectured to provide an exact solution. Presently a complete proof of his conjecture 
is still missing, see, however, the work of Morgan \cite{Morg86} and Leray \cite{Leray1,Leray2,Leray3}.
Higher order terms of the Fock expansion have been studied in the literature, cf.~\cite{AM87,GAM87,GM87}. 
Furthermore, there exist suggestions how to generalize the Fock expansion to systems with more than two electrons,
cf.~\cite{Demkov, Erm1, Erm2}. There is, however, an essential difference between Fock's expansion and the present work.
We do not attempt to provide a complete solution for the helium
atom, instead our approach focuses on the asymptotic behaviour
of the solution near coalescence points of particles. Within the
present work no ad hoc assumptions concerning the form of asymptotic expansions are involved. Instead we extract these
properties in the course of the asymptotic parametrix construction.
Nevertheless it is possible to compare our results with corresponding terms in the Fock expansion, cf.~\cite{AM87,GAM87,GM87}.
For such comparison it seems, however, to be desirable to incorporate the corner type singularity at the coalescence point of all three particles in the asymptotic analysis. Therefore, we leave this topic for our future work.

\subsection{Differential operators and function spaces on stratified manifolds}
\label{oafsosmwes}
In the following we want to consider a singular operator calculus applied 
on the configuration space of particles interacting via Coulomb potentials.
Owing to the singularities of Coulomb potentials at coalescence points of particles
it is convenient to regard the configuration space as a stratified manifold, where
strata are classified according to the number of merging particles.

Let us consider a Coulomb system consisting of $N$ electrons and $K$ nuclei in the
Born-Oppenheimer approximation, where nuclei are kept fixed and the configuration space
restricts to electronic degrees of freedom. First of all, the physical configuration space ${\cal M}$ of $N$ 
electrons\footnote{We do not consider spin degrees of freedom 
or equivalent permutational symmetries of the electron coordinates in our discussion.} can be identified with $\mathbb{R}^{3N}$.
We define the subset ${\cal M}_0 \subset {\cal M}$ of all possible coalescence points of particles
including any number of electrons and nuclei. With it, ${\cal M} \setminus {\cal M}_0$ can be considered
as an open smooth manifold\footnote{The manifold ${\cal M} \setminus {\cal M}_0$ actually correponds to the mathematical notion of a 
configuration space of $N$ ordered particles in $\mathbb{R}^3$.} 
or, more generally, as the inner part of an open smooth manifold with boundary. 
Next, let us consider the subset ${\cal M}_1 \subset {\cal M}_0$ of all coalescence points of
more than two particles. The stratum ${\cal M}_0 \setminus {\cal M}_1$ is an open smooth manifold representing edges of
${\cal M}$. Correspondingly, we denote ${\cal M} \setminus {\cal M}_1$ as a singular manifold with edges.
Higher order strata can be constructed along the same lines, e.g., let  ${\cal M}_2 \subset {\cal M}_1$
denote the set of coalescence points of more than three particles. Again the stratum ${\cal M}_1 \setminus {\cal M}_2$
is an open smooth manifold representing the lowest order type of corners in ${\cal M}$. Therefore, 
${\cal M} \setminus {\cal M}_2$ is a singular manifold with edges and corners.
In this way the configuration space can be decomposed into its strata, i.e.,
\[
 {\cal M} = {\cal M} \setminus {\cal M}_0 \cup {\cal M}_0 \setminus {\cal M}_1 \cup {\cal M}_1 \setminus {\cal M}_2 \cdots .
\]
The singular operator calculus associates classes of degenerate differential operators to the singular manifolds
${\cal M} \setminus {\cal M}_i$, $i=0,1,\ldots,$ and a corresponding hierarchy of symbols to the strata.

Within the present work we want to study edge singularities corresponding to coalescence points of
two particles, i.e., two electrons or an electron and a nucleus. Higher order corner sigularities are subject
of our future work.

Near an edge, ${\cal M} \setminus {\cal M}_1$ is identified with a wedge
\[
 W = X^{\Delta} \times Y \ \ \mbox{with} \ X^{\Delta} :=
 (\overline{\mathbb{R}}_+ \times X) / (\{0\} \times X) 
\]
with smooth base $X$, homeomorphic to $S^2$, the unit sphere\footnote{It is not our intention to keep the notation as general as possible, instead we fix certain quantities
iherent to the problem, like the basis of the cone and its dimension, from the very beginning.}, and edge $Y$.

The class $\mbox{Diff}^\mu_{deg}(\mathbb{W})$ of edge-degenerate Fuchs-type differential 
operators of order $\mu \in\mathbb{N}$
\begin{equation}
 A:=r^{-\mu}\sum_{j+|\alpha|\leq \mu} a_{j\alpha}(r,y)(-r \partial_r)^j (rD_y)^\alpha
\label{Amu}
\end{equation}
is defined on the associated open stretched wedge 
\[
 \mathbb{W} = X^{\wedge} \times Y \ \ \mbox{with} \ X^{\wedge} :=
 \mathbb{R}_+ \times X .
\]
Here $r\in\mathbb{R}_+$ denotes the distance to the open edge $Y\subset \mathbb{R}^q$ and $y$ is a $q$-dimensional variable, varying on $Y$ (for a Coulomb system consisting of $N$ electrons the dimension $q$  of $Y$ equals $3N-3$).
  The coefficients $a_{j\alpha}(r,y)$ take values in differential operators
of order $\mu-(j+|\alpha|)$ on the base $X$ of the cone and are smooth in the respective variables up to $r=0$.

On an open stretched wedge it is the distance variable $r\in\mathbb{R}_+$ which carries the asymptotic information.
In order to incorporate asymptotics into Sobolev spaces let us proceed in a recursive manner. 
Weighted Sobolev spaces ${\cal K}^{s,\gamma}(X^\wedge)$ on an open stretched cone with base $X$ are defined with respect to the corresponding polar coordinates $\tilde{x} \rightarrow (r,x)$ via
\[
 {\cal K}^{s,\gamma}(X^\wedge) := \omega {\cal H}^{s,\gamma}(X^\wedge) +(1-\omega) H^s(\mathbb{R}^{3}) ,
\]
for a cut-off function $\omega$, i.e., $\omega\in C_0^\infty(\overline{\mathbb{R}}_+)$ such that $\omega(r)=1$ near $r=0$.
Here ${\cal H}^{s,\gamma}(X^\wedge) = r^\gamma {\cal H}^{s,0}(X^\wedge)$, and ${\cal H}^{s,0}(X^\wedge)$
for $s \in \mathbb{N}_0$ is defined to be the set of all $u(r,x) \in r^{-1} L^2(\mathbb{R}_+ \times X)$
such that $(r \partial_r)^jDu \in r^{-1} L^2(\mathbb{R}_+ \times X)$ for all $D \in \mbox{Diff}^{s-j}(X)$,
$0 \leq j \leq s$. The definition for $s \in \mathbb{R}$ in general follows by duality and complex interpolation.
Weighted Sobolev spaces with asymptotics are subspaces of ${\cal K}^{s,\gamma}$ spaces which are defined as direct sums
\begin{equation}
 {\cal K}^{s,\gamma}_Q (X^\wedge) := {\cal E}^\gamma_Q (X^\wedge) + {\cal K}^{s,\gamma}_\Theta (X^\wedge) 
\label{E+K}
\end{equation}
of flattened weighted cone Sobolev spaces
\[
 {\cal K}^{s,\gamma}_\Theta (X^\wedge) := \bigcap_{\epsilon > 0} {\cal K}^{s,\gamma - \vartheta - \epsilon}
 (X^\wedge) 
\]
with $\Theta =(\vartheta,0]$, $-\infty \leq \vartheta < 0$, and
asymptotic spaces
\[
 {\cal E}^\gamma_Q (X^\wedge) := \biggl\{ \omega(r) \sum_j \sum_{k=0}^{m_j} c_{jk}(x) r^{-q_j} \ln^k r \biggr\} .
\]
The asymptotic space ${\cal E}^\gamma_Q (X^\wedge)$ is characterized by a sequence $q_j \in \mathbb{C}$
which is taken from a strip of the complex plane, i.e.,
\[
 q_j \in \left\{ z: \frac{3}{2}-\gamma + \vartheta < \Re z < \frac{3}{2}-\gamma \right\} ,
\]
where the width and location of this strip are determined by its {\em weight data} $(\gamma,\Theta)$
with $\Theta =(\vartheta,0]$ and $-\infty \leq \vartheta < 0$. Each substrip of finite width
contains only a finite number of $q_j$. Furthermore, the coefficients
$c_{jk}$ belong to finite dimensional subspaces $L_j \subset C^\infty(X)$.
The asymptotics of ${\cal E}^\gamma_Q(X^\wedge)$ is therefore completely
characterized by the {\em asymptotic type} $Q := \{(q_j,m_j,L_j)\}_{j \in \mathbb{Z}_+}$.
In the following, we employ the asymptotic subspaces
\[
 {\cal S}^\gamma_Q (X^\wedge) := \left\{ u \in {\cal K}^{\infty,\gamma}_Q (X^\wedge) :
 (1- \omega) u \in {\cal S}(\mathbb{R},C^\infty(X))|_{\mathbb{R}_+ \times X} \right\} 
\]
with Schwartz class behaviour for exit $r \rightarrow
\infty$. The spaces ${\cal K}^{s,\gamma}_Q(X^\wedge)$ and ${\cal
S}^\gamma_Q (X^\wedge)$ are Fr\'echet spaces equipped with
natural semi-norms according to the decomposition (\ref{E+K}); we
refer to \cite{ES97, Schulze98} for further details.

Weighted wedge Sobolev spaces on $\mathbb{W} := X^{\wedge} \times Y$ can be defined as
functions $Y \rightarrow {\cal K}_{(Q)}^{s, \gamma}(X^\wedge)$, where a subscript $(Q)$ optionally denotes cone spaces 
with and without asymptotics. Let us first consider the case $Y= \mathbb{R}^q$ and corresponding wedge Sobolev spaces
\[
 {\cal W}^s(\mathbb{R}^q, {\cal K}_{(Q)}^{s, \gamma}(X^\wedge)) :=
 \{ u : \mathbb{R}^q \rightarrow {\cal K}_{(Q)}^{s, \gamma}(X^\wedge)
 \, | \, u \in  \overline{{\cal S}(\mathbb{R}^q, {\cal K}_{(Q)}^{s, \gamma}(X^\wedge)} \} 
\]
with $s, \gamma \in \mathbb{R}$ and norm closure w.r.t.~the norm
\[
 \| u \|_{{\cal W}^{s}(\mathbb{R}^q, {\cal K}_{(Q)}^{s, \gamma}(X^\wedge))}^2 := \int [\eta]^{2s} \| \kappa^{-1}_{[\eta]} (F_{y\rightarrow \eta} u)(\eta)
 \|_{{\cal K}^{s,\gamma}_{(Q)}(X^\wedge)}^2 d\eta .
\]
Here $F_{y\rightarrow \eta}$ denotes the Fourier transform in $\mathbb{R}^q$ and $\{ \kappa_\lambda \}_{\lambda \in \mathbb{R}_+}$
a strongly continous group of isomorphisms $\kappa_\lambda : {\cal K}_{(Q)}^{s, \gamma}(X^\wedge) \rightarrow 
{\cal K}_{(Q)}^{s, \gamma}(X^\wedge)$ defined by
\[
 \kappa_\lambda u(r,x,y) := \lambda^{\frac{3}{2}} u(\lambda r,x,y) .
\]
The function $[\eta]$ involved in the norm is given by a strictly positive $C^{\infty}(\mathbb{R}^{3})$ function of the covariables $\eta$ such that 
$[\eta] = |\eta|$ for $|\eta| \geq \epsilon >0$. The motivation behind this group action is the twisted homogeneity 
of principal edge symbols discussed below. For $Y \subset \mathbb{R}^q$ an open subset, we define
\begin{equation}
 {\cal W}^s_{\mbox{\footnotesize comp}}(Y, {\cal K}_{(Q)}^{s, \gamma}(X^\wedge)) :=
 \{ u \in {\cal W}^s(\mathbb{R}^q, {\cal K}_{(Q)}^{s, \gamma}(X^\wedge)): \supp u \subset Y \ \mbox{compact} \} ,
\label{Wcomp}
\end{equation}
and
\begin{equation}
 {\cal W}^s_{\mbox{\footnotesize loc}}(Y, {\cal K}_{(Q)}^{s, \gamma}(X^\wedge)) :=
 \{ u \in {\cal D}'(Y, {\cal K}_{(Q)}^{s, \gamma}(X^\wedge)): \varphi u \in
 {\cal W}^s_{\mbox{\footnotesize comp}}(\mathbb{R}^q, {\cal K}_{(Q)}^{s, \gamma}(X^\wedge)) \ \mbox{for each} \ 
 \varphi \in C^\infty_0(Y) \} .
\label{Wloc}
\end{equation}
The class $\mbox{Diff}^\mu_{deg}(\mathbb{W})$ of edge-degenerate differential operators represents bounded operators
\begin{equation}\label{Aedge1}
 A: {\cal W}^s_{\mbox{\footnotesize comp(loc)}}(Y, {\cal K}_{(Q)}^{s, \gamma}(X^\wedge)) \longrightarrow
 {\cal W}^{s-\mu}_{\mbox{\footnotesize comp(loc)}} (Y, {\cal K}_{(P)}^{s-\mu, \gamma-\mu}(X^\wedge))
\end{equation}
between appropriate wedge Sobolev spaces. For a compact edge $Y$ both spaces (\ref{Wcomp}) and (\ref{Wloc}) 
become equivalent.

\subsection{Pseudo-differential operators on  manifolds with edge singularities}
\label{pseudoedge}
The basic idea of the present work is to construct an
``approximate inverse'', i.e., a parametrix, for edge-degenerate
differential operators (\ref{Amu}) and their corresponding function spaces. First of all, one has to extend the operator
space under consideration from differential to pseudo-differential operators. Furthermore, it requires an appropriate notion of
ellipticity, which implies existence of a parametrix for an elliptic operator. For the sake of a concise presentation,
we will state in the following only some basic notations
and refer to \cite{FHS16} as well as the monographs \cite{Schulze98} for further details. 
 
A pseudo-differential operator $P$, in the edge-degenerate operator class $L^{\mu}(Y \times \mathbb{R}^3, {\boldsymbol g})$, for weight data $\boldsymbol{g}=(\gamma ,\gamma -\mu ,\Theta ),\,\gamma ,\mu \in \mathbb{R}$, is given in the general form
\begin{equation}
 P = \sigma' \Op_y (p) \tilde{\sigma}'
 +(1-\sigma')P_{\interior}(1-\hat{\sigma}') 
\label{pseudiff}
\end{equation}
with cut-off functions $\hat{\sigma}' \prec \sigma' \prec \tilde{\sigma}'$, i.e., $\hat{\sigma}' \sigma'=\hat{\sigma}'$ and $\sigma'\tilde{\sigma}'=\sigma'$. 
The edge amplitude function $p$ belongs to the space of operator functions $R^{\mu}(Y \times \mathbb{R}^3, {\boldsymbol g})$
which in turn are of the form
\begin{equation}
 p(y,\eta) = \omega_{1,\eta} r^{\mu} \opm_{\M}^{\gamma-1}(p_M)(y,\eta)    \omega_{0,\eta} 
 + (1-\omega_{1,\eta}) r^{\mu} p_\psi(y,\eta) (1-\omega_{2,\eta}) 
 +m(y,\eta) +g(y,\eta) ,
\label{pinRmu}
\end{equation}
cf.~\cite[Section 2.1]{FHS16},
with cut-off functions $\omega_2, \omega_1, \omega_0$ satisfying $\omega_2 \prec \omega_1 \prec \omega_0,$ 
where we write here and in the following $\omega_\eta(r) := \omega ([\eta]r)$. The first and second term of the operator function (\ref{pinRmu}) correspond to a Mellin and an inner Fourier pseudo-differential operator, respectively. 
Let us just mention, that the operator-valued Mellin symbol can be defined according to
\begin{equation}
 p_M(r,y,w,\eta) := \tilde{p}_M(r,y,w,r\eta)
\label{pptilde}
\end{equation}
with respect to a symbol $\tilde{p}_M(r,y,w,\tilde{\eta})\in C^\infty (\overline{\mathbb{R}}_+\times \Omega ,M^\mu _\mathcal{O}(X,\mathbb{R}^3_{\tilde{\eta}}))$, see e.g.~\cite{FHS16} for
further details. Here, $M_{{\cal O}}^\mu (X,\mathbb{R}^3)$ is the set of all holomorphic Mellin symbols with values in $L_{{\mathrm cl}}^\mu(X,\mathbb{R}^3),$ the space of classical parameter-dependent pseudo differential operators on the base $X$. 
Moreover, $M_Q^{-\infty}(X)$, for an asymptotic type $Q$,  is the set of all meromorphic Mellin Symbols with values in $L^{-\infty}(X)$, cf. ~\cite[Section 1.1]{FHS16}.

Finally, the last two terms of (\ref{pinRmu}) represent Green and smoothing Mellin operators,
which belong to the subspaces $R^{\mu}_{\G}(Y \times \mathbb{R}^3, {\boldsymbol g})$ and $R^{\mu}_{\M+\G}(Y \times \mathbb{R}^3, {\boldsymbol g})$ of operator functions, respectively,
cf.~\cite[Section 2.1]{FHS16} for further details. The smoothing Mellin operators in $R^\mu_{M+G}(Y \times \mathbb{R}^3, {\boldsymbol g})$ have symbols in $M_{Q}^{-\infty}(X),$ 
which corresponds to the space of meromorphic functions with values in 
$L^{-\infty}(X)$, cf.~\cite[Section 1.1]{FHS16}. For further reference let us also define the extended symbol space 
\[
 M_Q^\mu(X):=M_{{\cal O}}^\mu(X)+M_Q^{-\infty}(X).
 \]

Within the present work we want to construct asymptotic parametrices in local neighbourhoods of the 
edges separated from the corner singularity at the tip of the cone. Such parametrices exist provided that the Hamiltonian 
corresponds to an elliptic operator in the edge degenerate sense. 
Ellipticity of a pseudo-differential operator on a manifold with edge singularity is characterized by a pair of symbols
$\big(\sigma_0, \sigma_1 \big)$,
where $\sigma_0$ is the usual homogeneous principal symbol of a classical pseudo-differential operator in the interior and $\sigma_1$ denotes 
the so-called \emph{principal edge} symbol, which is twisted 
homogeneous in the sense
\[
\sigma _1(P)(y,\lambda  \eta )= \lambda  ^\mu \kappa _\lambda  \sigma _1(P)(y,\eta ) \kappa _\lambda  ^{-1}
\]
for all $\lambda  \in \mathbb{R}_+.$
The principal edge symbol is defined as
$$\sigma _1(P)(y,\eta):=r^{-\mu }\textup{op}_{\M}^{\gamma -1}(h_M)(y,\eta )+\sigma _1(m+g)(y,\eta),$$
for $h_M(r,y,w,\eta ):=\tilde{p}_M(0,y,w,r\eta )$, cf.~(\ref{pptilde}), and corresponds to a parameter-dependent operator family in the cone algebra. 
The smoothing Mellin and Green part of the symbol is given by
$$\sigma _1(m+g)(y,\eta):=r^{-\mu}\omega_{|\eta |} \sum_{j=0}^k r^j\sum_{|\alpha |= j} \opm_{\M}^{\gamma_{j\alpha}-n/2}(f_{j\alpha})(y)\eta^\alpha \omega'_{|\eta |} +g_{(\mu )}(y,\eta ),$$
with $\eta \neq 0,\,\omega_{|\eta |}(r):=\omega (r|\eta |),$ where $g_{(\mu )}(y,\eta )$ is the twisted homogeneous principal symbol of $g$ as a classical operator-valued symbol of order $\mu $. 
Concerning the Mellin part, in particular, the meaning of the
Mellin symbols $f_{j\alpha}$, we refer to \cite[Eq.~2.32]{FHS16}.

In the particular case of a differential operator (\ref{Amu}), 
the principal edge symbol is given by
\[
\sigma_1(A)(y,\eta):=r^{-\mu}\sum_{j+|\alpha|\leq \mu} a_{j\alpha}(0,y)(-r\partial_r)^j (r\eta)^\alpha .
\]
Ellipticity requires that $\sigma_1$ represents Fredholm operators between weighted Sobolev spaces 
\begin{equation}
\sigma_1(P)(y,\eta):\mathcal{K}^{s,\gamma}(X^\wedge)\rightarrow \mathcal{K}^{s-\mu,\gamma-\mu}(X^\wedge) .
\label{swmap}
\end{equation}
In the general ellipticity theory developed in  \cite{HS08}  or \cite{Schulze98}, 
the mapping (\ref{swmap})  extends to an isomorphism. From a physical point of view, such an extension
looks artificial. Therefore, it seems natural to expect values of $\gamma$ for which $\sigma_1$ actually
represents isomorphisms between the corresponding Kegel spaces.
That such values of $\gamma$ actually exist is essential for our particular application and will be discussed in the following section. For a detailed discussion of the concept of ellipticity
and corresponding symbolic structures, 
we refer to the monographs \cite{HS08} and \cite{Schulze98}.

In order to define asymptotic parametrices for elliptic operators,
we have to go one step further in the symbolic hierarchy and 
consider {\em conormal symbols} $\sigma _\textup{c}^{\mu -j}(p)(y,w,\eta)$, $j=0,\dots,k$, for parameter dependent operator functions $p(y,\eta)$ in $R^{\mu}(Y \times \mathbb{R}^3, {\boldsymbol g})$
of asymptotic type $\boldsymbol{g}:=(\gamma ,\gamma -\mu ,\Theta ),\,\Theta =(-(k+1),0]$.
These conormal symbols are polynomials in $\eta $ of order $j$ with coefficients in $C^\infty (Y ,M_R^\mu (X))$ for certain Mellin asymptotic types $R$ and are explicitly given by
\[
 \sigma _\textup{c}^{\mu -j}(p)(y,w,\eta) :=
 \frac{1}{j!} \bigl( \partial^j_r p_M \bigr) (0,y,w,\eta)
 + \sum_{|\alpha |= j} f_{j\alpha}(y)\eta^\alpha .
\]
Conormal symbols can be used to define
a filtration of operator function spaces $R^{\mu}(Y \times \mathbb{R}^3, {\boldsymbol g})$. Let us set 
\[
 R^{\mu;0}(Y \times \mathbb{R}^3, {\boldsymbol g}) :=
 R^{\mu}(Y \times \mathbb{R}^3, {\boldsymbol g})
\]
and define, for $1 \leq j \leq k+1$, the sequence
\begin{equation}
 R^{\mu;j} (Y \times \mathbb{R}^3,\boldsymbol{g}):=\{p\in R^{\mu;j-1}(Y \times \mathbb{R}^3,\boldsymbol{g}):\sigma _\textup{c}^{\mu -(j-1)}(p)=0\} .
\label{filtration}
\end{equation}

Furthermore, let us introduce the filtration $R^{\mu;j}_G(\Omega\times\mathbb{R}^q,{\boldsymbol g})$, which is defined in an analogous manner on the subspace $R^\mu_G(Y \times\mathbb{R}^3,{\boldsymbol g}).$
With this filtration at hand, we can define the corresponding 
filtration for classes of edge-degenerate pseudo-differential operators $L^{\mu} (M,\boldsymbol{g})$ by taking 
$L^{\mu;j} (M,\boldsymbol{g})$ for $\boldsymbol{g}:=(\gamma ,\gamma -\mu ,\Theta ),\,\Theta =(-(k+1),0],$ as the set of all $P\in L^\mu  (M,\boldsymbol{g})$ such that the local edge amplitude functions $p(y,\eta )$ occurring 
in \eqref{pseudiff} belong to $R^{\mu;j} (\Omega \times \mathbb{R}^q,\boldsymbol{g}).$
This filtration provides the basis for the definition
of asymptotic expansions of pseudo-differential operators discussed  
below. 

\subsection{The Hamiltonian of the helium atom in its edge-degenerate form}
\label{thothaiiedf}
In the following, we consider  the configuration space $\mathbb{R}^6$ of two electrons
of the helium atom as a stratified manifold with embedded edge and corner singularities.
For this purpose let us introduce polar coordinates in $\mathbb{R}^6$ with radial variable
\begin{equation}
 t:= \sqrt{x_1^2 + x_2^2 + x_3^2 + x_4^2 + x_5^2 + x_6^2} .
\label{t-def}
\end{equation}
In such coordinates $\mathbb{R}^6$ can be formally considered as a conical manifold
with base $S^5$ and embedded conical singularity at the origin, i.e.,
\[
 \mathbb{R}^6 \equiv (S^5)^\Delta:=(\overline{\mathbb{R}}_+ \times S^5) / (\{0\} \times S^5) .
\]
Here, the origin represents the coalescence point of both electrons with the nucleus and
all Coulomb potentials in the Hamiltonian become singular.
Removal of this singular point defines an open stretched cone
\[
 (S^5)^{\wedge} := \mathbb{R}_+ \times S^5 .
\]
Beside the origin, singularities of the Coulomb potentials appear at 
\[
|\boldsymbol{x}_1|=0,\quad |\boldsymbol{x}_2|=0, \quad |\boldsymbol{x}_1-\boldsymbol{x}_2|=0;
\]
each subset of configuration space corresponds to a closed embedded submanifold $Y_i$, $i=1,2,3$, on the base $S^5$.
These disjoint submanifolds are homeomorphic to $S^2$ and for any of them there exists a local neighbourhood 
$U_i$, $i=1,2,3$, on $S^5$ which is homeomorphic to a wedge
\[
 W_i = X^{\Delta}_i \times Y_i \ \ \mbox{with} \ X^{\Delta}_i :=
 (\overline{\mathbb{R}}_+ \times X_i) / (\{0\} \times X_i),
\]
where the base $X_i$ of the wedge is again homeomorphic to $S^2$. 
The associated open stretched wedges are 
\[
 \mathbb{W}_i = X^{\wedge}_i \times Y_i \ \ \mbox{with} \ X^{\wedge}_i :=
\mathbb{R}_+ \times X_i, i=1,2,3.
\]
Globally, the conical manifold $(S^5)^\Delta$ with local wedges on its basis $S^5$
becomes a manifold with corner singularity at the origin.

After these general considerations it remains to explicitly specify local coordinates
in neigbourhoods of the edges which provide representations of the Hamiltonian in the 
class of edge-degenerate differential operators.
Fortunately, appropriate hyperspherical coordinates were already known in the physics literature, 
cf.~the monographs \cite[pp.~1730ff]{MF53} and \cite[pp.~398ff]{MM65}. Furthermore, Granzow's paper \cite{G63} 
provides some useful facts related to these coordinates in spectral and differential geometry. 

In a local neighbourhood $U_1$ of the $e-n$ singularity
$|\boldsymbol{x}_1|=0$
let us define hyperspherical coordinates
\begin{equation}
 \tilde{x}_1 = \sin\,r_1 \sin \theta_1 \cos \phi_1, \ \ \tilde{x}_2 = \sin\,r_1 \sin \theta_1
\sin \phi_1, \ \
 \tilde{x}_3 = \sin\,r_1 \cos \theta_1 ,
\label{polarX1}
\end{equation}
\begin{equation}
 \tilde{x}_4 = \cos\,r_1\sin \theta_2 \cos \phi_2, \ \ \tilde{x}_5 = \cos\,r_1\sin \theta_2
\sin \phi_2, \ \
 \tilde{x}_6 = \cos\,r_1 \cos \theta_2 
\label{polarY1}
\end{equation}
with respect to the projective coordinates $\tilde{x}_i := x_i/t$, $i=1,\ldots,6$,
on $S^5$.
One can consider (\ref{polarX1}) as polar coordinates on the stretched cone
$X^\wedge_1$
with $r_1\in (0,\frac{\pi}{2}]$, $ \theta_1 \in(0,\pi)$, $\phi_1 \in[0, 2\pi)$. The remaining two angular
variables in (\ref{polarY1})
with $ \theta_2 \in (0,\pi)$, $ \phi_2 \in [0, 2\pi)$ provide a spherical coordinate system on $Y_1$.
Similar local coordinates can be constructed in a local neighbourhood $U_2$ of the $e-n$ singularity $|\boldsymbol{x}_2|=0$.

The Hamiltonian of the helium atom is represented in these coordinates by an edge-degenerate differential
operator, e.g., in a tubular neighbourhood $U_1$ by
\begin{eqnarray}
 H_{\edge} := H|_{\tub U_1} & = & r^{-2}_1 \Big[ -\frac{1}{2 t^2} (- r_1
{\partial_{r_1}} )^2
 - \frac{h}{2t^2} (- r_1 {\partial_{r_1}})
 -\frac{1}{2} (r_1 {\partial_t} )^2
 - \frac{5 r_1}{2t} (r_1 {\partial_t} )  \nonumber\\
 & & - \frac{1}{2t^2 \cos^2 r_1} (r_1 {\partial_{\theta_2}})^2
 - \frac{r_1 \mbox{ctan} \, \theta_2}{2t^2 \cos^2 r_1} (r_1 {\partial_
{\theta_2}} )
- \frac{1}{2t^2 \sin^2 \theta_2 \cos^2 r_1} (r_1 {\partial_
{\phi_2}} )^2 \nonumber\\
 & & - \frac{r_1^2}{2t^2 \sin^2 r_1} \Delta_{X_1} + \frac{r_1}{t} v_{e-n} \Big] \label{He1n.edge}
\end{eqnarray}
with
\[
h := 1+2r_1\mbox{tan} \, r_1-2r_1\mbox{ctan}\, r_1 
\]
and
\[
 v_{e-n} := -\frac{Z r_1}{\sin r_1} -\frac{Z r_1}{\cos r_1} +
 \frac{ r_{1}}{\sqrt{1-\sin(2r_{1})
 [ \cos \theta_1 \cos \theta_2 + \sin \theta_1 \sin \theta_2 \cos(\phi_1-\phi_2)]}}.
\]
For the singular calculus it is important to note that $h$ and $v_{e-n}$ are smooth with respect to $r_1$ up to $r_1=0$.

The arguments are analogous in the case of $|\boldsymbol{x}_2|=0.$

In the remaining case of the $e-e$ edge singularity it is convenient to define center of mass coordinates
\begin{gather*}
 z_1 = \frac{1}{\sqrt{2}} (x_1 - x_4), \ \ z_2 = \frac{1}{\sqrt{2}} (x_2 - x_5), \ \
 z_3 = \frac{1}{\sqrt{2}} (x_3 - x_6), \\
 z_4 = \frac{1}{\sqrt{2}} (x_1 + x_4), \ \ z_5 = \frac{1}{\sqrt{2}} (x_2 + x_5), \ \
 z_6 = \frac{1}{\sqrt{2}} (x_3 + x_6)
 \end{gather*}
in a neigbourhood $U_3$ with
\[
 \sum_{i=1}^6 z_i^2 = \sum_{i=1}^6 x_i^2.
\]

The hyperspherical coordinates can now be introduced in a completely analogous manner.
In such a neigbourhood of the $(e-e)$ edge, the Coulomb potentials can be written in the form
\[
 -\frac{Z}{|{\bf x}_1|} -\frac{Z}{|{\bf x}_2|} + \frac{1}{|{\bf x}_1 - {\bf x}_2|} =
 \frac{1}{t r_{12}} v_{e-e}(r_{12},\theta_1,\phi_1,\theta_2,\phi_2)
\]
with
\begin{eqnarray*}
 v_{e-e} & = & - \frac{\sqrt{2} Z r_{12}}{\sqrt{1+\sin(2r_{12})
 [ \cos \theta_1 \cos \theta_2 + \sin \theta_1 \sin \theta_2 \cos(\phi_1-\phi_2)]}} \\
 & & - \frac{\sqrt{2} Z r_{12}}{\sqrt{1-\sin(2r_{12})
 [ \cos \theta_1 \cos \theta_2 + \sin \theta_1 \sin \theta_2 \cos(\phi_1-\phi_2)]}} \\
 & & + \frac{1}{\sqrt{2}j_0(r_{12})},
\end{eqnarray*}
where again $v_{e-e}$ is smooth with respect to $r_{12}$ up to $r_{12}=0$.

\begin{remark}\label{Hamellipt}
The following two properties of the Hamiltonian are essential for our applications.
\begin{itemize}
\item[{\em (i)}]
The Hamiltonian of the helium atom {\em (\ref{He1n.edge})} corresponds to a bounded edge-degenerate operator
\begin{equation}\label{Hedgebo}
H_{\edge}: {\cal W}^s_{\mathrm{comp(loc)}}(Y, {\cal K}_{(Q)}^{s, \gamma}(X^\wedge)) \longrightarrow
 {\cal W}^{s-2}_{\mathrm{comp(loc)}} (Y, {\cal K}_{(P)}^{s-2, \gamma-2}(X^\wedge))
\end{equation}
for any $s,\gamma\in\mathbb{R}.$
\item[{\em (ii)}]
It has been shown in Ref.~{\em \cite{FH10}} that $H_{\mathrm{edge}}$ satisfies the ellipticity conditions of the edge-degenerate calculus for 
\[
\tfrac{1}{2}<\gamma<\tfrac{3}{2}.
\]
Therefore, there exists a parametrix of $H_{\mathrm{edge}}$ in the corresponding pseudo-differential operator algebra, cf. {\em \cite{ES97}, \cite{HS08}, \cite{Schulze98}}.
\end{itemize}
\end{remark}

The second part of the previous remark is essentially based
on the following lemma, cf.~\cite{FH10} for a proof, which specifies the weights, where the principal edge symbols actually represent isomorphisms between the corresponding Kegel spaces.

\begin{lemma}\label{lem1}
The principal edge symbols $\sigma_1(H_{\edge}|_{U_i})$, $i=1,2,3$, represent Fredholm operators
\[
\sigma_1(H_{\edge}|_{U_1})(t,(\theta_2,\phi_2),\tau, (\Theta_2,\Phi_2)):\mathcal{K}^{s,\gamma}((S^2)^\wedge)\rightarrow \mathcal{K}^{s-2,\gamma-2}((S^2)^\wedge)
\]
for $\gamma\notin \mathbb{Z}+\frac{1}{2},$ for all $(t,(\theta_2,\phi_2))\in\Omega\times S^2$ and all $(\tau, (\Theta_2,\Phi_2))\not=0$.
Furthermore, they correspond to isomorphisms for any $\gamma$ with $\frac{1}{2}<\gamma<\frac{3}{2}.$
\end{lemma}

It is worth mentioning that the interval $\gamma$ with $\frac{1}{2}<\gamma<\frac{3}{2}$ actually corresponds to
the physically relevant weights for eigenfunctions of the Hamiltonian; we refer to \cite{FH10} for further discussions.

\section{Asymptotic parametrices for Hamiltonian operators}
\label{coap}

\subsection{Outline of our approach}
\label{oooa}

The underlying theory for the construction of asymptotic parametrices in the edge algebra has been presented in a previous paper \cite{FHS16} of the authors. For the convenience of the reader we give a brief outline of some basic ideas.
\begin{itemize}
\item[(i)]
{\bf Existence of an asymptotic parametrix:} \\
Every pseudo-differential operator in $L^{\mu}$ can be written in the form
\begin{equation}\label{Pap}
P=\sum_{i=0}^N r^i P_i
\end{equation}
modulo Green operators for any $N\in\mathbb{N}$ and for suitable $P_i\in L^{\mu}, i=0,1,\ldots,N.$ The individual summands $r^i P_i$  
belong to $L^{\mu;i }(M,\boldsymbol{g})\setminus L^{\mu;i+1 }(M,\boldsymbol{g}),\,0\leq i\leq N-1,$ and $r^N P_N\in L^{\mu;N}(M,\boldsymbol{g})$.
In particular, $r^i P_i$ is flat of order $i$, i.e., in the edge algebra
\[
r^i P_i:{\cal W}^s_{\mathrm{comp}}(Y,{\cal K}^{s,\gamma}_{(Q)}(X^\wedge))\rightarrow {\cal W}^{s-\mu}_{\mathrm{loc}}(Y,{\cal K}^{s-\mu,\gamma-\mu+i}_{(P)}(X^\wedge)).
\]
There exists an $N\in\mathbb{N}$ such that for $i\geq N$ the operator $r^i P_i$ maps on the flat part ${\cal K}_{\Theta}$ of ${\cal K}_P,$ cf. (\ref{E+K}). Therefore, for our purposes it is sufficient to consider finite expansions for the Hamiltonian and of the corresponding parametrix.

\item[(ii)]
{\bf Recursive construction of asymptotic parametrices:} \\
Let us start with a differential operator  $A\in\mbox{Diff}^\mu_{deg}(\mathbb{W})$ in the form (\ref{Pap}) and represent the corresponding parametrix in the same form. Then we have 
\[
\left(\sum_{i=0}^N r^i P_i\right)\left(\sum_{i=0}^N r^i A_i\right)\sim I,\quad 
\left(\sum_{i=0}^N r^i A_i\right)\left(\sum_{i=0}^N r^i P_i\right)\sim I
\]
modulo Green operators and flat remainders $F$, i.e., $r^{-(N+1)}F\in L^0({\cal M}).$
From this ansatz we can derive the following recursion scheme for 
the construction of an asymptotic parametrix.
\begin{itemize}
\item Initial step: $P_0 A_0 = I \quad \mbox{mod} \ L^0_G + L^0_{\gflat}$.
\item First recursion:  
$r^1 P_1 A_0 + P_0 r^1 A_1 = 0 \quad \mbox{mod} \ L^0_G + L^0_{\gflat}$.\\
In general an operator ${\cal O}\in L^\mu({\cal M})$ satisfies a commutator relation
\[
{\cal O}r^\beta-r^\beta {\cal O}_\beta=G_\beta
\]
with an ${\cal O}_\beta\in L^{\mu}({\cal M})$ and an Green operator $G_\beta.$ Then we have
\[
rP_1 \stackrel{\mod G}{=} -P_0 r A_1 P_0 \stackrel{\mod G}{=} -rP_{0,1} A_1 P_0.
\]
\item Second recursion: 
$r^2 P_2 A_0 + r P_1 r A_1 + P_0 r^2 A_2 = 0 \quad \mbox{mod} \ L^0_G + L^0_{\gflat}$.\\
Again one has to commute powers of $r$ to the left in order to
get the recursion relation
\begin{eqnarray*}
r^2 P_2& \stackrel{\mod G}{=} & -r P_1 r A_1 P_0-P_0r^2A_2 P_0\\ 
& \stackrel{\mod G}{=} &r^2\big(P_{0,1,1}A_{1,1}P_{0,1}A_1 P_0-P_{0,2}A_2 P_0\big).
\end{eqnarray*}
\item Higher recursions can be performed in a similar manner.
\end{itemize}
\item[(iii)]
{\bf Calculation of Green operators and the asymptotic behaviour of eigenfunctions:} \\
The parametrix construction of the previous step has been 
given modulo certain Green operators which carry asymptotic
information.  
In order to get the asymptotic behaviour of
eigenfunctions of the Hamiltonian operator, it is essential
to keep track of all of these Green operators
and to calculate them order by order in the asymptotic expansion.
A large part of the present work is actually devoted to the
details of their calculation. 

Let us define the shifted Hamiltonian operator
\[
 A_{\edge} := H_{\edge} -E 
\]
with $E$ an eigenvalue of the Hamiltonian and consider the equation $A_{\edge}u=0$. Now let us apply
a parametrix from the left, i.e.,
\begin{equation}
 P A_{\edge} u = u + Gu =0 ,
\label{PAedge}
\end{equation}
which means that $u$ is an eigenfunction of the 
Green operator $G$, too.
From our calculations, we obtain
an asymptotic expansion of the operator valued symbol
of the Green operator
\[
 g(y,\eta) \sim g_0(y,\eta) + r g_1(y,\eta) + r^{2} g(y,\eta) + \cdots
\]
with $r^{i} g_i(y,\eta) \in R_G^{0;i} (Y \times \mathbb{R}^3,\boldsymbol{g}) \setminus R_G^{0;i+1} (Y \times \mathbb{R}^3,\boldsymbol{g})$, $i=0,1,2,\ldots$,
and mapping property
\[
 r^{i} g_i(y,\eta): {\cal K}^{s,\gamma}(X^\wedge)\rightarrow
 {\cal S}^{\gamma+i}_{Q_i}(X^\wedge),
\]
for specific asymptotic types $Q_i$. Therefore, according to
equation (\ref{PAedge}), the asymptotic behaviour of
an eigenfunction $u$ is completely determined by the 
asymptotic types $Q_i$ of the asymptotic expansion of the
Green operator $G$.
\end{itemize}

\subsection{Asymptotic representation of the Hamiltonian}
\label{2.1}
Let us perform an asymptotic expansion of shifted Hamiltonian operator $A_{\edge}$ in the cone variable
\[
 A_{\edge} = \sum_{i=0}^N r^i A_i + A_N .
\]
In order to calculate the asymptotic terms from (\ref{He1n.edge}) some auxiliary Taylor series are required, i.e.,
\[
 h(r) = -1 +\tfrac{8}{3} r^2 +\tfrac{32}{45} r^4 +O(r^6) ,
\]
\[
 \frac{r^2}{\sin^2r} = 1 +\tfrac{1}{3} r^2 +\tfrac{1}{15} r^4 +O(r^6), \ \ 
 \frac{1}{\cos^2r} = 1 + r^2 +\tfrac{2}{3} r^4 +O(r^6) .
\]
Depending on the specific edge under consideration, the corresponding Taylor series of the Coulomb potentials 
are given by
\begin{eqnarray*}
 v_{e-n} & = & -Z + \bigl( 1- Z \bigr) r + \bigl( a - \tfrac{1}{6} Z \bigr) r^2 
 + \tfrac{1}{2} \bigl( 3 a^2 - Z \bigr) r^3 
 + \tfrac{1}{360} \bigl( - 240 a +900 a^3 -7Z \bigr) r^4 \\
 & & + \bigl( -2a^2 + \tfrac{35}{8} a^4 - \tfrac{5}{24} Z \bigr) r^5 + O(r^6) 
\end{eqnarray*}
and
\[
 v_{e-e} = \tfrac{1}{\sqrt{2}} -2 \sqrt{2} Z r + \tfrac{1}{6\sqrt{2}} r^2 - 3\sqrt{2} a^2 Z r^3
 + \tfrac{7}{360\sqrt{2}} r^4 + \bigl( 4 \sqrt{2} a^2 Z - \tfrac{35}{2 \sqrt{2}} a^4 Z \bigr) r^5 + O(r^6) ,
\]
with
\[
 a:= \cos \theta_1 \cos \theta_2 + \sin \theta_1 \sin \theta_2 \cos(\phi_1-\phi_2) .
\]

With this the asymptotic terms of the Hamiltonian become
\begin{eqnarray}
 A_0 & = & r^{-2} \Big[ -\frac{1}{2 t^2} (- r\partial_{r} )^2
 + \frac{1}{2t^2} (- r {\partial_{r}})
 -\frac{1}{2} (r {\partial_t} )^2
 - \frac{1}{2t^2 \sin^2 \theta_2} (r {\partial_{\phi_2}} )^2 
 - \frac{1}{2t^2} (r {\partial_{\theta_2}})^2
 - \frac{1}{2t^2} \Delta_{X} \Big];\nonumber\\
 A_1 & = & r^{-2} \biggl[ - \frac{5}{2t} (r \partial_t)
 - \frac{\mbox{ctan} \theta_2}{2t^2} (r \partial_{\theta_2}) + \frac{1}{t} \left\{ \begin{array}{ll}
 \tfrac{1}{\sqrt{2}} & \mbox{for} \ v_{e-e} \\
 -Z & \mbox{for} \ v_{e-n} \end{array} \right. \biggr];\nonumber\\
 A_2 & = & r^{-2} \biggl[ - \frac{4}{3t^2} (-r \partial_r) - \frac{1}{2t^2} (r \partial_{\theta_2})^2
 - \frac{1}{2t^2 \sin^2 \theta_2} (r \partial_{\phi_2})^2 - \frac{1}{6t^2} \Delta_{S^2} -E
 + \frac{1}{t} \left\{ \begin{array}{ll}
 -2 \sqrt{2} Z & \mbox{for} \ v_{e-e} \\
 1-Z & \mbox{for} \ v_{e-n} \end{array} \right. \biggr]\nonumber\\
  & \vdots & \nonumber
\end{eqnarray}

\subsection{The initial asymptotic parametrix and corresponding Green operators}
\label{2.2}
The edge-degenerate differential operator $A_0$ can be expressed as a Leibniz-Mellin pseudo-differential operator
close to the edge.
Let $\sigma,\tilde{\sigma}, \hat{\sigma}$ be cut-off functions\footnote{ 
Throughout this paper a cut-off function is any real-valued $\sigma
\in C^\infty_0(\overline{\mathbb{R}}_+)$ such that $\sigma(r) \equiv
1$ close to $r=0$.}
such that $\sigma \prec \tilde{\sigma}$, $\hat{\sigma} \prec \sigma$.
It is easy to see that $A_0$ can be decomposed into two parts in the
following way
\[
 A_0=\sigma A_0 \tilde{\sigma}+(1-\sigma)A_0(1-\hat{\sigma}).
\]
Next, we introduce a locally finite cover $\{U_i\}_{i=0,1 \ldots, \infty}$ on the edge 
$Y= \mathbb{R}_+ \times S^2 \cong \mathbb{R}^3 \setminus \{0\}$ and 
a subordinate partition of unity $\{\varphi_i\}_{i=0,1 \ldots, \infty}$. Furthermore,
let $\{\varphi'_i\}_{i=0,1 \ldots, \infty}$ denote a set of functions with $\varphi'_i \in C^\infty_0(U_i)$
and $\varphi_i \varphi'_i = \varphi_i$ for $i=0,1, \ldots, \infty$. Herewith $A_0$
can be expressed as a Leibniz-Mellin pseudo-differential operator in the algebra $L^2(M,{\boldsymbol g})$, i.e.,
\begin{equation}
 A_0 = \sum_i \sigma \varphi_i \Op_y r^{-2} \opm_M^{\gamma-1}(a_0) \varphi'_i \tilde{\sigma}
 +(1-\sigma)A_{\interior}(1-\hat{\sigma})
\label{A}
\end{equation}
with holomorphic operator valued symbol
\begin{equation}
 a_0(r,w,\eta) = -\tfrac{1}{2t^2} \bigl( w^2 -w - r^2C(\eta) + \Delta_{S^2} \bigr) 
\label{a0}
\end{equation}
and
\[
 C(\tau,\Theta_2,\Phi_2) := t^2 \tau^2+\Theta_2^2+\frac{\Phi_2^2}{\sin^2 \theta_2}.
\]
In the second term $A_0$ has been pushed forward in the interior to a standard pseudo-differential operator
\[
 A_{\interior} := \Op (a_{\interior})
\]
with symbol
\[
 a_{\interior}(\xi,\eta) = \frac{1}{2t^2} \left( \xi_1^2 + \xi_2^2 + \xi_3^2 +C(\eta) \right) .
\]

\begin{remark}
A parametrix $P_0$ of the edge-degenerate operator $\mathrm{(\ref{A})}$ can be written in the general form
\begin{equation}
 P_0 = \sum_i \sigma' \varphi_i \Op_y (p_0) \varphi'_i \tilde{\sigma}'
 +(1-\sigma')P_{\interior}(1-\hat{\sigma}') 
\label{A0}
\end{equation}
with $p_0 \in R^{-2}(Y \times \mathbb{R}^3, {\boldsymbol g})$ and
given cut-off functions $\hat{\sigma}' \prec \sigma' \prec \tilde{\sigma}'$. We can assume w.l.o.g. 
$\tilde{\sigma}' \prec \sigma$. Therefore,
\begin{eqnarray*}
 P_0 A_0 &=& \sum_i \sigma' \varphi_i \Op_y (p_0) \varphi'_i \tilde{\sigma}'
           \sigma A_0 \tilde{\sigma} +\sum_i \sigma' \varphi_i \Op_y (p_0) \varphi'_i \tilde{\sigma}'(1-\sigma)A_{\interior}(1-\hat{\sigma})\\
       && +
 (1-\sigma')P_{\interior}(1-\hat{\sigma}')
 \sigma A_0 \tilde{\sigma} +
 (1-\sigma')P_{\interior}(1-\hat{\sigma}') (1-\sigma)A_{\interior}(1-\hat{\sigma}) 
\end{eqnarray*}
and the asymptotic behaviour is completely determined by
\begin{equation}
 P_0 A_0 \sim \sum_i \sigma' \varphi_i \Op_y (p_0) \varphi'_i \tilde{\sigma}'
 \sigma A_0 \tilde{\sigma}.
\label{P0A0asymp}
\end{equation}
\end{remark}

The Leibniz-Mellin representation of (\ref{P0A0asymp}) is conveniently performed by the following steps 
\begin{eqnarray}
 \sum_i \sigma' \varphi_i \Op_y (p_0) \varphi'_i \tilde{\sigma}' 
 \sigma A_0 \tilde{\sigma}
 & = & \sum_i \sigma' \varphi_i \Op_y (p_0) \varphi'_i \tilde{\sigma}' 
 \sigma A_0 \tilde{\sigma} \varphi''_i \nonumber\\
 & = & \sum_i \sigma' \varphi_i \Op_y (p_0) \varphi'_i \tilde{\sigma}' 
 \sigma \Op_y r^{-2} \opm_M^{\gamma-1}(a_0) \tilde{\sigma} \varphi''_i \nonumber\\
 & = & \sum_i \sigma' \varphi_i \Op_y (p_0) \tilde{\sigma}'
 \sigma \Op_y r^{-2} \opm_M^{\gamma-1}(a_0) \tilde{\sigma} \varphi''_i \label{!}\\
 & & -\sum_i \underbrace{\sigma' \varphi_i \Op_y (p_0) (1-\varphi'_i) \tilde{\sigma}' 
 \sigma \Op_y r^{-2} \opm_M^{\gamma-1}(a_0) \tilde{\sigma} \varphi''_i}_{\in \, L^{-\infty}_G(M, {\boldsymbol g})},\label{*}
\end{eqnarray}
where $\{\varphi''_i\}_{i=0,1 \ldots \infty}$ denotes another set of functions with $\varphi''_i \in C^\infty_0(U_i)$
and $\varphi'_i \prec \varphi''_i$ for $i=0,1, \ldots \infty$. The Green operator character in the last line
follows from $\varphi_i (1-\varphi'_i)=0$, cf.~\cite[Remark 3.4.51]{Schulze98}.

\begin{proposition}\label{prop1}
A Green operator of the form
\[
 \sum_i \sigma' \varphi_i \Op_y (h) \xi_i \tilde{\sigma}' ,
\]
with $\supp \varphi_i \cap \supp \xi_i = \emptyset$ and $\varphi_i, \xi_i \in C^\infty_0(U_i)$, $i=0,1, \ldots \infty$,
is smoothing, i.e., it belongs to $L^{-\infty}_{\smooth}(M, {\boldsymbol g})$. 
\end{proposition}

\begin{proof}
This follows from the properties of the corresponding 
Schwartz kernel of the edge-degenerate operator which can be written on a coordinate neighbourhood as
\[
 \varphi_i(y) r^{-1} k_h ((y-y')/r,y) \xi_i(y') .
\]
According to our assumption, the kernel vanishes for $|y-y'| < \epsilon$ for some $\epsilon >0$.
The standard kernel estimate 
\[
 | \partial_z^\alpha \partial_y^\beta k_h (z,y) | \leq c |z|^{-3-m- |\alpha | -N} , \ \ z \neq 0 ,
\]
for all $N \in \mathbb{N}, z=y-y^\prime,$ so that $3+m+|\alpha |+N >0$, applied to the edge-degenerate case
yields the following estimate
\[
 | \varphi_i(y) r^{-1} k_h ((y-y')/r,y,y') \xi_i(y') | \leq c r^{2+m+N} , 
\]
which demonstrates the smoothing property.
\end{proof}

Let us write (\ref{*}) in the form
\[
 \sum_i \sigma' \varphi_i \Op_y (p_0) (1-\varphi'_i)\varphi'''_i\tilde{\sigma}' 
 \sigma \Op_y r^{-2} \opm_M^{\gamma-1}(a_0) \tilde{\sigma} \varphi''_i
\]
for a set of functions $\{\varphi_i'''\}_{i=0,1,\ldots,\infty}$ with $\varphi'''_i\in C_0^\infty(U_i)$ and $\varphi''_i\prec \varphi'''_i.$ Then, from Proposition \ref{prop1} it follows that (\ref{*}) 
belongs to $L^{-\infty}_{\smooth}(M, {\boldsymbol g})$.

In order to calculate the parametrix it is convenient to make the following ansatz
\begin{equation}
 p_0(y,\eta) = \omega'_{1,\eta} r^2 p_{M,0}(y,\eta) \omega'_{0,\eta} 
 + (1-\omega'_{1,\eta}) r^2 p_{\psi,0}(y,\eta) (1-\omega'_{2,\eta}) 
\label{p0}
\end{equation}
with cut-off functions $\omega'_2, \omega'_1, \omega'_0$ satisfying $\omega'_2 \prec \omega'_1 \prec \omega'_0,$ 
where we write $\omega_\eta(r) := \omega ([\eta]r)$; here, $[\eta]$ is any fixed strictly positive function in 
$C^\infty(\mathbb{R}^3)$ such that $[\eta] = |\eta|$ for $\eta \geq \epsilon$ for some $\epsilon >0$.
By a slight modification of the standard notation we incorporate into $p_M$ contributions from 
$R^{-2}_{M+G}(Y \times \mathbb{R}^3, {\boldsymbol g})$
and assume a Mellin representation
\begin{equation}
 p_{M,0}(y,\eta) := \opm_M^{\gamma-3}(a_0^{(-1)})(y,\eta) .
\label{pM0}
\end{equation}

\begin{remark}
The operator $(1-\omega'_{1,\eta}) r^2 p_{\psi,0} (1-\omega'_{2,\eta})$  does not contribute to the asymptotic behaviour.
For a proof see, e.g., {\rm \cite[Prop. 3.3.38]{Schulze98}}.
\label{remark2}
\end{remark}

According to the previous remark, we ignore the second term of (\ref{p0}) in the following considerations.
Inserting the ansatz (\ref{p0}) into the operator product
(\ref{!})
with the parametrix acting from the left, yields the Leibniz product
\[
 \sigma' \omega'_{1,\eta} r^2 \opm_M^{\gamma-3}(a_0^{(-1)}) \omega'_{0,\eta} \tilde{\sigma}' \#_{\eta,y}
 \sigma r^{-2} \opm_M^{\gamma-1}(a_0) \omega_{0,\eta} \tilde{\sigma},
\]
where an optional parameter dependent cut-off function $\omega_{0,\eta}$,with $\omega'_{0,\eta} \prec \omega_{0,\eta}$, has been 
multiplied to the right. Actually this does not affect the operator product because $\opm_M^{\gamma-1}(a_0)$ represents a local
differential operator. It enables us, however, to get rid of the disturbing parameter dependent cut-off function $\omega'_{0,\eta}$ 
by adding a further Green operator\footnote{The Green character of this operator actually follows from $\omega'_{1,\eta} \prec \omega'_{0,\eta}$, 
cf.~\cite[Lemmas 2.3.73, 3.3.27]{Schulze98}}
\[
 \sigma' \omega'_{1,\eta} r^2 \opm_M^{\gamma-3}(a_0^{(-1)}) (1- \omega'_{0,\eta}) \tilde{\sigma}' \#_{\eta,y}
 \sigma r^{-2} \opm_M^{\gamma-1}(a_0) \omega_{0,\eta} \tilde{\sigma} ,
\]
which belongs to $R^0_G(Y \times \mathbb{R}^3, {\boldsymbol g})$, to the Leibniz product. We denote our approach to simplify the construction of the parametrix by adding temporarily an appropriate Green operator and neutralizing it at the end of the calculation via dilation of the cut-off functions as $\varepsilon$-regularization. The basic idea and some relevant technical details
are given in Appendix \ref{epsreg}.

With this modification, the Leibniz product becomes
\begin{eqnarray}
\nonumber
 &&\sigma' \omega'_{1,\eta} r^2 \opm_M^{\gamma-3}(a_0^{(-1)}) \tilde{\sigma}' \#_{\eta,y}
 \sigma r^{-2} \opm_M^{\gamma-1}(a_0) \omega_{0,\eta} \tilde{\sigma}\\ \nonumber
 & = & \sigma' \sum_\alpha \frac{1}{\alpha !} \partial_\eta^\alpha \biggl( 
 \omega'_{1,\eta} r^2 \opm_M^{\gamma-3}(a_0^{(-1)}) \biggr) \tilde{\sigma}' 
 \sigma r^{-2} D^\alpha_y \biggl(\opm_M^{\gamma-1}(a_0)\biggr) \omega_{0,\eta} \tilde{\sigma} \ \mbox{mod} \ R^{-\infty}_{\gflat}(Y \times \mathbb{R}^3, {\boldsymbol g})
  \\ \nonumber
 & = & \sigma' \sum_\alpha \frac{1}{\alpha !} \partial_\eta^\alpha \biggl( 
 \omega'_{1,\eta} \opm_M^{\gamma-1}(T^2 a_0^{(-1)}) \biggr) \tilde{\sigma}' 
 \sigma D^\alpha_y \biggl(\opm_M^{\gamma-1}(a_0)\biggr) \omega_{0,\eta} \tilde{\sigma} \\ \nonumber
 & = & \sigma' \omega'_{1,\eta} \sum_\alpha \frac{1}{\alpha !} 
 \partial_\eta^\alpha \biggl(\opm_M^{\gamma-1}(T^2 a_0^{(-1)})\biggr) \tilde{\sigma}' 
 \sigma D^\alpha_y \biggl(\opm_M^{\gamma-1}(a_0)\biggr) \omega_{0,\eta} \tilde{\sigma} \ \mbox{mod} \ R^{0}_{\gflat}(Y \times \mathbb{R}^3, {\boldsymbol g}) \\ \nonumber
 & = & \sigma' \omega'_{1,\eta} \sum_\alpha \frac{1}{\alpha !} 
 \partial_\eta^\alpha \biggl(\opm_M^{\gamma-1}(T^2 a_0^{(-1)})\biggr)
 D^\alpha_y \biggl(\opm_M^{\gamma-1}(a_0)\biggr) \omega_{0,\eta} \tilde{\sigma} \\ \nonumber
 & & + \underbrace{ \sigma' \omega'_{1,\eta} \sum_\alpha \frac{1}{\alpha !} 
 \partial_\eta^\alpha \biggl(\opm_M^{\gamma-1}(T^2 a_0^{(-1)})\biggr)
 \bigl(\tilde{\sigma}' \sigma -1 \bigr)
 D^\alpha_y \biggl(\opm_M^{\gamma-1}(a_0)\biggr) 
 \omega_{0,\eta} \tilde{\sigma}}_{=: g_{0,1} \in \, R^{-\infty}_G(Y \times \mathbb{R}^3, {\boldsymbol g})} \ \mbox{mod} \ R^{0}_{\gflat}(Y \times \mathbb{R}^3, {\boldsymbol g})
 \\ \nonumber
 & = & \sigma' \omega'_{1,\eta} \sum_\alpha \frac{1}{\alpha !}
 \partial_\eta^\alpha \biggl(\opm_M^{\gamma-1}(T^2 a_0^{(-1)})\biggr)
 D^\alpha_y \biggl(\opm_M^{\gamma-1}(-\tfrac{1}{2t^2} T^2h_0 +\tfrac{r^2C}{2t^2})\biggr) \omega_{0,\eta} \tilde{\sigma} + g_{0,1} \\ \nonumber
 & = & \sigma' \omega'_{1,\eta} \sum_\alpha \frac{1}{\alpha !} \left[
 D^\alpha_y(\tfrac{r^2 C}{2t^2}) \opm_M^{\gamma-3}(\partial_\eta^\alpha a_0^{(-1)})
 + D^\alpha_y(-\tfrac{1}{2t^2}) \opm_M^{\gamma-1}\bigl(T^2(\partial_\eta^\alpha a_0^{(-1)}h_0)\bigr) 
 \right] \omega_{0,\eta} \tilde{\sigma} + g_{0,1} \\ \nonumber
 & = & \sigma' \omega'_{1,\eta} \sum_\alpha \frac{1}{\alpha !}
 \opm_M^{\gamma-1}\bigl( D^\alpha_y(-\tfrac{1}{2t^2}) T^2(\partial_\eta^\alpha a_0^{(-1)}h_0) +
 D^\alpha_y(\tfrac{r^2C}{2t^2}) \partial_\eta^\alpha a_0^{(-1)}\bigr) 
 \omega_{0,\eta} \tilde{\sigma} + g_{0,1} \\ \label{g012}
 & & + \underbrace{ \sigma' \omega'_{1,\eta} \sum_\alpha \frac{1}{\alpha !} 
 D^\alpha_y(\tfrac{r^2 C}{2t^2}) \partial_\eta^\alpha \biggl( \opm_M^{\gamma-3}(a_0^{(-1)}) - \opm_M^{\gamma-1}(a_0^{(-1)}) \biggr)
 \omega_{0,\eta} \tilde{\sigma}}_{=: g_{0,2} \in \, R^0_G(Y \times \mathbb{R}^3, {\boldsymbol g})} .
\end{eqnarray}
Here we used the fact that $\partial_\eta^\alpha \omega'_{1,\eta}$ for $|\alpha| >0$ generates 
smoothing operators which do not contribute to the asymptotics according to the same arguments as in Remark \ref{remark2}.
For the Green operator character of $g_1$ and $g_2$, cf.~\cite[Remark 3.12]{GSS00} and \cite[Proposition 2.3.69]{Schulze98}, respectively.   
In the following we have to determine the symbol $a_0^{(-1)}$ such that 
\begin{equation}
 \sigma' \sum_\alpha \frac{1}{\alpha !} 
 \opm_M^{\gamma-1}\bigl( D^\alpha_y(-\tfrac{1}{2t^2}) T^2(\partial_\eta^\alpha a_0^{(-1)}h_0) +
 D^\alpha_y(\tfrac{r^2C}{2t^2}) \partial_\eta^\alpha a_0^{(-1)}\bigr) \tilde{\sigma}
 = I \ \mbox{mod} \ R^0_{\gflat}(Y \times \mathbb{R}^3, {\boldsymbol g}_l) .
\label{leftp}
\end{equation}

For our application, we are particularly interested in the Green operator symbols $g_1$ and $g_2$ which correspond to
the left action of the parametrix on the Hamiltonian. It is, however, obvious that any left parametrix is
a right parametrix as well and vice versa. Acting on the right side of the Hamiltonian with the parametrix results 
in another sequence of Green and flat operator symbols plus the condition
\begin{equation}
 \sigma \sum_\alpha \frac{1}{\alpha !} 
 \opm_M^{\gamma-3}\bigl( \partial_\eta^\alpha T^{-2} a_0 \#_{r,w} D^\alpha_y a_0^{(-1)} \bigr) \tilde{\sigma}'
 = I \ \mbox{mod} \ R^0_{\gflat}(Y \times \mathbb{R}^3, {\boldsymbol g}_r) ,
\label{rightp}
\end{equation}
from which the symbol $a_0^{(-1)}$ can be derived.

\begin{remark}
Conditions $\mathrm{(\ref{leftp})}$ and $\mathrm{(\ref{rightp})}$ are equivalent. Once one has constructed e.g.~a right parametrix,
where the symbol $a_0^{(-1)}$ satisfies $\mathrm{(\ref{rightp})}$,
it will also satisfy condition $\mathrm{(\ref{leftp})}$.
This follows from the fact that the parameter dependent Mellin pseudo-differential operators in conditions
$\mathrm{(\ref{leftp})}$ and $\mathrm{(\ref{rightp})}$ are equal modulo Green and flat operators. Therefore, both Mellin pseudo-differential operators must have the same sequence of conormal symbols. 
\end{remark}

\subsection{Paving the way for the construction of the asymptotic parametrix}
In the previous section we have discussed the general structure
of the initial parametrix $P_0$ and derived the equivalent conditions $\mathrm{(\ref{leftp})}$ and $\mathrm{(\ref{rightp})}$ which can be
used for their actual construction. Furthermore, we have derived for a left parametrix the corresponding Green operators which contribute to the asymptotic behaviour. Next, we want to present an asymptotic ansatz
for the parameter dependent symbol of the parametrix which  can be used in conditions $\mathrm{(\ref{leftp})}$ and 
$\mathrm{(\ref{rightp})}$ for actual calculations. Here one has to
pay attention to the specific symbolic structure of the edge-degenerate calculus.  

In our calculations we have chosen condition (\ref{rightp}) because the Leibniz-Mellin product in the corresponding equation
\begin{equation}
 \sum_\alpha \frac{1}{\alpha !} \partial_\eta^\alpha T^{-2} a_0 \#_{r,w}
 D^\alpha_y a_0^{(-1)} =
 \sum_{\alpha} \sum_k \frac{1}{\alpha !} \frac{1}{k!} \partial_w^k \partial_\eta^\alpha 
 T^{-2} a_0 \, (-r\partial_r)^k D^\alpha_y a_0^{(-1)} \sim 1 
\label{LMright}
\end{equation}
is represented by a finite number of terms, since $a_0$ is a second order polynomial in $w$ and $\eta$. 
The symbol $a_0^{(-1)}$ of the parametrix has an asymptotic expansion, i.e., 
\begin{equation}
 a_0^{(-1)} \sim -2t^2 \big(q_{0,0} +rq_{0,1} +r^2 q_{0,2} + \cdots\big) ,
 \label{a0mansatz}
\end{equation}
where individual terms can be computed in a recursive manner. Inserting the ansatz (\ref{a0mansatz}) into
(\ref{LMright}) one gets at zeroth order
\begin{equation}
 r^0: \ \ T^{-2} a_0 q_{0,0} + \partial_w T^{-2} a_0 (-r \partial_r) q_{0,0}
 + \tfrac{1}{2} \partial^2_w T^{-2} a_0 (-r \partial_r)^2 q_{0,0} \sim 1 .
\label{eqr0}
\end{equation}

In order to obtain $q_{0,0}$ as a solution of (\ref{eqr0}) let us start with some formal considerations
motivated by the standard approach to the construction of a parametrix in the non-singular pseudo-differential calculus.
It should be mentioned, however, that each individual step requires a thorough justification in the singular edge calculus
which will be done in the following.
Since $q_{0,0}(r,y,w,\eta) \in C^{\infty}(\overline{\mathbb{R}}_+\times Y,L_{cl}^{-2}(X, \Gamma_\alpha \times \mathbb{R}^3))$ with $\Gamma_\alpha:=\{w\in\mathbb{C}: \mathrm{Re}\,w=\alpha\},$ 
we can obtain an asymptotic sum
\begin{equation}
 q_{0,0} \sim q'_{0,0} +q'_{0,1} + q'_{0,2} + \cdots 
\label{a0m1tilde}
\end{equation}
with $q'_{0,k}(r,y,w,\eta) = \tilde{q}'_{0,k}(r,y,w,r\eta)$ and $\tilde{q}'_{0,k}(r,y,w,\tilde{\eta}) \in 
C^\infty(\overline{\mathbb{R}}_+ \times Y, L_{cl}^{-2}(X, \Gamma_\alpha \times \mathbb{R}^3))$, cf.~\cite[Section 1.1]{FHS16} for further details and definitions. Let us take
\begin{equation}
 q'_{0,0} = -2 t^2 h^{-1} ,
\label{defqp0}
\end{equation}
where
\[
 h = (w-2)^2 -(w-2) - r^2C(\eta) + \Delta_{S^2}
\]
corresponds to a holomorphic operator valued symbol acting on the basis $X$.

\begin{remark}
For $\mathrm{(\ref{defqp0})}$ to make sense it requires $h^{-1} \in C^{\infty}(\overline{\mathbb{R}}_+\times
Y,L_{cl}^{-2}(X, \Gamma_\alpha \times \mathbb{R}^3))$ which follows from the
spectral invariance of $h$ as a parameter dependent uniformly elliptic differential operator on the base $X$.
\end{remark}

Successively solving 
\[
 \sum_{m+n=k} \tfrac{1}{m!} \partial^m_w T^{-2} a_0 (-r\partial_r)^m q'_{0,n} =0 \ \ 
 \mbox{for} \ k=1,2,\ldots
\]
yields the recurrence relation
\begin{equation}
 q'_{0,k} = - q'_{0,0} \sum_{\substack{m+n=k\\n<k}}
 \tfrac{1}{m!} \partial^m_w T^{-2} a_0 (-r\partial_r)^m q'_{0,n} ;
\label{qktilde}
\end{equation}
here the sum restricts to $m \leq 2$ because of (\ref{a0}).
In order to get an asymptotic sum with respect to powers in $r \eta$, it is necessary to rearange the  
individual terms. This can be done in two separate steps. Let us first reorder the sum in powers of $r^2 C$.
A simple calculation yields
\begin{eqnarray}
\nonumber
 q_{0,0} & \sim & -2t^2 h^{-1} -2t^2 \bigl( (2z-4) h^{-3} + (4z^2 -16z+16) h^{-4} \\ \nonumber
 & & + (8z^3-48z^2+96z-64) h^{-5} + \cdots \bigr) r^2 C
 + O\bigl((r^2 C)^2 \bigr) \\ \nonumber
 & \sim & -2t^2 h^{-1} -2t^2 \biggl( \sum_{n=1}^\infty (2z-4)^n h^{-n} \biggr) h^{-2} r^2 C 
 + O\bigl((r^2 C)^2 \bigr) \\ \label{a0hr2C}
 & \sim & -2t^2 h^{-1} -2t^2 (2z-4) \bigl( h -(2z-4) \bigr)^{-1} h^{-2} r^2 C
 + O\bigl((r^2 C)^2 \bigr)
\end{eqnarray}
with $z=2w-5$. The geometric sum is tentatively considered in a purely formal sense.
As provisional justification, we observe that the reordered asymptotic expansion (\ref{a0hr2C}) 
satisfies (\ref{eqr0}) up to terms of $O((r^2 C)^2)$. 

The edge-degenerate pseudo-differential calculus imposes stringent conditions on the symbol $a_0^{(-1)}$ of the
parametrix. In order to demonstrate that individual terms in our formal sum (\ref{a0hr2C}) are compatible
with the edge-degenerate calculus let us first consider the leading order term $h^{-1}$ which represents a 
meromorphic operator valued symbol acting on the basis $X$.
Because of $r^2 C$ in the denominator, its poles depend on the values of the covariables of the edge and therefore
it cannot be directly identified with a symbol in the edge calculus. In order to see how this symbol actually fits 
into the calculus it is convenient to consider a Lippmann-Schwinger expansion  
\begin{align*}
 h^{-1} & = h_0^{-1} + h^{-1} [h_0 -h] h_0^{-1} \\
 h^{-1} & = h_0^{-1} + h_0^{-1} [h_0 -h] h_0^{-1} + h^{-1} [h_0 -h] h_0^{-1} [h_0 -h] h_0^{-1} \\
 \vdots &
\end{align*}
with
\[
 h_0 = (w-2)^2 - (w-2) +\Delta_{S^2} .
\]
For arbitrary $N \in \mathbb{N}$ this becomes
\begin{equation}
 h^{-1} = \underbrace{\sum_{n=0}^{N-1} \bigl( r^2C \bigr)^n h_0^{-1-n}}_{=: h^{-1}_{\LS}} + \underbrace{\bigl( r^2C \bigr)^N h^{-1} h_0^{-N}}_{=: h^{-1}_{\gflat}} .
\label{LippSch}
\end{equation}
Applying the kernel cut-off $H(\phi)$, with $\phi \in C^\infty_0(\mathbb{R}_+)$ and $\phi=1$ in a neighbourhood of
$r=1$, to $h^{-1} \in C^{\infty}(\overline{\mathbb{R}}_+\times Y,L_{cl}^{-2}(X, 
\Gamma_\alpha \times \mathbb{R}^3))$ one gets
\begin{equation}\label{h-1}
 h^{-1} = H(\phi) h^{-1} + \bigl( 1- H(\phi) \bigr) h^{-1} ,
\end{equation}
where
\[
 H(\phi) h^{-1}(r,y,w,\eta)= \widetilde{H(\phi) h^{-1}}(y,w,r\eta) 
\]
with $\widetilde{H(\phi) h^{-1}} \in C^\infty(\overline{\mathbb{R}}_+\times Y,M^{-2}_{\cal O}(X,\mathbb{R}^3))$, cf.~\cite[Section 1.1]{FHS16} for further details and definitions.
Similarly, one can apply the kernel cut-off to the two terms
in the decomposition $h^{-1} = h^{-1}_{\LS} + h^{-1}_{\gflat}$, cf.~(\ref{LippSch}), separately.
The operators corresponding to the holomorphic Mellin symbol $H(\phi) h^{-1}$, $H(\phi) h^{-1}_{\LS}$ and $H(\phi) h^{-1}_{\gflat}$ map
\[
 \omega_\eta r^2 \opm_M^{\gamma-3}
 \bigl( H(\phi) h^{-1}_{(\LS)} \bigr) \omega'_\eta :
 {\cal K}^{s,\gamma-2} \longrightarrow {\cal K}^{s+2,\gamma} ,
\]
\begin{equation}
 \omega_\eta r^2 \opm_M^{\gamma-3}
 \bigl( H(\phi) h^{-1}_{\gflat} \bigr) \omega'_\eta :
 {\cal K}^{s,\gamma-2} \longrightarrow {\cal K}^{s+2+N,\gamma+N} ,
 \label{HLM} 
\end{equation}
respectively.

Inserting the Lippmann-Schwinger expansion (\ref{LippSch}), the operator corresponding to the second term of (\ref{h-1}) becomes
\[
 \omega_\eta r^2 \opm_M^{\gamma-3} \bigl(
 \bigl( 1- H(\phi) \bigr) h^{-1} \bigr) \omega'_\eta = \sum_{n=0}^{N-1} m_n + g_N ,
\]
with
\[
 m_n(y,\eta) := \omega_\eta r^2 \bigl( r^2C \bigr)^n \opm_M^{\gamma-3}
 \bigl( \bigl( 1- H(\phi) \bigr) h_0^{-1-n} \bigr) \omega'_\eta 
\]
and
\[
 g_N(y,\eta) := \omega_\eta r^2 \bigl( r^2C \bigr)^N \opm_M^{\gamma-3}
 \bigl( \bigl( 1- H(\phi) \bigr) h^{-1} h_0^{-N} \bigr) \omega'_\eta .
\]

\begin{proposition}\label{prop2}
The parameter dependent operators $m_n$, $n=0,\ldots,N-1$, represent 
smoothing Mellin operators 
\[
 m_n(y,\eta) : {\cal K}^{s,\gamma-2} \longrightarrow {\cal K}^{\infty,\gamma+2n} ,
\]
with $\bigl( 1- H(\phi) \bigr) h_0^{-1-n} \in M^{-\infty}_R(X)$.
For $n$ sufficiently large $m_n$ and the remainder $g_N$ belong to $R_G^{-2}(Y \times \mathbb{R}^3,{\boldsymbol g})$ 
and do not contribute to the asymptotics.
\end{proposition}

\begin{proof}
The mapping property for $m_n$ follows from \cite[Lemma 3.3.22]{Schulze98}.

Twisted homogeneity for $g_N$ can be easily shown, i.e.,
\begin{eqnarray*}
\kappa_\lambda g_N(y,\eta) \kappa^{-1}_\lambda u & = & \kappa_\lambda \omega_\eta r^2 \bigl( r^2C \bigr)^N \opm_M^{\gamma-3}
 \bigl( \bigl( 1- H(\phi) \bigr) h^{-1} h_0^{-N} \bigr) \omega'_\eta \kappa^{-1}_\lambda u \\
 & = & \omega_{\lambda \eta} (\lambda r)^2 \bigl( \lambda^2 r^2C \bigr)^N \int_\mathbb{R} \int_0^\infty
 \biggl( \frac{\lambda r}{r'} \biggr)^{-(\frac{7}{2} -\gamma +i\rho)} \biggl[ \bigl( 1- H(\phi) \bigr) h^{-1} h_0^{-N} \biggr]
 (\frac{7}{2} -\gamma +i\rho,\lambda \eta) \\ & & \times \omega'_\eta u( \lambda^{-1} r')\, \dbar \rho \frac{dr'}{r'} \\
 & = & \omega_{\lambda \eta} (\lambda r)^2 \bigl( \lambda^2 r^2C \bigr)^N \int_\mathbb{R} \int_0^\infty
 \biggl( \frac{r}{\tilde{r}} \biggr)^{-(\frac{7}{2} -\gamma +i\rho)} \biggl[ \bigl( 1- H(\phi) \bigr) h^{-1} h_0^{-N} \biggr]
 (\frac{7}{2} -\gamma +i\rho,\lambda \eta) \\
 & & \times \omega'_{\lambda \eta} u(\tilde{r}) \, \dbar \rho \frac{d\tilde{r}}{\tilde{r}} \\
 & = & \lambda^2 g_N(y,\lambda \eta) u .
\end{eqnarray*}
The poles $w_0$ of $h^{-1}(y,\eta)$ are all located outside a strip $\{ w: 2 < |w| < 3\}$ in the complex plane.
Therefore, the integration contour $\Gamma_{\frac{7}{2} -\gamma}$  for $\frac{1}{2} < \gamma < \frac{3}{2}$
is separated from the poles, and for fixed $y,\eta$ it follows from \cite[Section 7.2.3, Theorem 9]{ES97} that
\[
 g_N(y,\eta) : {\cal K}^{s,\gamma-2} \longrightarrow {\cal S}^{\gamma+2N} .
\]
With respect to the norm of ${\cal K}^{s,\gamma}$ spaces one gets
\begin{multline*}
 \bigl\| \kappa^{-1}_{[\eta]} \omega_\eta r^2 \bigl( r^2C(\eta) \bigr)^N \opm_M^{\gamma-3}
 \bigl( \bigl( 1- H(\phi) \bigr) h^{-1} h_0^{-N} \bigr)(y,\eta) \omega'_\eta \kappa_{[\eta]}
 \bigr\|_{{\cal L}({\cal K}^{s,\gamma-2},{\cal K}^{\infty,\gamma+2N})} \\
 = \bigl\| \omega_{[\eta]^{-1} \eta} r^2 \bigl( r^2C( [\eta]^{-1} \eta) \bigr)^N \opm_M^{\gamma-3}
 \bigl( \bigl( 1- H(\phi) \bigr) h^{-1} h_0^{-N} \bigr)(y, [\eta]^{-1} \eta) \omega'_{[\eta]^{-1} \eta}
 \bigr\|_{{\cal L}({\cal K}^{s,\gamma-2},{\cal K}^{\infty,\gamma+2N})}
 [\eta]^{-2} ,
\end{multline*}
and for derivatives of the symbol
\begin{multline*}
 \bigl\| \kappa^{-1}_{[\eta]} D^\alpha_y D^\beta_\eta \omega_\eta r^2 \bigl( r^2C(\eta) \bigr)^N \opm_M^{\gamma-3}
 \bigl( \bigl( 1- H(\phi) \bigr) h^{-1} h_0^{-N} \bigr)(y,\eta) \omega'_\eta \kappa_{[\eta]}
 \bigr\|_{{\cal L}({\cal K}^{s,\gamma-2},{\cal K}^{\infty,\gamma+2N})} \\
 \leq C [\eta]^{-2-|\beta|} 
\end{multline*}
for $y \in K \subset U \subset Y$, $\eta \in \mathbb{R}^3$. Therefore, $g_N$ belongs to 
$S^{-2}_{cl}(U \times \mathbb{R}^3; {\cal K}^{s,\gamma-2},{\cal K}^{\infty,\gamma+2N})$ 
and $g_N \in R_G^{-2}(U \times \mathbb{R}^3, {\boldsymbol g})$ follows according to \cite[Definition 3.3.1]{Schulze98} .
\end{proof}

Individual terms of our formal expansions (\ref{a0hr2C}) can be brought into the specific operator formats using 
kernel cut-offs and Lippmann-Schwinger expansions.  It is straightforward to extend our previous discussion for $h^{-1}$ to more general higher order terms.
In summary, we may say that our formal expansion (\ref{a0hr2C}) fulfills the requirements of the
edge-degenerate calculus. 

As a matter of principle we can proceed by calculating higher order terms of $q_0$ in the asymptotic expansion   
(\ref{a0hr2C}) and use an analogous ansatz for the higher order terms $q_i$, $i=1,2,\ldots$, 
of the asymptotic expansion (\ref{a0mansatz}) of $a_0^{(-1)}$. To this end let us reformulate (\ref{a0mansatz}) by using the Lippmann-Schwinger expansion in the following form
\begin{eqnarray*}
 a_0^{(-1)} & \sim & -2t^2 \sum_{n \geq 0} r^n q_{0,n} \\
 & = & -2t^2 \sum_{n \geq 0} r^n \biggl( H(\phi) q_{0,n} + \bigl( 1- H(\phi) \bigr) q_{0,n} \biggr) \\
 & = & -2t^2 \sum_{n \geq 0} r^n \biggl( H(\phi) [q_{0,n} ]_{\LS} +
 H(\phi) [q_{0,n} ]_{\gflat} + \bigl( 1- H(\phi) \bigr) [q_{0,n} ]_{\LS} 
 + \bigl( 1- H(\phi) \bigr) [q_{0,n} ]_{\gflat} \biggr) ,
\end{eqnarray*}
where $ [q_{0,n} ]_{\LS}$ and $[q_{0,n} ]_{\gflat}$ for $q_{0,n}$ have the same meaning as the first and second term of the decomposition (\ref{LippSch}) for $h^{-1}$, respectively.
In the third line we have, beside holomorphic symbols, smoothing Mellin symbols
$\bigl( 1- H(\phi) \bigr) [q_{0,n} ]_{\LS}$ and flat Green remainders $\bigl( 1- H(\phi) \bigr) [q_{0,n} ]_{\gflat}$. 
The latter can be neglected in our considerations according to Proposition \ref{prop2}.
Furthermore, it turns out that holomorphic symbols $H(\phi) [q_{0,n} ]_{\LS}$ and $H(\phi) [q_{0,n} ]_{\gflat}$ do not contribute to the Green operators $g_{l,1}$ and $g_{l,2}$, i.e.,
only the $\bigl( 1- H(\phi) \bigr) [q_{0,n} ]_{\LS}$ term actually contributes. Furthermore, according to (\ref{HLM}), the operator
corresponding to the holomorphic symbol $H(\phi) [q_{0,n} ]_{\gflat}$
does not even contribute to the asymptotic parametrix.
Therefore, in the following we consider in the asymptotic expansion of $a_0^{(-1)}$ only the terms $[q_{0,n}]_{\LS}.$ For the latter we take the asymptotic ansatz
\begin{equation}
[q_{0,n} ]_{\LS} \sim -2t^2 \sum_{j \geq 0} d_{n,j}^{(0)}(w,r\eta) ,
\label{qnLS}
\end{equation}
where $d_{n,j}$ denotes a homogeneous polynomial of order $j$ in the degenerate covariables $r\eta$, i.e.,
$d_{n,j}^{(0)}(\lambda r \eta) = \lambda^j d_{n,j}^{(0)}(r \eta)$ with coefficients depending on the operator valued symbol $h_0$.

Summarising our previous discussion, we have derived for the symbol $a_0^{(-1)}$ of the initial parametrix $P_0$ the asymptotic
expansion 
\begin{eqnarray}
\label{a0asymp}
 a_0^{(-1)} & \sim & -2t^2 \bigg( [q_{0,0} ]_{\LS}  +r [q_{0,1} ]_{\LS} +r^2 [q_{0,2} ]_{\LS} + \cdots\bigg) \\ \nonumber 
 & \sim & -2t^2 \sum_{n \geq 0} \sum_{j \geq 0} r^n d_{n,j}^{(0)}(w,r\eta) .
\end{eqnarray} 
For the sake of notational simplicity,
we skip the subscript LS in the following, notwithstanding
that all asymptotic expansions of the type (\ref{a0asymp}), here and in the following,
refer to the $[ \cdot ]_{\LS}$ part of the Lippmann-Schwinger decomposition. In order to show, that the asymptotic
expansion (\ref{a0asymp}) is actually consistent 
with the symbolic structure of asymptotic parametrices in the edge degenerate calculus, cf.~Section \ref{oooa}, let us first reveal
its relation to the conormal symbols of the parametrix.
 
\begin{proposition} 
The asymptotic expansion {\em (\ref{a0asymp})} 
is equivalent modulo flat remainders to an asymptotic expansion of $a_0^{(-1)}$ via the conormal symbols of $p_{M,0}$, cf.~$\mathrm{(\ref{p0})}$, 
\begin{equation}
 a_0^{(-1)} \sim \sum_{j \geq 0} r^j \sigma^{-2-j}_M(p_{M,0}) (y,\eta) 
 \label{a0conorm}
\end{equation} 
with 
\[
 \sigma_M^{-2-j}(p_{M,0}) (y,\eta) = \frac{1}{j!} \frac{\partial^j}{\partial r^j} a_{0,h}^{(-1)}|_{r=0}
 + \sigma_M^{-2-j}(m_0) (y,\eta),
\]
where $a_{0,h}^{(-1)} \in M^{-2}_{\cal O}(X,\mathbb{R}^3)$ denotes the holomorphic Mellin symbol
and $m_0 \in R^{-2}_{M+G}(Y \times \mathbb{R}^3, {\boldsymbol g})$ refers to the smoothing Mellin operator part
of $p_0$.
\label{propcon}
\end{proposition}

\begin{proof}
Conormal symbols are polynomials in the edge covariable $\eta$. In order to represent $d_{n,j}^{(0)}(w,r\eta)$
by conormal symbols one can decompose the latter into their homogeneous components, i.e.,
\[
 \sigma_M^{-2-j}(p_{M,0}) (y,\eta) = \sum_{k=0}^j Q_{k,j}^{(0)}(\eta) \ \ \mbox{with} \ Q_{k,j}^{(0)}(\lambda \eta) = \lambda^k Q_{k,j}^{(0)}(\eta) .
\]
The asymptotic expansion (\ref{a0conorm}) becomes
\begin{eqnarray*}
 a_0^{(-1)} & \sim & \sum_{j \geq 0} r^j \sum_{k=0}^j Q_{k,j}^{(0)}(\eta) \\
 & \sim & \sum_{k \geq 0} \sum_{j \geq k} r^{j-k} Q_{k,j}^{(0)}(r\eta) \\
 & \sim & \sum_{m \geq 0} r^m \sum_{k \geq 0} Q_{k,m+k}^{(0)}(r\eta) ,
\end{eqnarray*}
and a comparison with (\ref{a0asymp}) shows $Q_{k,m+k}^{(0)}(r\eta) =-2t^2 d^{(0)}_{m,k}(r\eta)$.
\end{proof} 
 
\begin{remark} 
The asymptotic expansion $\mathrm{(\ref{a0asymp})}$ does not depend on the details of the parametrix construction.
In particular it does not matter whether it refers to a left or right parametrix. This is an immediate 
consequence of Proposition $\mathrm{\ref{propcon}}$ and the fact that two edge-degenerate pseudo-differential
operators which agree modulo Green and flat operators have the same conormal symbols.
\end{remark} 

It is obvious that the same type of asymptotic
expansion can be as well applied in the subsequent
steps of the asymptotic parametrix construction. 

The actual calculation of individual terms and corresponding cotributions to Green operators, is rather involved.
Therefore, in the following section, we first want to present our 
main result and some of its immediate consequences for the sake of the
reader who is not interested in the details of the calculations.    
The proof and all technical details are given in Section \ref{4}.

\section{Asymptotics of the Green operator}
\label{3}
\subsection{Main result}
\label{3.1}
Before we present our main theorem, we give a few remarks.
\begin{itemize}
\item[i)] In Section \ref{coap}, we have derived an asymptotic parametrix modulo Green operators in $L^{0}_G(M,\boldsymbol{g})$.
This actually represents the penultimate step in the asymptotic parametrix construction discussed in Ref.~\cite{FHS16}, cf.~Corollary 2.24 and Theorem 2.26 therein. For our purposes
it is sufficient to stop at this point because it already provides us with the desired insight into the asymptotic behaviour of solutions to Schr\"{o}dinger's equation near edge type singularities.
\item[ii)] According to our discussion in Section \ref{2.2}, we have applied the $\varepsilon$-regularization technique, cf.~Appendix \ref{epsreg}, for the construction of the parametrix and Green operator. This is justified by the fact that we apply
the Green operator to an eigenfunction of the Hamiltonian $u$
which according to standard regularity theory, cf.~\cite{Schulze98},
belongs to ${\cal W}^\infty_{\mathrm{loc}} \bigl(Y_i,{\cal K}^{\infty,\gamma}\big((S^2)^{\wedge}\big) \bigr)$. Therefore,
$u$ is $\varepsilon$-regularizable, cf.~Lemma \ref{lemmaeps} below.

\item[iii)] It should be mentioned that the spectral resolution
of the Laplace-Beltrami operator acting on the base of the cone
$S^{2}$ is a pleasant feature of our approach which relates
the relative angular momentum of a pair of particles to its
radial correlation.
\item[iv)] The subsequent formula for the asymptopic expansion of the symbol of the Green operator seems to be rather complicated. 
In particular, terms depending on edge variables and covariables
are rather involved.
It is therefore desirable to provide an independent check of its
correctness. Such a check has been performed in Appendix \ref{noninteracting},
where we apply the corresponding Green operator to eigenfunctions
of a Hamiltonian without $e-e$ interaction potential. These
eigenfunctions are explicitly known for different angular momenta 
and can be used to verify the essential relation $Gu=-u$, cf.~(\ref{PAedge}), via explicit calculations.
\end{itemize}

\begin{theorem}\label{t1}
 The Green operator $g\in R_G^0 (Y_i\times\mathbb{R}^3,{\boldsymbol g})\, (i=1,2,3)$ for $\frac{1}{2}\leq\gamma\leq \frac{3}{2}$ has a leading order asymptotic expansion of the form
\begin{eqnarray}
 g\hat{u}(y,\eta)&=&\sigma^\prime 2t^2\left[\left(1+rtZ_1+r^2\left(-2+\tfrac{1}{3}(tZ_1)^2+\tfrac{1}{3}tZ_2\right)\right){\cal P}_0{\cal Q}_{0,1}(\hat{u})(y,\eta)\right.\nonumber\\
 &&+\tfrac{1}{6}r^2 {\cal P}_0{\cal Q}_{0,2}(\hat{u})(y,\eta)+\left(\tfrac{1}{3}r+\tfrac{1}{6}tZ_1 r^2\right){\cal P}_1{\cal Q}_{1,1}(\hat{u})(y,\eta)\nonumber\\
 &&+\left.\tfrac{1}{5}r^2{\cal P}_2{\cal Q}_{2,1}(\hat{u})(y,\eta)-\tfrac{1}{30}r^2 {\cal P}_2{\cal Q}_{2,2}({\hat u})(y,\eta)\right]+{\cal O}(r^3)\label{Gasymp}
\end{eqnarray}
with 
\[
Z_1:=\left\{\begin{array}{ll}
             \frac{1}{\sqrt{2}} & \mbox{for} \ v_{e-e}\\
             -Z & \mbox{for} \ v_{e-n}
            \end{array}\right., 
            \,\,\,\,\,
Z_2:=-tE+\left\{\begin{array}{ll}
             -2\sqrt{2}Z & \mbox{for} \ v_{e-e}\\
             1-Z & \mbox{for} \ v_{e-n}
            \end{array}\right.,         
\]
where $\varphi_i u\in {\cal W}^\infty_{\mathrm{comp}} \bigl(Y_i,{\cal K}^{\infty,\gamma}\big((S^2)^{\wedge}\big) \bigr)$ and $\hat{u}(r,\varphi_1,\theta_1,\eta):=F_{y\rightarrow\eta}\varphi_i u(r,\varphi_1,\theta_1,y).$ Here, $P_l, l=0,1,2,\ldots,$ 
denote projection operators on subspaces which belong to eigenvalues $-l(l+1)$ of the Laplace-Beltrami operator on $S^2.$ For terms depending on edge variables and covariables, we have used the following notations
\begin{eqnarray*}
 {\cal Q}_{0,1}(\hat{u}) & := & M \bigl( (\tilde{\sigma}' \sigma -1) \opm_M^{\gamma-1}(a) \tilde{\sigma} \hat{u} \bigr) (0)
 + M \bigl( \tfrac{1}{2t^2} r^2 C_0 \tilde{\sigma} \hat{u} \bigr) (0) + M \bigl( \opm_M^{\gamma-1}(rs_1) \tilde{\sigma} \hat{u} \bigr) (0) \\
 & = & M \bigl( \tilde{\sigma}' \sigma \opm_M^{\gamma-1}(a) \tilde{\sigma} \hat{u} \bigr) (0)
 + M \bigl( \opm_M^{\gamma-1}(h_1^{(0)}) \tilde{\sigma} \hat{u} \bigr) (0);
\\
 {\cal Q}_{0,2}(\hat{u}) & := & M \bigl( (\tilde{\sigma}' \sigma -1) \opm_M^{\gamma-1} \bigl( \sum_\alpha \tfrac{1}{\alpha !} 
 \partial_\eta^\alpha \tilde{C}_1 D_y^\alpha a \bigr) \tilde{\sigma} \hat{u} \bigr) (0) \\
 & & + M \bigl( \sum_\alpha \tfrac{1}{\alpha !} \partial_\eta^\alpha \tilde{C}_1 D_y^\alpha
 (\tfrac{1}{2t^2} r^2 C_0) \tilde{\sigma} \hat{u} \bigr) (0)
 + M \bigl( \opm_M^{\gamma-1} \bigl( \sum_\alpha \tfrac{1}{\alpha !} \partial_\eta^\alpha
 \tilde{C}_1 D_y^\alpha r s_1 \bigr) \tilde{\sigma} \hat{u} \bigr) (0) \\
 & = & M \bigl( \tilde{\sigma}' \sigma \opm_M^{\gamma-1} \bigl( \sum_\alpha \tfrac{1}{\alpha !}
 \partial_\eta^\alpha \tilde{C}_1 D_y^\alpha a \bigr) \tilde{\sigma} \hat{u} \bigr) (0) 
 + M \bigl( \opm_M^{\gamma-1} \bigl( \sum_\alpha \tfrac{1}{\alpha !} \partial_\eta^\alpha
 \tilde{C}_1 D_y^\alpha h_1^{(0)} \bigr) \tilde{\sigma} \hat{u} \bigr) (0);\\
 {\cal Q}_{1,1}(\hat{u}) & := & M \bigl( (\tilde{\sigma}' \sigma -1) \opm_M^{\gamma-1}(a) \tilde{\sigma} \hat{u} \bigr) (-1)
 + M \bigl( \tfrac{1}{2t^2} r^2C_0 \tilde{\sigma} \hat{u} \bigr) (-1) \\
 & & + M \bigl( \tfrac{i}{2t^2} r^2 C_1 \tilde{\sigma} \hat{u} \bigr) (-1)
 + M \bigl( \opm_M^{\gamma-1}(r^2 s_2) \tilde{\sigma} \hat{u} \bigr) (-1) \\
 & & -tZ_1 \bigg[ M \bigl( (\tilde{\sigma}' \sigma -1) \opm_M^{\gamma-1}(a) \tilde{\sigma} \hat{u} \bigr) (0)
 + M \bigl( \tfrac{1}{2t^2} r^2C_0 \tilde{\sigma} \hat{u} \bigr) (0)
 + M \bigl( \opm_M^{\gamma-1}(r s_1) \tilde{\sigma} \hat{u} \bigr) (0) \biggr] \\
 & = & M \bigl( \tilde{\sigma}' \sigma \opm_M^{\gamma-1}(a) \tilde{\sigma} \hat{u} \bigr) (-1)
 + M \bigl( \opm_M^{\gamma-1}(h_2^{(1)}) \tilde{\sigma} \hat{u} \bigr) (-1) \\
 & & -tZ_1 \bigg[ M \bigl( \tilde{\sigma}' \sigma \opm_M^{\gamma-1}(a) \tilde{\sigma} \hat{u} \bigr) (0)
 + M \bigl( \opm_M^{\gamma-1}(h_1^{(1)}) \tilde{\sigma} \hat{u} \bigr) (0) \biggr];\\
 {\cal Q}_{2,1}(\hat{u}) & := & 
 M \bigl( (\tilde{\sigma}' \sigma -1) \opm_M^{\gamma-1} \bigl( a \bigr) \tilde{\sigma} \hat{u} \bigr) (-2) 
 + M \bigl( \tfrac{1}{2t^2} r^4C_2 \tilde{\sigma} \hat{u} \bigr) (-2)
 + M \bigl( \opm_M^{\gamma-1}(r^3 s_3) \tilde{\sigma} \hat{u} \bigr) (-2) \\
 & & -\tfrac{1}{2} tZ_1 \bigl[ M \bigl( (\tilde{\sigma}' \sigma -1) \opm_M^{\gamma-1} \bigl( a \bigr) \tilde{\sigma} \hat{u} \bigr) (-1) 
 + M \bigl( \tfrac{1}{2t^2} r^2C_0 \tilde{\sigma} \hat{u} \bigr) (-1) \\
 & & + M \bigl( \tfrac{i}{2t^2} r^2C_1 \tilde{\sigma} \hat{u} \bigr) (-1)
 + M \bigl( \opm_M^{\gamma-1}(r^2 s_2) \tilde{\sigma} \hat{u} \bigr) (-1) \bigr] \\
 & & + \tfrac{1}{6} \bigl( 10 +(tZ_1)^2 -2tZ_2 \bigr) \bigl[
 M \bigl( (\tilde{\sigma}' \sigma -1) \opm_M^{\gamma-1} \bigl( a \bigr) \tilde{\sigma} \hat{u} \bigr) (0) \\
 & & + M \bigl( \tfrac{1}{2t^2} r^2C_0 \tilde{\sigma} \hat{u} \bigr) (0)
 + M \bigl( \opm_M^{\gamma-1}(r s_1) \tilde{\sigma} \hat{u} \bigr) (0) \bigr] \\ 
 & = & M \bigl( \tilde{\sigma}' \sigma \opm_M^{\gamma-1} \bigl( a \bigr) \tilde{\sigma} \hat{u} \bigr) (-2)
 + M \bigl( \opm_M^{\gamma-1}(h_3^{(2)}) \tilde{\sigma} \hat{u} \bigr) (-2) \\
 & & -\tfrac{1}{2} tZ_1 \bigl[ M \bigl( \tilde{\sigma}' \sigma \opm_M^{\gamma-1} \bigl( a \bigr) \tilde{\sigma} \hat{u} \bigr) (-1)
 + M \bigl( \opm_M^{\gamma-1}(h_2^{(2)}) \tilde{\sigma} \hat{u} \bigr) (-1) \bigr] \\
 & & + \tfrac{1}{6} \bigl( 10 +(tZ_1)^2 -2tZ_2 \bigr) \bigl[
 M \bigl( \tilde{\sigma}' \sigma \opm_M^{\gamma-1} \bigl( a \bigr) \tilde{\sigma} \hat{u} \bigr) (0)
 + M \bigl( \opm_M^{\gamma-1}(h_1^{(2)}) \tilde{\sigma} \hat{u} \bigr) (0) \bigr];\\
 {\cal Q}_{2,2}(\hat{u}) & := &
 M \bigl( \tilde{\sigma}' \sigma \opm_M^{\gamma-1} \bigl( \sum_\alpha \tfrac{1}{\alpha !}
 \partial_\eta^\alpha \tilde{C}_1 D_y^\alpha a \bigr) \tilde{\sigma} \hat{u} \bigr) (0) 
 + M \bigl( \opm_M^{\gamma-1} \bigl( \sum_\alpha \tfrac{1}{\alpha !} \partial_\eta^\alpha
 \tilde{C}_1 D_y^\alpha h_1^{(2)} \bigr) \tilde{\sigma} \hat{u} \bigr) (0)
\end{eqnarray*}
with
\begin{eqnarray*}
 & & C_1 := -5t\tau -\cot \theta_2 \Theta_2,\quad\tilde{C}_1 := C_0 -4it\tau +iC_1;
\\
& & rs_1 := a-a_0, \quad r^2s_2 := a-a_0 -ra_1 ;
\\
& & a_0 := -\tfrac{1}{2t^2} T^2 h_0 + \tfrac{1}{2t^2} r^2 C_0, \quad 
 a_1 := \tfrac{i}{2t^2} r C_1 + \tfrac{Z_1}{t},\quad
 a_2 := -\tfrac{4}{3t^2} w - \tfrac{1}{6t^2} \Delta_{S^2} + \tfrac{1}{2t^2} r^2 C_2 + \tfrac{Z_2}{t};
\\
& &  h_1^{(l)} := \tfrac{1}{2t^2} ( w^2-w-l(l+1)), \quad h_2^{(l)} := \tfrac{1}{2t^2} \bigl( w^2-w-l(l+1) -2rtZ_1 \bigr),
\\
& &  h_3^{(l)} := \tfrac{1}{2t^2} \bigl( w^2-w-l(l+1) + \tfrac{8}{3} r^2 w -r^2 C_0 -ir^2 C_1 -2rtZ_1 -2r^2 -2r^2tZ_2 \bigr).
\end{eqnarray*}
\end{theorem}

\subsection{Leading order asymptotic behaviour of helium and the hydrogen molecule}
\label{3.2}

\noindent
Let us first consider the asymptotic behaviour near the $e-e$ edge. The euclidean distance between two electrons 
expressed in hyperspherical coordinates is
\[
 |{\bf x}_1 - {\bf x}_2 | = \sqrt{2} t \sin r_{12} = \sqrt{2} t \bigl( r_{12} + {\cal O}(r_{12}^3) \bigr) .
\]
With this the asymptotic behaviour on the subspace of $l=0$ relative angular momentum, defined by $P_0$, 
becomes 
\[
 G{\cal P}_0 u \sim \bigl( 1 + \tfrac{1}{2} |{\bf x}_1 - {\bf x}_2 | \bigr) w_0(y) + \cdots
\]
which is equivalent to Kato's famous ``cusp'' condition
\[
 \left. \frac{1}{4\pi} \frac{\partial \int_{S^2} u(x_{12},\omega_{12},{\bf y}) d\omega_{12}}{\partial x_{12}} \right|_{x_{12}=0} =
 \frac{1}{2} u(0,0,{\bf y}), \quad x_{12} := |{\bf x}_1 - {\bf x}_2 |,
 \quad {\bf y} := \tfrac{1}{2}( {\bf x}_1 + {\bf x}_2) \neq {\bf 0} ,
\]
because asymptotic contributions from subspaces $P_n$, $n=1,2,\ldots$, vanish in the spherical average.

For the subspace of $l=1$ relative angular momentum, defined by $P_1$, one gets
\begin{eqnarray*}
 G{\cal P}_1 u & \sim & r \bigl( 1+ \tfrac{1}{4} \sqrt{2} t r \bigr) w_1(y) + \cdots \\
 & \sim & r \bigl( 1+ \tfrac{1}{4} |{\bf x}_1 - {\bf x}_2 | \bigr) w_1(y) + \cdots .
\end{eqnarray*}
The subspace is antisymmetric with respect to the interchange of the electrons, and according to Pauli's principle
it must correspond to a triplet state with respect to the total spin of the wavefunction. This is also reflected 
by the fact that triplet wavefunctions vanish at the $e-e$ edge. 

The leading order asymptotic behaviour in the $l=0$ case near one of the two $e-n$ edges gives
\[
 G{\cal P}_0 u \sim \bigl( 1 - Z |{\bf x}_1| \bigr) w_0(y) + \cdots,
\]
which is again equivalent to the corresponding variant of Kato's ``cusp'' condition
\[
 \left. \frac{1}{4\pi} \frac{\partial \int_{S^2} u(x_1,\omega_1,{\bf x}_2) d\omega_1}{\partial x_1} \right|_{x_1=0} =
 -Z u(0,0,{\bf x}_2), \quad x_1 := |{\bf x}_1|, \quad {\bf x}_2 \neq {\bf 0} .
\]

The corresponding asymptotic expression in the $l=1$ case becomes
\[
 G{\cal P}_1 u \sim r \bigl( 1- \tfrac{1}{2} Z |{\bf x}_1| \bigr) w_1(y)
 + \cdots . 
\]

In the limit $t \rightarrow \infty$ one electron approaches the nucleus whereas the other electron is far appart of it.
This situation resembles to the single electron case and the asymptotic behaviour can be compared with
the hydrogen series. For the subspace $l=0$ one gets
\[
 G{\cal P}_0 u \sim \bigl( 1 - Z |{\bf x}_1| + \tfrac{1}{3} |{\bf x}_1|^2 (Z^2-E+\tfrac{1}{2} \partial^2_t) \bigr) \tilde{w}_0(y)
 + \cdots \quad \mbox{for} \ t \rightarrow \infty .
\]
This is formally equivalent to the asymptotic expansion of the Green operator for the hydrogen series,
cf.~\cite{FHSS10}. However, let us mention that the second order term for $l=0$ depends on the eigenvalue of the energy which is different
for bound states in the helium and hydrogen series. Only in the limitting case of highly excited states 
in the helium series the energies approach the ionization threshold, i.e., the ground state energy of the corresponding ion. 
The additional differential operator $\tfrac{1}{2} \partial^2_t$ acts as a correction term which can be seen by considering 
the case of two noninteracting electrons, cf.~Appendix ref{testcal}.

Let us first consider a system consisting of two electrons and two nuclei $A,B$ with charges $Z_a, Z_b$. For simplicity, the nuclei
are arranged at distance $R$ along the z-axis at positions $\frac{1}{2} R {\bf e}_z$ and $-\frac{1}{2} R {\bf e}_z$, respectively.
We introduce the local distance variables
\[
 {\bf r}_1 := {\bf x}_1 - \tfrac{1}{2} R {\bf e}_z , \quad {\bf r}_2 := {\bf x}_2 + \tfrac{1}{2} R {\bf e}_z ,
\]
and express ${\bf r}_1, {\bf r}_2$ in terms of hyperspherical coordinates. With this the potential becomes
\begin{eqnarray*}
 v_{e-n} & := & -\frac{Z_a r}{\sin r} -\frac{Z_b r}{\cos r} +\frac{Z_a Z_b rt}{R} \\
 & & -\frac{Z_b r}{\sqrt{\sin^2r +2\sin r\cos\theta_1 \tfrac{R}{t} + \bigl( \tfrac{R}{t} \bigr)^2}} 
 -\frac{Z_a r}{\sqrt{\cos^2r -2\cos r\cos\theta_2 \tfrac{R}{t} + \bigl( \tfrac{R}{t} \bigr)^2}} \\
 & & +\frac{r}{\sqrt{1-\sin(2r) \bigl( \cos \theta_1 \cos \theta_2 + \sin \theta_1 \sin \theta_2 \cos(\phi_1-\phi_2) \bigr)
 +2(\sin r \cos \theta_1 - \cos r \cos \theta_2)\tfrac{R}{t} + + \bigl( \tfrac{R}{t} \bigr)^2}} .
\end{eqnarray*}
The asymptotic expansion becomes
\[
 v_{e-n} \sim \underbrace{-Z_a}_{=Z_1} + \biggl[ \underbrace{-Z_b +\frac{Z_a Z_b t}{R} -\frac{Z_b t}{R}
 -\frac{Z_a}{\sqrt{1-2 \cos \theta_2 \tfrac{R}{t} + \bigl( \tfrac{R}{t} \bigr)^2}}
 +\frac{1}{\sqrt{1-2 \cos \theta_2 \tfrac{R}{t} + \bigl( \tfrac{R}{t} \bigr)^2}}}_{=Z_2+tE} \biggr] r
 +{\cal O}(r^2) ,
\]
where $R>t$ has been assumed.
For the special case of the hydrogen molecule, i.e., $Z_a=Z_b=1$, one gets
\[
 v_{e-n} \sim -1-r+{\cal O}(r^2) ,
\]
and 
\[
 Z_2 = -tE-1=-t(2E_0+E_{VdW})-1 , 
\]
where we have decomposed the total energy $E$ into a nonintercting part $E_0$ and the Van der Waals energy $E_{VdW}$.
This shows that the Van der Waals interaction of a distant atom enters already in the second order term of the asymptotic expansion
of the wavefunction at a nucleus.

\subsection{Absence of logarithmic terms in the asymptotics of eigenfunctions near edges}
\label{nolog}
The qualitative behaviour of eigenfunctions of Hamiltonian operators near edges of the stratified configuration space,
i.e., at coalscence points of two particles, can be
easily derived from a fundamental theorem of Fournais 
et al.~\cite{FHO2S09}, already mentioned in the introduction.  
There it has been shown in a neighbourhood
of a coalescence point of two particles, with respective coordinates $x_i,x_j \in \mathbb{R}^3$ in configuration space, that 
the wavefunction can be represented
in the form $\Psi = \Psi_1 + |x_i - x_j| \Psi_2$ with $\Psi_1$ and $\Psi_2$ real analytic functions of the particle 
coordinates. This result implies an asymptotic expansion 
in terms of non negative integer powers of the interparticle distance.
In particular there can be no logarithmic terms, which show up
in the neighbourhood of coalescence points of three and more particles, cf.~Fock's expansion of eigenfunctions of the helium atom. 

Within the present work the asymptotic type is encoded in the
meromorphic Mellin symbol of the paramatrix. The location of its poles and corresponding multiplicities determines the asymptotic
behaviour. In particular, the absence of logarithmic terms
requires that no multiple zeros appear in the denominator
of asymptotic symbols $d^{n}_{k,j}$, with $j,k,n \in \mathbb{N}_0$. 
We do not want to proof this property in full generality,
instead we restrict ourselves to look at those symbols 
which have been already evaluated in the course of our calculation.
In Appendix \ref{poles}, we have summarized our state of knowledge
concerning the location and multiplicities of poles 
of some asymptotic symbols. It is interesting to see that multiple
poles appear in these symbols at values $w=3,4$, however, their location in the complex plane is to the right of the integration contour $\Gamma_{\frac{7}{2}-\gamma}$ of 
the inverse Mellin transformation which enters into the Mellin pseudo-differential operator part of the parametrix. In our particular application, the integration contour $\Gamma_{\frac{7}{2}-\gamma}$ restricts to the strip of the complex 
plane $\{ z \in \mathbb{C} : 2 < |z| < 3 \}$. This is because we consider the
Hamiltonian as an unbounded essentially self-adjoint operator on $L_{2}(\mathbb{R}^{6})$,
which has its natural domain $H^{2}(\mathbb{R}^{6})$, cf.~\cite{RS4}. Therefore, all eigenfunctions are bounded
and $\gamma$ can be restricted to the intervall of ellipticity, i.e., $\tfrac{1}{2}<\gamma<\tfrac{3}{2}$, cf.~Remark \ref{Hamellipt}. 

\section{Proof of the main result}
\label{4}
\subsection{Calculation of the asymptotic parametrix up to second order} 
\label{4.1}
Let us take the ansatz (\ref{Pap}) for the construction of the asymptotic parametrix, where the pseudo-differential operator structure of $P_i$ for
$i>0$ is completely analogous to $P_0$ given by (\ref{A0}), (\ref{p0}) and (\ref{pM0}), respectively. Therefore, we accordingly define the parameter dependent Mellin part by
\[
 p_{M,i}(y,\eta) := r^{i}\opm_M^{\gamma-3}(a_i^{(-1)})(y,\eta) \ \mbox{for} \ i=0,1,2,\ldots,
\]
where the Mellin symbols $a_i^{(-1)}$, $i=0,1,2,\ldots$, have asymptotic expansions of the form
\begin{eqnarray}
\label{aiasymp}
 a_i^{(-1)}(y,\eta) & \sim & -2t^2 \bigg( q_{i,0} +r q_{i,1} +r^2 q_{i,2} + \cdots\bigg) \\ \nonumber 
 & \sim & -2t^2 \sum_{n \geq 0} \sum_{j \geq 0} r^n d_{n,j}^{(i)}(w,r\eta) .
\end{eqnarray} 
It is easy to see, that this particular expansion satisfies the  
formal requirements of an asymptotic parametrix construction outlined in Section \ref{oooa}. Let us just mention the conormal
symbols which satisfy
\[
 \sigma^{-2-j}_c(p_{M,i}) (y,\eta)=0 , \quad j=0,1,\ldots,i-1
\]
and
\[
 \sigma^{-2-i}_c(p_{M,i})(y,\eta)= -2t^{2} d_{0,0}^{(i)} (y,\eta) \neq 0 .
\]
Like in Proposition \ref{propcon}, we can express
Mellin symbols via an asymptotic expansion in terms of
conormal symbols, i.e.,
\[
 r^{i}a_i^{(-1)}(y,\eta) \sim \sum_{j \geq i} r^j \sigma^{-2-j}_c(p_{M,i}) (y,\eta) ,
\]
and vice versa conormal symbols via our homogeneous symbols, i.e.,
\[
 \sigma^{-2-j}_c(p_{M,i}) (y,\eta) = -2t^2 \sum_{k=0}^{j-i} d_{j-i-k,k}^{(i)}(w,\eta) .
\]
It has been already mentioned in the remark following Proposition \ref{propcon}, that the asymptotic expansion (\ref{aiasymp}) does not depend on the details of the parametrix construction.
In the following we therefore calculate all symbols of the asymptotic parametrix by applying it from the right side to the Hamiltonian operator.

\subsubsection{Initial parametrix $P_0$}
In the initial step of the parametrix construction we consider
the equation $P_0A_0=I$ modulo Green operators. The symbol $a_0^{-1}$ of the initial parametrix $P_0$ has itself an asymptotic expansion 
(\ref{a0asymp}) and the corresponding zeroth order equation (\ref{eqr0}), derived from (\ref{LMright}), has been already given.
Let us start the actual calculations by inserting the asymptotic expansion
\begin{equation}
  q_{0,0} = -2t^2 \bigl( d^{(0)}_{0,0} + d^{(0)}_{0,1} + \cdots \bigr)
\label{q00} 
\end{equation}
into the zeroth order equation (\ref{eqr0}).
A simple calculation yields
\[
 d^{(0)}_{0,2j} = (r^2 C_0)^j b_j \ \ \mbox{and} \ 
 d^{(0)}_{0,2j+1} = 0, \quad j=0,1,\ldots ,
\]
with
\begin{equation}
\begin{split}
 b_0 & = h_0^{-1}, \\
 b_1 & = b_0/\bigl( h_0-2(2w-7) \bigr), \\
 \vdots \\
 b_n & = b_{n-1}/\bigl( h_0-2n(2w-5-2n) \bigr), \\
 \vdots
\end{split}
\label{bn}
\end{equation}

The first order equation, derived from (\ref{LMright}), becomes
\begin{eqnarray*}
 r^1: \ \ T^{-2} a_0 r q_{0,1} &+& \partial_w T^{-2} a_0 (-r \partial_r) r q_{0,1}
 + \tfrac{1}{2} \partial^2_w T^{-2} a_0 (-r \partial_r)^2 r q_{0,1} \\
 &+& \partial_\tau T^{-2} a_0 D_t q_{0,0} + \partial_{\Theta_2} T^{-2} a_0 D_{\theta_2} q_{0,0} =0 .
\end{eqnarray*} 
Inserting the asymptotic expansion
\begin{equation}
  q_{0,1} = -2t^2 \bigl( d^{(0)}_{1,0} + d^{(0)}_{1,1} + \cdots \bigr) 
\label{q01} 
\end{equation}
as well as (\ref{q00}) yields after a simple calculation
\begin{equation}
\begin{split}
 d^{(0)}_{1,2j} & = 0 \ \ \mbox{for} \ j=0,1,\ldots, \\
 d^{(0)}_{1,2j+1}& =-\sum_{k=0}^ j (r^2 C)^{j-k} P_{1,2k+1}b_{j+1} \ \ \mbox{for} \  j=0,1,\ldots \\
 \vdots 
\end{split}
\label{d1j}
\end{equation}
with polynomials
\[
 P_{1,2j+1} := 4i \biggr[ t(r\tau) (r^2 C)^j +j \biggr( t^3 (r\tau)^3 - (r\Theta_2) (r\Phi_2)^2
 \tfrac{\cot \theta_2}{\sin^2 \theta_2} \biggr) (r^2 C)^{j-1} \biggr] ,
\]
which are homogeneous of order $2j+1$ in the edge covariables.

The recursive calculation can be carried out to any order, what remains to be done for our purpose is to consider the second order equation, derived from (\ref{LMright}), which becomes 
\begin{multline*}
 r^2: \ \ T^{-2} a_0 r^2 q_{0,2} + \partial_w T^{-2} a_0 (-r \partial_r) r^2 q_{0,2}
 + \tfrac{1}{2} \partial^2_w T^{-2} a_0 (-r \partial_r)^2 r^2 q_{0,2} \\
 + \partial_\tau T^{-2} a_0 D_t r q_{0,1} + \partial_{\Theta_2} T^{-2} a_0 D_{\theta_2} r q_{0,1}
 + \tfrac{1}{2} \partial^2_\tau T^{-2} a_0 D_t^2 q_{0,0}
 + \tfrac{1}{2} \partial^2_{\Theta_2} T^{-2} a_0 D_{\theta_2}^2 q_{0,0} =0 ,
\end{multline*}
from which $q_{0,2}$ can be calculated. In particular, we get
\begin{equation}
\begin{split}
 d^{(0)}_{2,2j+1} & = 0 \ \ \mbox{for} \ j=0,1,\ldots, \\
 d^{(0)}_{2,0}& = -2b_1 \\
 d^{(0)}_{2,2}& = -P_{2,2} \, b_1/(h_0-4(2w-9))\\
 \vdots 
\end{split}
\label{d2j}
\end{equation}
with 
\[
 P_{2,2} := 34t^2 (r\tau)^2+\tfrac{2(1+2\cos^2\theta_2)}{\sin^4\theta_2}(r\Phi_2)^2+4r^2C_0 . 
\]

\subsubsection{First order parametrix $P_1$}
Once we have calculated the symbol of the initial parametrix 
up to second order in the asymptotic expansion, it can be used
in the first recursion step to construct the symbol of the 
first order parametrix $P_1$ from the equation
\begin{eqnarray}
\nonumber
 \lefteqn{\sum_\alpha \tfrac{1}{\alpha !} 
 r^{-2} \partial_\eta^\alpha \bigl(\opm_M^{\gamma-1}(a_0)\bigr)
 r^3 D^\alpha_y \bigl(\opm_M^{\gamma-3}(a_1^{(-1)})\bigr)} \\   \nonumber
 & & + \sum_\alpha \tfrac{1}{\alpha !} 
 r^{-1} \partial_\eta^\alpha \bigl(\opm_M^{\gamma-1}(a_1)\bigr)
 r^2 D^\alpha_y \bigl(\opm_M^{\gamma-3}(a_0^{(-1)})\bigr) \\ \label{P1r1}
 & = & \sum_\alpha \tfrac{1}{\alpha !} \opm_M^{\gamma-3} \bigl(
 \partial_\eta^\alpha T^{-2} a_0 \#_{r,w} r D^\alpha_y a_1^{(-1)} \bigr) \\ \nonumber
  & & + \sum_\alpha \tfrac{1}{\alpha !} \opm_M^{\gamma-3} \bigl(
 r \partial_\eta^\alpha T^{-2} a_1 \#_{r,w} D^\alpha_y a_0^{(-1)} \bigr) \\ \nonumber
& = & 0 \quad \mbox{mod} \ R^0_{\gflat}(Y \times \mathbb{R}^3, {\boldsymbol g}_r) .
\end{eqnarray}
From this equation we can derive the first order equation of the
corresponding symbol
\begin{eqnarray*}
 \lefteqn{r^1: \ \ T^{-2} a_0 r q_{1,0} + \partial_w T^{-2} a_0 (-r \partial_r) (r q_{1,0})
 + \tfrac{1}{2} \partial^2_w T^{-2} a_0 (-r \partial_r)^2 (r q_{1,0})} \\
 & & \quad \quad + r T^{-2} a_1 q_{0,0} + r \partial_w T^{-2} a_1 (-r \partial_r) q_{0,0} \sim 0 .
\end{eqnarray*}
Inserting the asymptotic expansions
\begin{equation}
 q_{1,0} = -2t^2 \bigl( d^{(1)}_{0,0} + d^{(1)}_{0,1} + \cdots \bigr) ,
\label{q10}
\end{equation}
and (\ref{q00}) which has been already calculated up to second order in the previous subsection, one gets the asymptotic equation
\[
 \bigl( h_0 -r^2C_0 -(2w-6) \bigr) \sum_{j \geq 0} d^{(1)}_{0,j} - (2w-7) \sum_{j \geq 1} j d^{(1)}_{0,j}
 + \sum_{j \geq 1} j^2 d^{(1)}_{0,j} - \bigl( irC_1 +2tZ_1 \bigr) \sum_{j \geq 0} d^{(0)}_{0,j}
 \sim 0 ,
\]
from which we obtain, with respect to powers of $r\eta$, in zeroth order
\[
 (r\eta)^0: \ \ d^{(1)}_{0,0} = \frac{2tZ_1}{h_0 \bigl( h_0 -(2w-6) \bigr)}
\]
and first order
\[
 (r\eta)^1: \ \ d^{(1)}_{0,1} = \frac{irC_1}{h_0 \bigl( h_0 -2(2w-7) \bigr)} ,
\]
respectively. 

The second order asymptotic parametrix requires one more term
which can be obtained from the second order contribution to Eq.~(\ref{P1r1}), i.e.,
\begin{eqnarray*}
 \lefteqn{r^2: \ \ T^{-2} a_0 r^2 q_{1,1} + \partial_w T^{-2} a_0 (-r \partial_r) (r^2 q_{1,1})
 + \tfrac{1}{2} \partial^2_w T^{-2} a_0 (-r \partial_r)^2 (r^2 q_{1,1})} \\
 & & \quad \quad + r T^{-2} a_1 r q_{0,1} + r \partial_w T^{-2} a_1 (-r \partial_r) (r q_{0,1})
 + \partial_\eta T^{-2} a_0 D_y (r q_{1,0}) + r \partial_\eta T^{-2} a_1 D_y q_{0,0} \sim 0 .
\end{eqnarray*}
Inserting the asymptotic expansion
\begin{equation}
  q_{1,1} = -2t^2 \bigl( d^{(1)}_{1,0} + d^{(1)}_{1,1} + \cdots \bigr) 
\label{q11} 
\end{equation}
as well as the already previously considered asymptotic expansions
(\ref{q00}), (\ref{q01}) and (\ref{q10}), we get
\begin{eqnarray*}
 \lefteqn{\bigl( h_0 -r^2C_0 -2(2w-7) \bigr) \sum_{j \geq 0} d^{(1)}_{1,j} - (2w-9) \sum_{j \geq 1} j d^{(1)}_{1,j}
 + \sum_{j \geq 1} j^2 d^{(1)}_{1,j}} \\ & &  \quad \quad
 - \bigl( irC_1 +2tZ_1 \bigr) \sum_{j \geq 0} d^{(0)}_{1,j} 
 +r \tau \sum_{j \geq 0} D_t (-2t^2 d^{(1)}_{0,j}) \\ & &  \quad \quad
 -2 r \Theta \sum_{j \geq 0} D_{\theta_2} d^{(1)}_{0,j} -2 r \Phi \tfrac{1}{\sin^2 \theta_2} \sum_{j \geq 0} D_{\phi_2} d^{(1)}_{0,j}
 - \tfrac{i5}{2t} \sum_{j \geq 0} D_t (-2t^2 d^{(0)}_{0,j}) \\ & &  \quad \quad
 +i \cot \theta_2 \sum_{j \geq 0} D_{\theta_2} d^{(0)}_{0,j} \sim 0 ,
\end{eqnarray*}
from which we obtain, with respect to powers of $r\eta$, in zeroth order
\[
 (r\eta)^0: \ \ d^{(1)}_{1,0} = - \frac{10}{h_0 \bigl( h_0 -2(2w-7) \bigr)} .
\]

\subsubsection{Second order parametrix $P_2$}
Finally, we want to use the symbols of the initial and first order
parametrix to calculate in the second recursion step, the remaining symbol of the second order parametrix $P_2$ from the equation
\begin{eqnarray*}
 \lefteqn{\sum_\alpha \tfrac{1}{\alpha !} 
 r^{-2} \partial_\eta^\alpha \bigl(\opm_M^{\gamma-1}(a_0)\bigr)
 r^4 D^\alpha_y \bigl(\opm_M^{\gamma-3}(a_2^{(-1)})\bigr)} \\
 & & + \sum_\alpha \tfrac{1}{\alpha !} 
 r^{-1} \partial_\eta^\alpha \bigl(\opm_M^{\gamma-1}(a_1)\bigr)
 r^3 D^\alpha_y \bigl(\opm_M^{\gamma-3}(a_1^{(-1)})\bigr) \\
 & & + \sum_\alpha \tfrac{1}{\alpha !} 
 \partial_\eta^\alpha \bigl(\opm_M^{\gamma-1}(a_2)\bigr)
 r^2 D^\alpha_y \bigl(\opm_M^{\gamma-3}(a_0^{(-1)})\bigr) \\
 & = & \sum_\alpha \tfrac{1}{\alpha !} \opm_M^{\gamma-3} \bigl(
 \partial_\eta^\alpha T^{-2} a_0 \#_{r,w} r^2 D^\alpha_y a_2^{(-1)} \bigr) \\
  & & + \sum_\alpha \tfrac{1}{\alpha !} \opm_M^{\gamma-3} \bigl(
 r \partial_\eta^\alpha T^{-2} a_1 \#_{r,w} r D^\alpha_y a_1^{(-1)} \bigr) \\
  & & + \sum_\alpha \tfrac{1}{\alpha !} \opm_M^{\gamma-3} \bigl(
 r^2 \partial_\eta^\alpha T^{-2} a_2 \#_{r,w} D^\alpha_y a_0^{(-1)} \bigr) \\
& = & 0 \quad \mbox{mod} \ R^0_{\gflat}(Y \times \mathbb{R}^3, {\boldsymbol g}_r).
\end{eqnarray*}
Like before, we derive from it the second order equation of the 
corresponding symbol
\begin{eqnarray*}
 \lefteqn{r^2: \ \ T^{-2} a_0 r^2 q_{2,0} + \partial_w T^{-2} a_0 (-r \partial_r) (r^2 q_{2,0})
 + \tfrac{1}{2} \partial^2_w T^{-2} a_0 (-r \partial_r)^2 (r^2 q_{2,0})} \\
 & & \quad \quad + r T^{-2} a_1 r q_{1,0} + r \partial_w T^{-2} a_1 (-r \partial_r) (r q_{1,0})
 + r^2 T^{-2} a_2 q_{0,0} + r^2 \partial_w T^{-2} a_2 (-r \partial_r) q_{0,0} \sim 0 .
\end{eqnarray*}
Inserting the asymptotic expansion
\begin{equation}
  q_{2,0} = -2t^2 \bigl( d^{(2)}_{0,0} + d^{(2)}_{0,1} + \cdots \bigr) 
\label{q20} 
\end{equation}
as well as the already previously considered asymptotic expansions
(\ref{q00}) and (\ref{q10}), we get
\begin{eqnarray*}
 \lefteqn{\bigl( h_0 -r^2C_0 -2(2w-7) \bigr) \sum_{j \geq 0} d^{(2)}_{0,j} - (2w-9) \sum_{j \geq 1} j d^{(2)}_{0,j}
 + \sum_{j \geq 1} j^2 d^{(2)}_{0,j}} \\ & & \quad \quad - \bigl( irC_1 +2tZ_1 \bigr) \sum_{j \geq 0} d^{(1)}_{0,j}
 + \bigl( \tfrac{8}{3} (w-2) + \tfrac{1}{3} \Delta_{S^2} -r^2 C_2 -2tZ_2 \bigr)
 \sum_{j \geq 0} d^{(0)}_{0,j} - \tfrac{8}{3} \sum_{j \geq 1} j d^{(0)}_{0,j} \sim 0 ,
\end{eqnarray*}
from which we obtain, with respect to powers of $r\eta$, in zeroth order
\[
 (r\eta)^0: \ \ d^{(2)}_{0,0} = \frac{(2tZ_1)^2}
 {\bigl( h_0 -2(2w-7) \bigr) \bigl( h_0 -(2w-6) \bigr) h_0}
 - \frac{1}{3} \frac{8(w-2) + \Delta_{S^2} -6tZ_2}{\bigl( h_0 -2(2w-7) \bigr) h_0} . 
\]
\subsubsection{Asymptotic expansion of the Mellin symbols up to second order}
Summing up the asymptotic terms of the previous calculations
up to second order leads to the following explicit expressions
for the Mellin symbols of the parametrix
\begin{align*}
 a_0^{(-1)} & \sim q_{0,0} +rq_{0,1} +r^2 q_{0,2} + \cdots \\
 & \sim -2t^2 \biggl( d_{0,0}^{(0)} + d_{0,2}^{(0)} + r d_{1,1}^{(0)} + r^2 d_{2,0}^{(0)} + \cdots \biggr) \\ 
 & \sim -2t^2 \biggl( \frac{1}{h_0} + \frac{r^2C_0}{h_0 \bigl( h_0-2(2w-7) \bigr)} - \frac{r P_{1,1}}{h_0 \bigl( h_0-2(2w-7) \bigr)}
 - r^2 \frac{2}{h_0 \bigl( h_0-2(2w-7) \bigr)} + \cdots \biggr) ,
\end{align*}

\begin{align*}
 a_1^{(-1)} & \sim q_{1,0} +rq_{1,1} + \cdots \\
 & \sim -2t^2 \biggl( d_{0,0}^{(1)} + d_{0,1}^{(1)} + r d_{1,0}^{(1)} + \cdots \biggr) \\
 & \sim -2t^2 \biggl( \frac{2tZ_1}{h_0 \bigl( h_0 -(2w-6) \bigr)} + \frac{irC_1}{h_0 \bigl( h_0 -2(2w-7) \bigr)} - r \frac{10}{h_0 \bigl( h_0 -2(2w-7) \bigr)}
 + \cdots \biggr) ,
\end{align*}

\begin{align*}
 a_2^{(-1)} & \sim q_{2,0} + \cdots \\
 & \sim -2t^2 \biggl( d_{0,0}^{(2)} + \cdots \biggr) \\
 & \sim -2t^2 \biggl( \frac{(2tZ_1)^2} {\bigl( h_0 -2(2w-7) \bigr) \bigl( h_0 -(2w-6) \bigr) h_0} - \frac{1}{3} \frac{8(w-2) + \Delta_{S^2} -6tZ_2}{\bigl( h_0 -2(2w-7) \bigr) h_0} + \cdots \biggr) .
\end{align*}

\subsection{Calculation of left Green operators up to second order}
\label{4.2}
In Section \ref{2.2}, we have discussed the construction of the 
initial parametrix and derived the corresponding Green operators.
The recursive construction of higher order parametrices proceeds
in a similar manner leading to the same types of Green operator
which can be grouped in two classes.
Let us first consider the class of type-$a$ Green operator symbols which are of the general form
\begin{equation}
 g_{a} := \sigma' \omega'_{1,\eta} r^{k+l} \sum_\alpha \frac{1}{\alpha !} 
 \partial_\eta^\alpha \biggl(\opm_M^{\gamma-1}(T^2 a_k^{(-1)})\biggr)
 \bigl(\tilde{\sigma}' \sigma -1 \bigr)
 D^\alpha_y \biggl(\opm_M^{\gamma-1}(a_l)\biggr) 
 \omega_{0,\eta} \tilde{\sigma}, \quad k,l=0,1,2,\ldots ,
\label{gtypa}
\end{equation}
and belong to $R^{-\infty}_G(Y \times \mathbb{R}^3, {\boldsymbol g})$. The second class of type-$b$ Green operator symbols are of the general form
\begin{equation}
 g_{b} := \sigma' \omega'_{1,\eta} r^{k+l} \sum_\alpha \frac{1}{\alpha !} 
 D^\alpha_y(\tfrac{C}{2t^2}) \partial_\eta^\alpha \biggl( \opm_M^{\gamma-3}(a_k^{(-1)}) - \opm_M^{\gamma-1}(a_k^{(-1)}) \biggr)
 \omega_{0,\eta} \tilde{\sigma} ,
\label{gtypb}
\end{equation}
and belong to $R^{0}_G(Y \times \mathbb{R}^3, {\boldsymbol g})$.
They arise from the commutation of a $r^{l}$ term, which belongs to a homogeneous polynomial in the covariables $C_{i}(\eta)$, $i=0,1,\ldots$, from the right to the left, cf.~our discussion in Section \ref{oooa}.

\subsubsection{Green operators from initial step of the parametrix construction}

According to the calculations in Section \ref{2.2} we get a
type-$a$ Green operator symbol, cf.~(\ref{g012}), which after insertion
of the corresponding asymptotic symbol of the initial parametrix
becomes
\begin{eqnarray*}
 g_{0,1} & = & \sigma' \omega'_{1,\eta} \sum_\alpha \tfrac{1}{\alpha !} 
 \opm_M^{\gamma-1}\bigl( (-2t^2) T^2 \partial_\eta^\alpha (d^{(0)}_{0,0} + d^{(0)}_{0,2} + r d^{(0)}_{1,1} + r^2 d^{(0)}_{2,0}) \bigr)
 \bigl(\tilde{\sigma}' \sigma -1 \bigr)
 \opm_M^{\gamma-1} \bigl( D^\alpha_y a_0 \bigr) \tilde{\sigma} .
\end{eqnarray*}
For the type-$b$ Green operator symbol, cf. (\ref{g012}), we get
the explicit formula 
\begin{eqnarray*}
 g_{0,2} & = & \sigma' \omega'_{1,\eta} \sum_\alpha \tfrac{1}{\alpha !} 
 D^\alpha_y(\tfrac{r^2 C_0}{2t^2}) \biggl[ \opm_M^{\gamma-3}\bigl( (-2t^2) \partial_\eta^\alpha
 (d^{(0)}_{0,0} + d^{(0)}_{0,2} + r d^{(0)}_{1,1} + r^2 d^{(0)}_{2,0}) \bigr) \\ 
 & & - \opm_M^{\gamma-1} \bigl( (-2t^2) \partial_\eta^\alpha (d^{(0)}_{0,0} + d^{(0)}_{0,2} + r d^{(0)}_{1,1} + r^2 d^{(0)}_{2,0}) \bigr)
 \biggr] \tilde{\sigma} \\
 & = & -2t^2 \sigma' \omega'_{1,\eta} \biggl[ r \res (d^{(0)}_{0,0},1) M \bigl( \tfrac{C_0}{2t^2} \tilde{\sigma} (\cdot) \bigr) (1)
 + \res (d^{(0)}_{0,0},2) M \bigl( \tfrac{C_0}{2t^2} \tilde{\sigma} (\cdot) \bigr) (2) \\
 & & + \sum_\alpha \tfrac{1}{\alpha !} \bigl( \res (\partial_\eta^\alpha d^{(0)}_{0,2},2)
 + r \res (\partial_\eta^\alpha d^{(0)}_{1,1},2)
 + r^2 \res (\partial_\eta^\alpha d^{(0)}_{2,0},2) \bigr)
 M \bigl( D^\alpha_y(\tfrac{C_0}{2t^2}) \tilde{\sigma} (\cdot) \bigr) (2) + {\cal O}(r^3) \biggr] ,
\end{eqnarray*}
where we applied Cauchy's residue theorem in order to get an explicit formula. Here and in the following $\res(f, w_0)$
denotes the residuum of the meromophic function $f$ at its
pole $w_0$.

\subsubsection{Green operators from first recursion step of the parametrix construction}

The first order contribution to the asymptotic parametrix from the left side is given by
\begin{eqnarray*}
 \lefteqn{\sigma' \omega'_{1,\eta} r \sum_\alpha \tfrac{1}{\alpha !} 
 \partial_\eta^\alpha \bigl(\opm_M^{\gamma-1}(T^2 a_1^{(-1)})\bigr) \tilde{\sigma}' 
 \sigma D^\alpha_y \bigl(\opm_M^{\gamma-1}(a_0)\bigr) \tilde{\sigma}} \\
 & & + \sigma' \omega'_{1,\eta} r \sum_\alpha \tfrac{1}{\alpha !}
 \partial_\eta^\alpha \bigl(\opm_M^{\gamma-2}(T a_0^{(-1)})\bigr) \tilde{\sigma}' 
 \sigma D^\alpha_y \bigl(\opm_M^{\gamma-1}(a_1)\bigr) \tilde{\sigma} \\
 & = & \sigma' \omega'_{1,\eta} r \sum_\alpha \tfrac{1}{\alpha !}
 \opm_M^{\gamma-1} \bigl( D^\alpha_y(-\tfrac{1}{2t^2}) T^2 \partial_\eta^\alpha
 (a_1^{(-1)} h_0) + D^\alpha_y(\tfrac{r^2 C_0}{2t^2}) \partial_\eta^\alpha a_1^{(-1)} \\
 & & + T \partial_\eta^\alpha a_0^{(-1)} D^\alpha_y(\tfrac{1}{t} Z_1)
 + D^\alpha_y(\tfrac{ir C_1}{2t^2}) \partial_\eta^\alpha a_0^{(-1)} \bigr) \tilde{\sigma} + g_{1,1} + g_{1,2}
\end{eqnarray*}
which has been converted into a single Mellin operator symbol modulo a type-$a$ and $b$ Green operator symbol, in the same manner as it has been discussed for the initial left parametrix in Section \ref{2.2}. After insertion
of the corresponding asymptotic symbols of the initial and first order parametrix, the type-$a$ Green operator symbol becomes
\begin{eqnarray*}
 g_{1,1} & = & \sigma' \omega'_{1,\eta} r \sum_\alpha \tfrac{1}{\alpha !} 
 \opm_M^{\gamma-1}\bigl( (-2t^2) T^2 \partial_\eta^\alpha (d^{(1)}_{0,0} + d^{(1)}_{0,1} + r d^{(1)}_{1,0}) \bigr)
 \bigl(\tilde{\sigma}' \sigma -1 \bigr)
 \opm_M^{\gamma-1} \bigl( D^\alpha_y a_0 \bigr) \tilde{\sigma} \\
 & & + \sigma' \omega'_{1,\eta} \sum_\alpha \tfrac{1}{\alpha !}
 \opm_M^{\gamma-1}\bigl( (-2t^2) T^2 \partial_\eta^\alpha (d^{(0)}_{0,0} + d^{(0)}_{0,2} + r d^{(0)}_{1,1} + r^2 d^{(0)}_{2,0}) \bigr)
 \bigl(\tilde{\sigma}' \sigma -1 \bigr)
 \opm_M^{\gamma-1} \bigl( r' D^\alpha_y a_1 \bigr) \tilde{\sigma} .
\end{eqnarray*}
The corresponding type-$b$ Green operator symbol is given by
\begin{eqnarray*}
 g_{1,2} & = & \sigma' \omega'_{1,\eta} r \sum_\alpha \tfrac{1}{\alpha !} 
 D^\alpha_y(\tfrac{r^2 C_0}{2t^2}) \biggl[ \opm_M^{\gamma-3}\bigl( (-2t^2) \partial_\eta^\alpha
 (d^{(1)}_{0,0} + d^{(1)}_{0,1} + r d^{(1)}_{1,0}) \bigr) \\ 
 & & - \opm_M^{\gamma-1} \bigl( (-2t^2) \partial_\eta^\alpha (d^{(1)}_{0,0} + d^{(1)}_{0,1} + r d^{(1)}_{1,0}) \bigr)
 \biggr] \tilde{\sigma} \\
  & & + \sigma' \omega'_{1,\eta} r \sum_\alpha \tfrac{1}{\alpha !} 
 \biggl[ \opm_M^{\gamma-2}\bigl( (-2t^2) T \partial_\eta^\alpha 
 (d^{(0)}_{0,0} + d^{(0)}_{0,2} + r d^{(0)}_{1,1} + r^2 d^{(0)}_{2,0}) \bigr) \\ 
 & & - \opm_M^{\gamma-1} \bigl( (-2t^2) T \partial_\eta^\alpha (d^{(0)}_{0,0} + d^{(0)}_{0,2} + r d^{(0)}_{1,1} + r^2 d^{(0)}_{2,0}) \bigr)
 \biggr] \opm_M^{\gamma-1} \bigl( D^\alpha_y( \tfrac{1}{t} Z_1 ) \bigr) \tilde{\sigma} \\
 & & + \sigma' \omega'_{1,\eta} r \sum_\alpha \tfrac{1}{\alpha !} 
 D^\alpha_y(\tfrac{ir C_1}{2t^2}) \biggl[ \opm_M^{\gamma-3}\bigl( (-2t^2) \partial_\eta^\alpha
 (d^{(0)}_{0,0} + d^{(0)}_{0,2} + r d^{(0)}_{1,1} + r^2 d^{(0)}_{2,0}) \bigr) \\ 
 & & - \opm_M^{\gamma-1} \bigl( (-2t^2) \partial_\eta^\alpha (d^{(0)}_{0,0} + d^{(0)}_{0,2} + r d^{(0)}_{1,1} + r^2 d^{(0)}_{2,0}) \bigr)
 \biggr] \tilde{\sigma} \\
 & = & -2t^2 \sigma' \omega'_{1,\eta} \biggl[ r^2 \res (d^{(1)}_{0,0},1) M \bigl( \tfrac{C_0}{2t^2} \tilde{\sigma} (\cdot) \bigr) (1)
 +r \res (d^{(1)}_{0,0},2) M \bigl( \tfrac{C_0}{2t^2} \tilde{\sigma} (\cdot) \bigr) (2) \\
 & & + \sum_\alpha \tfrac{1}{\alpha !} \bigl( r \res (\partial_\eta^\alpha d^{(1)}_{0,1},2)
 + r^2 \res (\partial_\eta^\alpha d^{(1)}_{1,0},2) \bigr) 
 M \bigl( D^\alpha_y(\tfrac{C_0}{2t^2}) \tilde{\sigma} (\cdot) \bigr) (2) \\
 & & + \res (T d^{(0)}_{0,0},1) M \bigl( \tfrac{1}{t} Z_1 \tilde{\sigma} (\cdot) \bigr) (1) 
 + \sum_\alpha \tfrac{1}{\alpha !} \bigl( \res (T \partial_\eta^\alpha d^{(0)}_{0,2},1) \\
 & & + r \res (T \partial_\eta^\alpha d^{(0)}_{1,1},1)
 + r^2 \res (T \partial_\eta^\alpha d^{(0)}_{2,0},1) \bigr)
 M \bigl( D^\alpha_y(\tfrac{1}{t} Z_1) \tilde{\sigma} (\cdot) \bigr) (1) \\
 & & +r \res (d^{(0)}_{0,0},1) M \bigl( \tfrac{iC_1}{2t^2} \tilde{\sigma} (\cdot) \bigr) (1)
 + \res (d^{(0)}_{0,0},2) M \bigl( \tfrac{iC_1}{2t^2} \tilde{\sigma} (\cdot) \bigr) (2) \\
 & & + \sum_\alpha \tfrac{1}{\alpha !} \bigl( \res (\partial_\eta^\alpha d^{(0)}_{0,2},2)
 + r \res (\partial_\eta^\alpha d^{(0)}_{1,1},2)
 + r^2 \res (\partial_\eta^\alpha d^{(0)}_{2,0},2) \bigr)
 M \bigl( D^\alpha_y(\tfrac{iC_1}{2t^2}) \tilde{\sigma} (\cdot) \bigr) (2) + {\cal O}(r^3) \biggr] ,
\end{eqnarray*}
where in the second step Cauchy's residue theorem has been applied. 

\subsubsection{Green operators from second recursion step of the parametrix construction}

The second order contribution to the asymptotic parametrix from the left side is given by
\begin{eqnarray*}
 \lefteqn{\sigma' \omega'_{1,\eta} r^2 \sum_\alpha \tfrac{1}{\alpha !} 
 \partial_\eta^\alpha \bigl(\opm_M^{\gamma-1}(T^2 a_2^{(-1)})\bigr) \tilde{\sigma}' 
 \sigma D^\alpha_y \bigl(\opm_M^{\gamma-1}(a_0)\bigr) \tilde{\sigma}} \\
 & & + \sigma' \omega'_{1,\eta} r^2 \sum_\alpha \tfrac{1}{\alpha !}
 \partial_\eta^\alpha \bigl(\opm_M^{\gamma-2}(T a_1^{(-1)})\bigr) \tilde{\sigma}' 
 \sigma D^\alpha_y \bigl(\opm_M^{\gamma-1}(a_1)\bigr) \tilde{\sigma} \\
 & & + \sigma' \omega'_{1,\eta} r^2 \sum_\alpha \tfrac{1}{\alpha !}
 \partial_\eta^\alpha \bigl(\opm_M^{\gamma-3}(a_0^{(-1)})\bigr) \tilde{\sigma}' 
 \sigma D^\alpha_y \bigl(\opm_M^{\gamma-1}(a_2)\bigr) \tilde{\sigma} \\
 & = & \sigma' \omega'_{1,\eta} r^2 \sum_\alpha \tfrac{1}{\alpha !}
 \opm_M^{\gamma-1} \bigl( D^\alpha_y(-\tfrac{1}{2t^2}) T^2 \partial_\eta^\alpha
 (a_2^{(-1)} h_0) + D^\alpha_y(\tfrac{r^2 C_0}{2t^2}) \partial_\eta^\alpha a_2^{(-1)} \\
 & & + T \partial_\eta^\alpha a_1^{(-1)} D^\alpha_y(\tfrac{1}{t} Z_1)
 + D^\alpha_y(\tfrac{ir C_1}{2t^2}) \partial_\eta^\alpha a_1^{(-1)} 
 + \partial_\eta^\alpha a_0^{(-1)} D^\alpha_y(-\tfrac{4}{3t^2} w - \tfrac{1}{6t^2} \Delta_{S^2} + \tfrac{1}{t} Z_2) \\
 & & + D^\alpha_y(\tfrac{r^2 C_2}{2t^2}) T^{-2} \partial_\eta^\alpha a_0^{(-1)} \bigr) \tilde{\sigma} + g_{2,1} + g_{2,2},
\end{eqnarray*}
which has been converted into a single Mellin operator symbol modulo a type-$a$ and $b$ Green operator symbol, similar to the previous recursion steps. After insertion
of the corresponding asymptotic symbols of the initial, first and second order parametrix, the type-$a$ Green operator symbol becomes
\begin{eqnarray*}
 g_{2,1} & = & \sigma' \omega'_{1,\eta} r^2 \sum_\alpha \tfrac{1}{\alpha !} 
 \opm_M^{\gamma-1}\bigl( (-2t^2) T^2 \partial_\eta^\alpha d^{(2)}_{0,0} \bigr)
 \bigl(\tilde{\sigma}' \sigma -1 \bigr)
 \opm_M^{\gamma-1} \bigl( D^\alpha_y a_0 \bigr) \tilde{\sigma} \\
 & & + \sigma' \omega'_{1,\eta} r \sum_\alpha \tfrac{1}{\alpha !} 
 \opm_M^{\gamma-1}\bigl( (-2t^2) T^2 \partial_\eta^\alpha (d^{(1)}_{0,0} + d^{(1)}_{0,1} + r d^{(1)}_{1,0}) \bigr)
 \bigl(\tilde{\sigma}' \sigma -1 \bigr)
 \opm_M^{\gamma-1} \bigl( r' D^\alpha_y a_1 \bigr) \tilde{\sigma} \\
 & & + \sigma' \omega'_{1,\eta} \sum_\alpha \tfrac{1}{\alpha !} 
 \opm_M^{\gamma-1}\bigl( (-2t^2) T^2 \partial_\eta^\alpha (d^{(0)}_{0,0} + d^{(0)}_{0,2} + r d^{(0)}_{1,1} + r^2 d^{(0)}_{2,0}) \bigr)
 \bigl(\tilde{\sigma}' \sigma -1 \bigr)
 \opm_M^{\gamma-1} \bigl( r'^2 D^\alpha_y a_2 \bigr) \tilde{\sigma} .
\end{eqnarray*}
The corresponding type-$b$ Green operator symbol is given by
\begin{eqnarray*}
 g_{2,2} & = & \sigma' \omega'_{1,\eta} r^2 \sum_\alpha \tfrac{1}{\alpha !} \biggl[
 \opm_M^{\gamma-3} \bigl( (-2t^2) \partial_\eta^\alpha (d^{(0)}_{0,0} + d^{(0)}_{0,2} + r d^{(0)}_{1,1} + r^2 d^{(0)}_{2,0}) \bigr) \\
 & & - \opm_M^{\gamma-1} \bigl( (-2t^2) \partial_\eta^\alpha (d^{(0)}_{0,0} + d^{(0)}_{0,2} + r d^{(0)}_{1,1} + r^2 d^{(0)}_{2,0}) \bigr)
 \biggr] \opm_M^{\gamma-1} \bigl( D^\alpha_y( -\tfrac{4}{3t^2} w - \tfrac{1}{6t^2} \Delta_{S^2} + \tfrac{1}{t} Z_2 ) \bigr) \tilde{\sigma} \\
 & & + \sigma' \omega'_{1,\eta} r^2 \sum_\alpha \tfrac{1}{\alpha !} 
 D^\alpha_y(\tfrac{r^2 C_2}{2t^2}) \biggl[ \opm_M^{\gamma-5}\bigl( (-2t^2) T^{-2} \partial_\eta^\alpha
 (d^{(0)}_{0,0} + d^{(0)}_{0,2} + r d^{(0)}_{1,1} + r^2 d^{(0)}_{2,0}) \bigr) \\ 
 & & - \opm_M^{\gamma-1}\bigl( (-2t^2) T^{-2} \partial_\eta^\alpha
 (d^{(0)}_{0,0} + d^{(0)}_{0,2} + r d^{(0)}_{1,1} + r^2 d^{(0)}_{2,0}) \bigr) \biggr] \tilde{\sigma} \\ 
 & & + \sigma' \omega'_{1,\eta} r^2 \sum_\alpha \tfrac{1}{\alpha !} \biggl[
 \opm_M^{\gamma-2} \bigl( (-2t^2) T \partial_\eta^\alpha (d^{(1)}_{0,0} + d^{(1)}_{0,1} + r d^{(1)}_{1,0}) \bigr) \\
 & & - \opm_M^{\gamma-1} \bigl( (-2t^2) T \partial_\eta^\alpha (d^{(1)}_{0,0} + d^{(1)}_{0,1} + r d^{(1)}_{1,0}) \bigr)
 \biggr] \opm_M^{\gamma-1} \bigl( D^\alpha_y( \tfrac{1}{t} Z_1 ) \bigr) \tilde{\sigma} \\
 & & + \sigma' \omega'_{1,\eta} r^2 \sum_\alpha \tfrac{1}{\alpha !} 
 D^\alpha_y(\tfrac{ir C_1}{2t^2}) \biggl[ \opm_M^{\gamma-3}\bigl( (-2t^2) \partial_\eta^\alpha
 (d^{(1)}_{0,0} + d^{(1)}_{0,1} + r d^{(1)}_{1,0}) \bigr) \\ 
 & & - \opm_M^{\gamma-1}\bigl( (-2t^2) \partial_\eta^\alpha
 (d^{(1)}_{0,0} + d^{(1)}_{0,1} + r d^{(1)}_{1,0}) \bigr) \biggr] \tilde{\sigma} \\ 
 & & + \sigma' \omega'_{1,\eta} r^2 
 \tfrac{r^2 C_0}{2t^2} \biggl[ \opm_M^{\gamma-3}\bigl( (-2t^2) d^{(2)}_{0,0} \bigr)  
 - \opm_M^{\gamma-1}\bigl( (-2t^2) d^{(2)}_{0,0} \bigr) \biggr] \tilde{\sigma} \\ 
 & = & -2t^2 \sigma' \omega'_{1,\eta} \biggl[ r \res (d^{(0)}_{0,0},1) 
 M \bigl( \opm_M^{\gamma-1} \bigl( -\tfrac{4}{3t^2} w - \tfrac{1}{6t^2} \Delta_{S^2} + \tfrac{1}{t} Z_2 \bigr) 
\tilde{\sigma} (\cdot) \bigr) (1) \\
 & & + \res (d^{(0)}_{0,0},2) M \bigl( \opm_M^{\gamma-1} \bigl( -\tfrac{4}{3t^2} w - \tfrac{1}{6t^2} \Delta_{S^2} + \tfrac{1}{t} Z_2 \bigr)
 \tilde{\sigma} (\cdot) \bigr) (2) \\
 & & + \bigl( \res (\partial_\eta^\alpha d^{(0)}_{0,2},2) +r \res (\partial_\eta^\alpha d^{(0)}_{1,1},2)
 +r^2 \res (\partial_\eta^\alpha d^{(0)}_{2,0},2) \bigr)
 M \bigl( \opm_M^{\gamma-1} \bigl( D^\alpha_y(-\tfrac{4}{3t^2} w - \tfrac{1}{6t^2} \Delta_{S^2} + \tfrac{1}{t} Z_2) \bigr)
 \tilde{\sigma} (\cdot) \bigr) (2) \\
 & & +r^2 \res (T^{-2}d^{(0)}_{0,0},2) M \bigl( \tfrac{C_2}{2t^2} \tilde{\sigma} (\cdot) \bigr) (2)
 +r \res (T^{-2}d^{(0)}_{0,0},3) M \bigl( \tfrac{C_2}{2t^2} \tilde{\sigma} (\cdot) \bigr) (3)
 + \res (T^{-2}d^{(0)}_{0,0},4) M \bigl( \tfrac{C_2}{2t^2} \tilde{\sigma} (\cdot) \bigr) (4) \\
 & & + \bigl( \res (T^{-2}\partial_\eta^\alpha d^{(0)}_{0,2},4) +r \res (T^{-2}\partial_\eta^\alpha d^{(0)}_{1,1},4)
 +r^2 \res (T^{-2}\partial_\eta^\alpha d^{(0)}_{2,0},4) \bigr) M \bigl( D^\alpha_y (\tfrac{C_2}{2t^2}) \tilde{\sigma} (\cdot) \bigr) (4) \\
 & & +r \res (T d^{(1)}_{0,0},1) M \bigl( \tfrac{1}{t} Z_1 \tilde{\sigma} (\cdot) \bigr) (1) 
 + \sum_\alpha \tfrac{1}{\alpha !} \bigl( r \res (T \partial_\eta^\alpha d^{(1)}_{0,1},1) 
 +r^2 \res (T \partial_\eta^\alpha d^{(1)}_{1,0},1) \bigr)
 M \bigl( D^\alpha_y(\tfrac{1}{t} Z_1) \tilde{\sigma} (\cdot) \bigr) (1) \\
 & & +r^2 \res (d^{(1)}_{0,0},1) M \bigl( \tfrac{iC_1}{2t^2} \tilde{\sigma} (\cdot) \bigr) (1)
 +r \res (d^{(1)}_{0,0},2) M \bigl( \tfrac{iC_1}{2t^2} \tilde{\sigma} (\cdot) \bigr) (2) \\
 & & + \sum_\alpha \tfrac{1}{\alpha !} \bigl( r \res (\partial_\eta^\alpha d^{(1)}_{0,1},2)
 +r^2 \res (\partial_\eta^\alpha d^{(1)}_{1,0},2) \bigr) 
 M \bigl( D^\alpha_y(\tfrac{iC_1}{2t^2}) \tilde{\sigma} (\cdot) \bigr) (2) \\
 & & +r^2 \res (d^{(2)}_{0,0},2) M \bigl( \tfrac{C_0}{2t^2} \tilde{\sigma} (\cdot) \bigr) (2) 
 + {\cal O}(r^3) \biggr],
\end{eqnarray*}
where again in the second step Cauchy's residue theorem has been applied. 

\subsubsection{Green operators from remainders}
In our application, we want to apply the asymptotic parametrix up to second order, i.e., $P \sim P_0 +rP_1 +r^{2}P_2 + \cdots$,
from the left to the shifted Hamiltonian $A$. In the previous
paragraphs, however, this has been done only up to second order
in the asymptotic expansion of the shifted Hamiltonian , i.e.,
$A \sim A_0 +rA_1 +r^{2}A_2 + \cdots$. Therefore, certain Green
operators of second and lower order are still missing which originate from higher order equations in the asymptotic parametrix construction.

Let us first consider contributions from $P_0$, which can be derived from the third order equation
\[
 P_0 r^3 S_3 =0 \quad \mbox{mod} \ L^0_G + L^0_{\gflat} \quad \mbox{with} \ r^3 S_3 := r^3 A_3 + r^4 A_4 + \cdots .
\]
In terms of operator valued symbols this corresponds to 
\[
 \sigma' \omega'_{1,\eta} r^2 \sum_\alpha \tfrac{1}{\alpha !} 
 \partial_\eta^\alpha \bigl(\opm_M^{\gamma-3}(a_0^{(-1)})\bigr) \tilde{\sigma}' 
 \sigma r' D^\alpha_y \bigl(\opm_M^{\gamma-1}(s_3)\bigr) \tilde{\sigma} 
 = g_{0,3} + g_{0,4} \quad \mbox{mod} \ L^0_{\gflat} ,
\]
where the operator valued Green symbols $g_{0,3}$ and $g_{0,4}$ can be obtained
along the some line of arguments as presented in the previous paragraphs. Let us first consider the type-$b$ Green operator symbol
\begin{eqnarray*}
 g_{0,3} & = & \sigma' \omega'_{1,\eta} r^3 \sum_\alpha \tfrac{1}{\alpha !} \biggl[
 \opm_M^{\gamma-4}\bigl( (-2t^2) T^{-1} \partial_\eta^\alpha (d^{(0)}_{0,0} + d^{(0)}_{0,2} + r d^{(0)}_{1,1} + r^2 d^{(0)}_{2,0}) \bigr) \\
 & & - \opm_M^{\gamma-1}\bigl( (-2t^2) T^{-1} \partial_\eta^\alpha 
(d^{(0)}_{0,0} + d^{(0)}_{0,2} + r d^{(0)}_{1,1} + r^2 d^{(0)}_{2,0}) \bigr) \biggr]
 \opm_M^{\gamma-1} \bigl( D^\alpha_y s_3 \bigr) \tilde{\sigma} \\
 & = & -2t^2 \sigma' \omega'_{1,\eta} \biggl[ r^2 \res (T^{-1} d^{(0)}_{0,0},1)
 M \bigl( \opm_M^{\gamma-1} \bigl( s_3 \bigr) \tilde{\sigma} (\cdot) \bigr) (1)
 +r \res (T^{-1} d^{(0)}_{0,0},2) M \bigl( \opm_M^{\gamma-1} \bigl( s_3 \bigr) \tilde{\sigma} (\cdot) \bigr) (2) \\
 & & + \res (T^{-1} d^{(0)}_{0,0},3) M \bigl( \opm_M^{\gamma-1} \bigl( s_3 \bigr) \tilde{\sigma} (\cdot) \bigr) (3) 
 + \sum_\alpha \tfrac{1}{\alpha !} \bigl( \res (T^{-1} \partial_\eta^\alpha d^{(0)}_{0,2},3) 
 +r \res (T^{-1} \partial_\eta^\alpha d^{(0)}_{1,1},3) \\
 & & +r^2 \res (T^{-1} \partial_\eta^\alpha d^{(0)}_{2,0},3) \bigr)
 M \bigl( \opm_M^{\gamma-1} \bigl( D^\alpha_y s_3 \bigr) \tilde{\sigma} (\cdot) \bigr) (3) + {\cal O}(r^3) \biggr] ,
\end{eqnarray*}
where once again Cauchy's residue theorem has been applied. 
Similarly, the type-$a$ Green operator symbol also follows from this calculation, i.e.,
\begin{eqnarray*}
 g_{0,4} & = & \sigma' \omega'_{1,\eta} \sum_\alpha \tfrac{1}{\alpha !} 
 \opm_M^{\gamma-1}\bigl( (-2t^2) T^2 \partial_\eta^\alpha (d^{(0)}_{0,0} + d^{(0)}_{0,2} + r d^{(0)}_{1,1} + r^2 d^{(0)}_{2,0}) \bigr)
 \bigl(\tilde{\sigma}' \sigma -1 \bigr)
 \opm_M^{\gamma-1} \bigl( r'^3 D^\alpha_y s_3 \bigr) \tilde{\sigma}  .
\end{eqnarray*}
Actually, as already mentioned before, these type-$a$ and $b$ Green operator symbols are of second and lower order. 

Next, let us consider the additional Green operators originating from $P_1$, which can be obtained from the third order equation
\[
 r P_1 r^2 S_2 = 0 \quad \mbox{mod} \ L^0_G + L^0_{\gflat} \quad \mbox{with} \ r^2 S_2 := r^2 A_2 + r^3 A_3 + \cdots .
\]
In terms of operator valued symbols this corresponds to 
\[
 \sigma' \omega'_{1,\eta} r^3 \sum_\alpha \tfrac{1}{\alpha !} 
 \partial_\eta^\alpha \bigl(\opm_M^{\gamma-3}(a_1^{(-1)})\bigr) \tilde{\sigma}' 
 \sigma D^\alpha_y \bigl(\opm_M^{\gamma-1}(s_2)\bigr) \tilde{\sigma} 
 = g_{1,3} + g_{1,4} \quad \mbox{mod} \ L^0_{\gflat}
\]
Like before let us first consider the type-$b$ Green operator symbol
\begin{eqnarray*}
 g_{1,3} & = & \sigma' \omega'_{1,\eta} r^3 \sum_\alpha \tfrac{1}{\alpha !} \bigl[
 \opm_M^{\gamma-3}\bigl( (-2t^2) \partial_\eta^\alpha (d^{(1)}_{0,0} + d^{(1)}_{0,1} + r d^{(1)}_{1,0}) \bigr) \\
 & & - \opm_M^{\gamma-1}\bigl( (-2t^2) \partial_\eta^\alpha 
(d^{(1)}_{0,0} + d^{(1)}_{0,1} + r d^{(1)}_{0,1}) \bigr) \biggr]
 \opm_M^{\gamma-1} \bigl( D^\alpha_y s_2 \bigr) \tilde{\sigma} \\
 & = & -2t^2 \sigma' \omega'_{1,\eta} \biggl[ r^2 \res (d^{(1)}_{0,0},1) M \bigl( 
 \opm_M^{\gamma-1} \bigl( s_2 \bigr) \tilde{\sigma} (\cdot) \bigr) (1)
 +r \res (d^{(1)}_{0,0},2) M \bigl( \opm_M^{\gamma-1} \bigl( s_2 \bigr) \tilde{\sigma} (\cdot) \bigr) (2) \\
 & & + \sum_\alpha \tfrac{1}{\alpha !} \bigl( r \res (\partial_\eta^\alpha d^{(1)}_{0,1},2)
 + r^2 \res (\partial_\eta^\alpha d^{(1)}_{1,0},2) \bigr) 
 M \bigl( \opm_M^{\gamma-1} \bigl( D^\alpha_y s_2 \bigr) \tilde{\sigma} (\cdot) \bigr) (2) 
 + {\cal O}(r^3) \biggr]
\end{eqnarray*}
and subsequently the type-$a$ Green operator symbol 
\begin{eqnarray*}
 g_{1,4} & = & \sigma' \omega'_{1,\eta} r \sum_\alpha \tfrac{1}{\alpha !} 
 \opm_M^{\gamma-1}\bigl( (-2t^2) T^2 \partial_\eta^\alpha (d^{(1)}_{0,0} + d^{(1)}_{0,1} + r d^{(1)}_{1,0}) \bigr)
 \bigl(\tilde{\sigma}' \sigma -1 \bigr)
 \opm_M^{\gamma-1} \bigl( r'^2 D^\alpha_y s_2 \bigr) \tilde{\sigma}  .
\end{eqnarray*}

Finally, let us consider the additional Green operators originating from $P_2$, which can be obtained from the third order equation
\[
 r^2 P_2 r S_1 = 0 \quad \mbox{mod} \ L^0_G + L^0_{\gflat} \quad \mbox{with} \ r S_1 := r A_1 + r^2 A_2 + \cdots .
\]
In terms of operator valued symbols this corresponds to
\[
 \sigma' \omega'_{1,\eta} r^4 \sum_\alpha \tfrac{1}{\alpha !} 
 \partial_\eta^\alpha \bigl(\opm_M^{\gamma-3}(a_2^{(-1)})\bigr) \tilde{\sigma}' 
 \sigma r'^{-1} D^\alpha_y \bigl(\opm_M^{\gamma-1}(s_1)\bigr) \tilde{\sigma} 
 = g_{2,3} + g_{2,4} \quad \mbox{mod} \ L^0_{\gflat} .
\]
The type-$b$ Green operator symbol is given by
\begin{eqnarray*}
 g_{2,3} & = & \sigma' \omega'_{1,\eta} r^3 \bigl[
 \opm_M^{\gamma-2}\bigl( (-2t^2) T d^{(2)}_{0,0} \bigr) 
 - \opm_M^{\gamma-1}\bigl( (-2t^2) T d^{(2)}_{0,0} \bigr) \biggr] 
 \opm_M^{\gamma-1} \bigl( s_1 \bigr) \tilde{\sigma} \\
 & = & -2t^2 \sigma' \omega'_{1,\eta} \biggl[ 
 r^2 \res (T d^{(2)}_{0,0},1) M \bigl( \opm_M^{\gamma-1} \bigl( s_1 \bigr) \tilde{\sigma} (\cdot) \bigr) (1) 
 + {\cal O}(r^3) \biggr]
\end{eqnarray*}
and the type-$a$ Green operator symbol by
\begin{eqnarray*}
 g_{2,4} & = & \sigma' \omega'_{1,\eta} r^2 
 \opm_M^{\gamma-1}\bigl( (-2t^2) T^2 d^{(2)}_{0,0} \bigr)
 \bigl(\tilde{\sigma}' \sigma -1 \bigr)
 \opm_M^{\gamma-1} \bigl( r' s_1 \bigr) \tilde{\sigma} .
\end{eqnarray*}

\subsubsection{Summing up the $a$-type Green operator symbols}
\label{sumatypeG} 
So far, we have applied Cauchy's residue theorem only to type-$b$ Green operator symbols in order to fully exploit the formulas for these operators. It is the purpose of this paragraph 
to perform these calculations for type-$a$ Green operator symbols as well. Before applying Cauchy's residue theorem, it is, however,
convenient to first sum up all $a$-type Green operator symbols
into a single operator, i.e.,
\begin{eqnarray*}
 g_a & := & g_{0,1} + g_{0,4} + g_{1,1} + g_{1,4} + g_{2,1} + g_{2,4} \\
 & = & \sigma' \omega'_{1,\eta} \sum_\alpha \tfrac{1}{\alpha !} 
 \opm_M^{\gamma-1}\bigl( (-2t^2) T^2 \partial_\eta^\alpha (d^{(0)}_{0,0} + d^{(0)}_{0,2} + r d^{(0)}_{1,1} + r^2 d^{(0)}_{2,0}) \\
 & & + r (d^{(1)}_{0,0} + d^{(1)}_{0,1} + r d^{(1)}_{1,0})  +r^2 d^{(2)}_{0,0} \bigr)
 \bigl(\tilde{\sigma}' \sigma -1 \bigr)
 \opm_M^{\gamma-1} \bigl( D^\alpha_y a \bigr) \tilde{\sigma} \\
 & = & -2t^2 \sigma' \omega'_{1,\eta} \biggl[
 \res (T^2 d^{(0)}_{0,0},0)
 M \bigl( (\tilde{\sigma}' \sigma -1) \opm_M^{\gamma-1} \bigl( a \bigr) \tilde{\sigma} (\cdot) \bigr) (0) \\
 & & +r \res (T^2 d^{(0)}_{0,0},-1)
 M \bigl( (\tilde{\sigma}' \sigma -1) \opm_M^{\gamma-1} \bigl( a \bigr) \tilde{\sigma} (\cdot) \bigr) (-1) \\
 & & +r^2 \res (T^2 d^{(0)}_{0,0},-2)
 M \bigl( (\tilde{\sigma}' \sigma -1) \opm_M^{\gamma-1} \bigl( a \bigr) \tilde{\sigma} (\cdot) \bigr) (-2) \\
 & & + \sum_\alpha \tfrac{1}{\alpha !} \bigl( \res (T^2 \partial_\eta^\alpha d^{(0)}_{0,2},0)
 +r \res (T^2 \partial_\eta^\alpha d^{(0)}_{1,1},0) \\
 & & +r^2 \res (T^2 \partial_\eta^\alpha d^{(0)}_{2,0},0) \bigl)
 M \bigl( (\tilde{\sigma}' \sigma -1) \opm_M^{\gamma-1} \bigl( D^\alpha_y a \bigr) \tilde{\sigma} (\cdot) \bigr) (0) \\
 & & +r \res (T^2 d^{(1)}_{0,0},0)  
 M \bigl( (\tilde{\sigma}' \sigma -1) \opm_M^{\gamma-1} \bigl( a \bigr) \tilde{\sigma} (\cdot) \bigr) (0) \\
 & & +r^2 \res (T^2 d^{(1)}_{0,0},-1) 
 M \bigl( (\tilde{\sigma}' \sigma -1) \opm_M^{\gamma-1} \bigl( a \bigr) \tilde{\sigma} (\cdot) \bigr) (-1) \\
 & & + \sum_\alpha \tfrac{1}{\alpha !} \bigl( r \res (T^2 \partial_\eta^\alpha d^{(1)}_{0,1},0)
 +r^2 \res (T^2 \partial_\eta^\alpha d^{(1)}_{1,0},0) \\
 & & +r^2 \res (T^2 \partial_\eta^\alpha d^{(2)}_{0,0},0) \bigr)
 M \bigl( (\tilde{\sigma}' \sigma -1) \opm_M^{\gamma-1} \bigl( D^\alpha_y a \bigr) \tilde{\sigma} (\cdot) \bigr) (0) 
 + {\cal O}(r^3) \biggr] .
\end{eqnarray*}

\subsubsection{Calculation of the residues}
Up to this point, we have obtained explicit formulas for all
Green operator symbols which contribute up to second order in the asymptotics. What is only missing are the actual values of 
residues $\res(T^m d^{(k)}_{n,j}, w_0)$ at certain poles $w_0$ of the shifted meromophic operator valued symbols 
$T^m d^{(k)}_{n,j}$ of the asymptotic parametrix. The calculation
of these residues is rather straightforward. However, for the convenience of the reader and in order to improve the comprehensibility of our calculations, we list values of all required residues in Appendix \ref{Appendixresidues}.

\subsection{Contributions of specific angular momenta to Green operators}
\label{4.3}
It is a particularly pleasant feature of our asymptotic expansion
to provide a resolution of the wavefunction near an edge
not only with respect to the distance variable $r$ but also
with respect to the relative angular momentum of the coalescing
particles. There is a well known constraint, cf.~\cite{HO2S94}, on the relative angular momentum $l$ of two particles with respect to the asymptotic order,
i.e., only angular momenta $l \leq k$ contribute to the $r^{k}$-term of the asymptotic expansion. 
This constraint is also an immediate consequence of our calculations and it is convenient to extract from our formulas of the Green operators, contributions of specific angular momenta. 

Within the present work we consider the asymptotic expansion up
to second order. Therefore, one can derive asymptotic information
for angular momenta $l=0,1,2$ with corresponding projection operators ${\cal  P}_0$, ${\cal  P}_1$ and ${\cal  P}_2$.
For the entire $a$-type Green operator, defined in Section
(\ref{sumatypeG}), the individual angular momentum contributions are given by
\begin{eqnarray*}
{\cal  P}_0 g_a & = & 2t^2 \sigma' \omega'_{1,\eta} \biggl[ \bigl(
 1 +rtZ_1 +r^2 \bigl( -2 +\tfrac{(tZ_1)^2}{3} +\tfrac{tZ_2}{3} \bigr) \bigr)
 M \bigl( (\tilde{\sigma}' \sigma -1) \opm_M^{\gamma-1} \bigl( a \bigr) \tilde{\sigma} (\cdot) \bigr) (0) \\
 & & + \tfrac{1}{6} r^2 \sum_\alpha \tfrac{1}{\alpha !} \partial_\eta^\alpha \bigl( C_0 - 4it\tau + i C_1 \bigr) 
 M \bigl( (\tilde{\sigma}' \sigma -1) \opm_M^{\gamma-1} \bigl( D^\alpha_y a \bigr) \tilde{\sigma} (\cdot) \bigr) (0) 
 \biggr] + {\cal O}(r^3), \\
 {\cal  P}_1 g_a & = & 2t^2 \sigma' \omega'_{1,\eta} \bigl(
 \tfrac{1}{3} r + \tfrac{1}{6}r^2tZ_1 \bigr) \bigr[
 M \bigl( (\tilde{\sigma}' \sigma -1) \opm_M^{\gamma-1} \bigl( a \bigr) \tilde{\sigma} (\cdot) \bigr) (-1) \\
 & & -tZ_1 M \bigl( (\tilde{\sigma}' \sigma -1) \opm_M^{\gamma-1} \bigl( a \bigr) \tilde{\sigma} (\cdot) \bigr) (0) \bigr]
 + {\cal O}(r^3), \\
 {\cal  P}_2 g_a & = & 2t^2 \sigma' \omega'_{1,\eta} \biggl[ \tfrac{1}{5} r^2 \bigl[
 M \bigl( (\tilde{\sigma}' \sigma -1) \opm_M^{\gamma-1} \bigl( a \bigr) \tilde{\sigma} (\cdot) \bigr) (-2) 
 -\tfrac{1}{2} tZ_1 M \bigl( (\tilde{\sigma}' \sigma -1) \opm_M^{\gamma-1} \bigl( a \bigr) \tilde{\sigma} (\cdot) \bigr) (-1) \\
 & & + \tfrac{1}{6} \bigl( 10 +(tZ_1)^2 -2tZ_2 \bigr) 
 M \bigl( (\tilde{\sigma}' \sigma -1) \opm_M^{\gamma-1} \bigl( a \bigr) \tilde{\sigma} (\cdot) \bigr) (0) \bigr] \\
 & & - \tfrac{1}{30} r^2 \sum_\alpha \tfrac{1}{\alpha !} \partial_\eta^\alpha \bigl( C_0 - 4it\tau + i C_1 \bigr) 
 M \bigl( (\tilde{\sigma}' \sigma -1) \opm_M^{\gamma-1} \bigl( D^\alpha_y a \bigr) \tilde{\sigma} (\cdot) \bigr) (0) 
 \biggr] + {\cal O}(r^3) .
\end{eqnarray*}
The remaining $b$-type Green operators are treated separately, with individual angular momentum contributions given by
\begin{eqnarray*}
 {\cal  P}_0 g_{0,2} & = & 2t^2 \sigma' \omega'_{1,\eta} \biggl[ \bigl( 1 -\tfrac{1}{3} r^2 \bigr)
  M \bigl( \tfrac{r^2C_0}{2t^2} \tilde{\sigma} (\cdot) \bigr) (0) \\
 &  & + \tfrac{1}{6} r^2 \sum_\alpha \tfrac{1}{\alpha !} \partial_\eta^\alpha \bigl( C_0 - 4it\tau \bigr)
 M \bigl( D^\alpha_y (\tfrac{r^2C_0}{2t^2}) \tilde{\sigma} (\cdot) \bigr) (0)
 \biggr] + {\cal O}(r^3), 
\\
 {\cal  P}_1 g_{0,2} & = & 2t^2 \sigma' \omega'_{1,\eta} \tfrac{1}{3} r 
 M \bigl( \tfrac{r^2C_0}{2t^2} \tilde{\sigma} (\cdot) \bigr) (-1) 
 + {\cal O}(r^3), 
\\
 {\cal  P}_2 g_{0,2} & = & 2t^2 \sigma' \omega'_{1,\eta} \biggl[ \tfrac{1}{15} r^2 
  M \bigl( \tfrac{r^2C_0}{2t^2} \tilde{\sigma} (\cdot) \bigr) (0) \\
 &  & - \tfrac{1}{30} r^2 \sum_\alpha \tfrac{1}{\alpha !} \partial_\eta^\alpha \bigl( C_0 - 4it\tau \bigr)
 M \bigl( D^\alpha_y (\tfrac{r^2C_0}{2t^2}) \tilde{\sigma} (\cdot) \bigr) (0)
 \biggr] + {\cal O}(r^3) , 
\end{eqnarray*}
\begin{eqnarray*}
 {\cal  P}_0 g_{0,3} & = & 2t^2 \sigma' \omega'_{1,\eta} \biggl[ \bigl( 1 -\tfrac{1}{3} r^2 \bigr)
 M \bigl( \opm_M^{\gamma-1} \bigl( r^3 s_3 \bigr) \tilde{\sigma} (\cdot) \bigr) (0) \\
 &  & + \tfrac{1}{6} r^2 \sum_\alpha \tfrac{1}{\alpha !} \partial_\eta^\alpha \bigl( C_0 - 4it\tau \bigr)
 M \bigl( \opm_M^{\gamma-1} \bigl( r^3 D^\alpha_y s_3 \bigr) \tilde{\sigma} (\cdot) \bigr) (0)
 \biggr] + {\cal O}(r^3), 
\\
{\cal  P}_1 g_{0,3} & = & 2t^2 \sigma' \omega'_{1,\eta} \tfrac{1}{3} r 
 M \bigl( \opm_M^{\gamma-1} \bigl( r^3 s_3 \bigr) \tilde{\sigma} (\cdot) \bigr) (-1) 
 + {\cal O}(r^3), 
\\
{\cal  P}_2 g_{0,3} & = & 2t^2 \sigma' \omega'_{1,\eta} \biggl[ \tfrac{1}{5} r^2 
 M \bigl( \opm_M^{\gamma-1} \bigl( r^3 s_3 \bigr) \tilde{\sigma} (\cdot) \bigr) (-2) 
 + \tfrac{1}{15} r^2 M \bigl( \opm_M^{\gamma-1} \bigl( r^3 s_3 \bigr) \tilde{\sigma} (\cdot) \bigr) (0) \\
 &  & - \tfrac{1}{30} r^2 \sum_\alpha \tfrac{1}{\alpha !} \partial_\eta^\alpha \bigl( C_0 - 4it\tau \bigr)
 M \bigl( \opm_M^{\gamma-1} \bigl( r^3 D^\alpha_y s_3 \bigr) \tilde{\sigma} (\cdot) \bigr) (0)
 \biggr] + {\cal O}(r^3) , 
\end{eqnarray*}
\begin{eqnarray*}
{\cal  P}_0 g_{1,2} & = & 2t^2 \sigma' \omega'_{1,\eta} \biggl[ \bigl( rtZ_1 -\tfrac{5}{3} r^2 \bigr)
 M \bigl( \tfrac{r^2C_0}{2t^2} \tilde{\sigma} (\cdot) \bigr) (0) 
 + \tfrac{1}{6} r^2 \sum_\alpha \tfrac{1}{\alpha !} i\partial_\eta^\alpha C_1 
 M \bigl( D^\alpha_y (\tfrac{r^2C_0}{2t^2}) \tilde{\sigma} (\cdot) \bigr) (0) \\
 & & + \bigl( 1 -\tfrac{1}{3} r^2 \bigr) 
 M \bigl( \opm_M^{\gamma-1} \bigl( r a_1 \bigr) \tilde{\sigma} (\cdot) \bigr) (0)
 + \tfrac{1}{6} r^2 \sum_\alpha \tfrac{1}{\alpha !} \partial_\eta^\alpha \bigl( C_0 - 4it\tau \bigr)
 M \bigl( \opm_M^{\gamma-1} \bigl( r D^\alpha_y a_1 \bigr) \tilde{\sigma} (\cdot) \bigr) (0)
 \biggr] + {\cal O}(r^3), 
\\
{\cal  P}_1 g_{1,2} & = & 2t^2 \sigma' \omega'_{1,\eta} \biggl[ \tfrac{1}{6} r^2 tZ_1
 M \bigl( \tfrac{r^2C_0}{2t^2} \tilde{\sigma} (\cdot) \bigr) (-1)
 - \tfrac{1}{3} r tZ_1 M \bigl( \tfrac{r^2C_0}{2t^2} \tilde{\sigma} (\cdot) \bigr) (0)
 + \tfrac{1}{3} r M \bigl( \tfrac{r^2iC_1}{2t^2} \tilde{\sigma} (\cdot) \bigr) (-1)
 \biggr] + {\cal O}(r^3), 
\\
 {\cal  P}_2 g_{1,2} & = & 2t^2 \sigma' \omega'_{1,\eta} \biggl[ -\tfrac{1}{10} r^2 tZ_1
 M \bigl( \tfrac{r^2C_0}{2t^2} \tilde{\sigma} (\cdot) \bigr) (-1) 
 + \tfrac{1}{3} r^2 M \bigl( \tfrac{r^2C_0}{2t^2} \tilde{\sigma} (\cdot) \bigr) (0)
 - \tfrac{1}{30} r^2 \sum_\alpha \tfrac{1}{\alpha !} i\partial_\eta^\alpha C_1 
 M \bigl( D^\alpha_y (\tfrac{r^2C_0}{2t^2}) \tilde{\sigma} (\cdot) \bigr) (0) \\
 & & + \tfrac{1}{15} r^2  
 M \bigl( \opm_M^{\gamma-1} \bigl( r a_1 \bigr) \tilde{\sigma} (\cdot) \bigr) (0)
 - \tfrac{1}{30} r^2 \sum_\alpha \tfrac{1}{\alpha !} \partial_\eta^\alpha \bigl( C_0 - 4it\tau \bigr)
 M \bigl( \opm_M^{\gamma-1} \bigl( r D^\alpha_y a_1 \bigr) \tilde{\sigma} (\cdot) \bigr) (0)
 \biggr] + {\cal O}(r^3) ,
\end{eqnarray*}
\begin{eqnarray*}
 {\cal  P}_0 g_{1,3} & = & 2t^2 \sigma' \omega'_{1,\eta} \biggl[ \bigl( rtZ_1 -\tfrac{5}{3} r^2 \bigr)
 M \bigl( \opm_M^{\gamma-1} \bigl( r^2 s_2 \bigr) \tilde{\sigma} (\cdot) \bigr) (0) 
 + \tfrac{1}{6} r^2 \sum_\alpha \tfrac{1}{\alpha !} i\partial_\eta^\alpha C_1 
 M \bigl( \opm_M^{\gamma-1} \bigl( r^2 D^\alpha_y s_2 \bigr) \tilde{\sigma} (\cdot) \bigr) (0)
 \biggr]\\
 & & + {\cal O}(r^3),
\\
 {\cal  P}_1 g_{1,3} & = & 2t^2 \sigma' \omega'_{1,\eta} \biggl[ \tfrac{1}{6} r^2 tZ_1
 M \bigl( \opm_M^{\gamma-1} \bigl( r^2 s_2 \bigr) \tilde{\sigma} (\cdot) \bigr) (-1)
 - \tfrac{1}{3} r tZ_1 M \bigl( \opm_M^{\gamma-1} \bigl( r^2 s_2 \bigr) \tilde{\sigma} (\cdot) \bigr) (0)
 \biggr] + {\cal O}(r^3), 
\\
{\cal  P}_2 g_{1,3} & = & 2t^2 \sigma' \omega'_{1,\eta} \biggl[ -\tfrac{1}{10} r^2 tZ_1
 M \bigl( \opm_M^{\gamma-1} \bigl( r^2 s_2 \bigr) \tilde{\sigma} (\cdot) \bigr) (-1) 
 + \tfrac{1}{3} r^2 M \bigl( \opm_M^{\gamma-1} \bigl( r^2 s_2 \bigr) \tilde{\sigma} (\cdot) \bigr) (0) \\
 & & - \tfrac{1}{30} r^2 \sum_\alpha \tfrac{1}{\alpha !} i\partial_\eta^\alpha C_1 
 M \bigl( \opm_M^{\gamma-1} \bigl( r^2 D^\alpha_y s_2 \bigr) \tilde{\sigma} (\cdot) \bigr) (0) 
 \biggr]  + {\cal O}(r^3) ,
\end{eqnarray*}
\begin{eqnarray*}
{\cal  P}_0 g_{2,2} & = & 2t^2 \sigma' \omega'_{1,\eta} \biggl[ \bigl( 1 -\tfrac{1}{3} r^2 \bigr) 
 M \bigl( \opm_M^{\gamma-1} \bigl( r^2 a_2 \bigr) \tilde{\sigma} (\cdot) \bigr) (0) \\
 & & + \tfrac{1}{6} r^2 \sum_\alpha \tfrac{1}{\alpha !} \partial_\eta^\alpha \bigl( C_0 - 4it\tau \bigr)
 M \bigl( \opm_M^{\gamma-1} \bigl( r^2 D^\alpha_y a_2 \bigr) \tilde{\sigma} (\cdot) \bigr) (0) \\
 & & + \bigl( rtZ_1 -\tfrac{5}{3} r^2 \bigr)
 M \bigl( \opm_M^{\gamma-1} \bigl( r a_1 \bigr) \tilde{\sigma} (\cdot) \bigr) (0)
 + \tfrac{1}{6} r^2 \sum_\alpha \tfrac{1}{\alpha !} i\partial_\eta^\alpha C_1
 M \bigl( \opm_M^{\gamma-1} \bigl( r a_1 \bigr) \tilde{\sigma} (\cdot) \bigr) (0) \\
 & & + r^2 \bigl( \tfrac{1}{3} (tZ_1)^2 +\tfrac{1}{3} tZ_2 \bigr) 
 M \bigl( \tfrac{r^2C_0}{2t^2} \tilde{\sigma} (\cdot) \bigr) (0)
 \biggr] + {\cal O}(r^3), 
\\
{\cal  P}_1 g_{2,2} & = & 2t^2 \sigma' \omega'_{1,\eta} \biggl[ \tfrac{1}{3} r 
 M \bigl( \opm_M^{\gamma-1} \bigl( r^2 a_2 \bigr) \tilde{\sigma} (\cdot) \bigr) (-1) 
 - \tfrac{1}{3} r tZ_1 M \bigl( \opm_M^{\gamma-1} \bigl( r a_1 \bigr) \tilde{\sigma} (\cdot) \bigr) (0) \\
 & & + \tfrac{1}{6} r^2 tZ_1 M \bigl( \tfrac{r^2iC_1}{2t^2} \tilde{\sigma} (\cdot) \bigr) (-1) 
 - \tfrac{1}{6} r^2 (tZ_1)^2 M \bigl( \tfrac{r^2C_0}{2t^2} \tilde{\sigma} (\cdot) \bigr) (0)
 \biggr] + {\cal O}(r^3), 
\\
 {\cal  P}_2 g_{2,2} & = & 2t^2 \sigma' \omega'_{1,\eta} \biggl[ \tfrac{1}{15} r^2  
 M \bigl( \opm_M^{\gamma-1} \bigl( r^2 a_2 \bigr) \tilde{\sigma} (\cdot) \bigr) (0) \\
 & & - \tfrac{1}{30} r^2 \sum_\alpha \tfrac{1}{\alpha !} \partial_\eta^\alpha \bigl( C_0 - 4it\tau \bigr)
 M \bigl( \opm_M^{\gamma-1} \bigl( r^2 D^\alpha_y a_2 \bigr) \tilde{\sigma} (\cdot) \bigr) (0) \\ 
 & & + \tfrac{1}{5} r^2 M \bigl( \tfrac{r^4C_2}{2t^2} \tilde{\sigma} (\cdot) \bigr) (-2) 
 + \tfrac{1}{3} r^2 M \bigl( \opm_M^{\gamma-1} \bigl( r a_1 \bigr) \tilde{\sigma} (\cdot) \bigr) (0) \\
 & & - \tfrac{1}{30} r^2 \sum_\alpha \tfrac{1}{\alpha !} i\partial_\eta^\alpha C_1 
 M \bigl( \opm_M^{\gamma-1} \bigl( r D^\alpha_y a_1 \bigr) \tilde{\sigma} (\cdot) \bigr) (0) 
 - \tfrac{1}{10} r^2 tZ_1 M \bigl( \tfrac{r^2iC_1}{2t^2} \tilde{\sigma} (\cdot) \bigr) (-1) \\
 & & + r^2 \bigl( \tfrac{1}{30} (tZ_1)^2 -\tfrac{1}{15} (1+tZ_2) \bigr) 
 M \bigl( \tfrac{r^2C_0}{2t^2} \tilde{\sigma} (\cdot) \bigr) (0)
 \biggr] + {\cal O}(r^3) ,
\end{eqnarray*}
\begin{eqnarray*}
 {\cal  P}_0 g_{2,3} & = & 2t^2 \sigma' \omega'_{1,\eta} \biggl[ r^2 \bigl( \tfrac{1}{3} (tZ_1)^2 + \tfrac{1}{3} tZ_2 \bigr) 
 M \bigl( \opm_M^{\gamma-1} \bigl( r s_1 \bigr) \tilde{\sigma} (\cdot) \bigr) (0) 
 + {\cal O}(r^3) \biggr],
\\
 {\cal  P}_1 g_{2,3} & = & 2t^2 \sigma' \omega'_{1,\eta} \biggl[ - \tfrac{1}{6} r^2 (tZ_1)^2
 M \bigl( \opm_M^{\gamma-1} \bigl( r s_1 \bigr) \tilde{\sigma} (\cdot) \bigr) (0) 
 + {\cal O}(r^3) \biggr],
\\
 {\cal  P}_2 g_{2,3} & = & 2t^2 \sigma' \omega'_{1,\eta} \biggl[ r^2 \bigl( \tfrac{1}{30} (tZ_1)^2 - \tfrac{1}{15} (1+tZ_2) \bigr) 
 M \bigl( \opm_M^{\gamma-1} \bigl( r s_1 \bigr) \tilde{\sigma} (\cdot) \bigr) (0) 
 + {\cal O}(r^3) \biggr] .
\end{eqnarray*}

Finally, the individual angular momentum resolved symbols
of the Green operators can be arranged together in the form
of Theorem \ref{t1}, thus finishing its proof.

\section{Acknowledgement}
Financial support from the Deutsche Forschungsgemeinschaft DFG (Grant No.~HA 5739/3-1) is gratefully acknowledged.

\appendix
\vspace{1cm}
\noindent {\Large {\bf Appendix}}

\section{$\varepsilon$-regularization of parametrix and Green operator}
\label{epsreg}
The motivation for this regularization procedure is to simplify the construction and evaluation of the asymptotic parametrix and corresponding Green operators.This can be achieved by attaching a scaling parameter $\varepsilon$ to certain cut-off functions and  
by adding an appropriate Green operator to the equation for the 
parametrix. Let us first outline the basic idea in an informal
manner for an eigenvalue problem of the general form $Au_i=0$, where
$A=H-E_i$ corresponds to a Hamiltonian shifted by one its eigenvalues.
We want to recall, that the left parametrix equation  
\[
 PAu = \bigl( 1+G \bigr) u
\]
yields the asymptotic behaviour of the eigenfunction via $u_i=-Gu_i$. Now we add a parameter
dependent Green operator $G_\varepsilon^{0}$, which in a certain weak sense (to bespecified below) vanishes for $\varepsilon \rightarrow 0$, to the left hand side
of the parametrix equation and introduce scaling parameters in
certain cut-off functions which results in a modified equation
\begin{equation}
 P_{\varepsilon} A u + G_{\varepsilon}^{(0)}u = \bigl( 1+G_{\varepsilon} \bigr) u ,
\label{PAepsilon}
\end{equation}
where a parameter dependent parametrix $P_{\varepsilon}$ appears on the left and a modified Green operator $G_{\varepsilon}$ on the right hand side of the equation. With these modifications, the asymptotic behaviour follows from 
\begin{equation}
 u_i= \bigl( G_{\varepsilon}^{(0)} -G_{\varepsilon} \bigr) u_i .
\label{Gepsilon} 
\end{equation}
In the generic case, it is not possible to perform the limit
$\varepsilon \rightarrow 0$ in Eq.~(\ref{PAepsilon}) with respect to a given norm $\| \cdot \|$. However we are only interested in some particular functions $u$, i.e., eigenfunctions
which have special regularity and decay properties. In these particular cases
it is actually possible to perform the limit $\varepsilon \rightarrow 0$ with respect to the norm $\| \cdot \|_{{\cal W}^{s,\gamma}}$ and to get a limitting equation of the form
\begin{equation}
 P_0 A u = \bigl( 1+G_0 \bigr) u ,
\label{PA0}
\end{equation}
with
\begin{equation}
 P_0 A u := \lim_{\varepsilon \rightarrow 0} P_{\varepsilon} A u
 \quad \mbox{and} \quad
 G_0 u := \lim_{\varepsilon \rightarrow 0} G_{\varepsilon} u , 
\label{PGlimit}
\end{equation}
where by construction, we have assumed
\[
 \lim_{\varepsilon \rightarrow 0} G_{\varepsilon}^{(0)} u =0 .
\] 
The asymptotic behaviour of the eigenfunction $u_i$ can therefore be obtained from the equation 
\[
 -u_i = G_0 u_i .
\]
 
In our specific applications we consider additional Green operators 
\begin{equation}
 G_{\varepsilon}^{(0)}u = \int \dbar \eta \, e^{iy \eta} g_{\varepsilon}^{(0)}(y,\eta) F_{y \rightarrow \eta}
 u 
\label{Geps0}
\end{equation}
with operator valued symbols
$g_{\varepsilon}^{(0)}(y,\eta) \in R_G^{-m}(Y \times \mathbb{R}^3,{\boldsymbol g})$ for $m=0,1,2,\cdots$, 
which are twisted homogeneous of order $m$, i.e.,
\[
 \kappa_\lambda g_{\varepsilon}^{(0)}(y,\eta) \kappa^{-1}_\lambda u = \lambda^m g_{\varepsilon}^{(0)}(y,\lambda \eta) u.
\]
In particular let us assume a Green operator symbol of the following form 
\begin{equation}
g_{\varepsilon}^{(0)}(y,\eta) = \omega'_{1,\varepsilon\eta} r^{2+k} \opm_M^{\gamma-3}(d_{k,j})(y,\eta) 
 (1-\omega'_{0,\varepsilon\eta}) 
 a_i(y,\eta) \omega_{0,\varepsilon\eta} ,
\label{geps0}
\end{equation}
with homogeneous parameter dependent differential operators $a_i(y,\eta) \in \mbox{Diff}^i_{deg}(X^{\wedge})$, $i=0,1,2$, on the right
and Mellin pseudo-differential operator with homogeneous parameter dependent symbol $d_{k,j}(y,\eta)$ of order $j$ on the left side.  
Such type of Green operator is twisted homogeneous of order $m=2-i+k$.

\begin{definition}
Let us denote a particular function $u \in {\cal W}^{s}(\mathbb{R}^q, {\cal K}^{s, \gamma}(X^\wedge))$ to be 
$\varepsilon$-regularizable with respect to $G_{\varepsilon}^{(0)}$ if
\[
 \lim_{\varepsilon \rightarrow 0} \| G_{\varepsilon}^{(0)}u \|_{{\cal W}^{s,\gamma}} =0
\]
and the limits {\rm (\ref{PGlimit})} exist,
where we assume a Green operator symbol, cf.~{\rm (\ref{Geps0})}, of the type {\rm (\ref{geps0})} discussed before.
\end{definition}

For our purposes, the following lemma is sufficient. 
 
\begin{lemma}
Any function $u \in {\cal W}^{\infty}_{\mbox{\footnotesize comp}}(\mathbb{R}^q, {\cal K}^{s, \gamma}(X^\wedge)) \subset {\cal W}^{s}(\mathbb{R}^q, {\cal K}^{s, \gamma}(X^\wedge))$ is $\varepsilon$-regularizable.
\label{lemmaeps}
\end{lemma}
\begin{proof}
Let us consider the specific ${\cal W}^{s}(\mathbb{R}^q, {\cal K}^{s, \gamma}(X^\wedge))$ norm
\[
 \| G_{\varepsilon}^{(0)}u \|_{{\cal W}^{s,\gamma}}^2 := \int [\eta]^{2s} \| \kappa^{-1}_{[\eta]} (F_{y\rightarrow \eta} G_{\varepsilon}^{(0)}u)(\eta)
 \|_{{\cal K}^{s,\gamma}}^2 \, \dbar \eta .
\]
Taking into accout $u \in {\cal W}^{\infty}_{\mbox{\footnotesize comp}}(\mathbb{R}^q, {\cal K}^{s, \gamma}(X^\wedge))$, let us perform the following estimate
\begin{eqnarray*}
 \| G_{\varepsilon}^{(0)}u \|_{{\cal W}^{s,\gamma}}^2 & = &
 \int [\eta]^{2s} \bigl\| \kappa^{-1}_{[\eta]} \bigl( F_{y\rightarrow \eta}
 \int \dbar \tilde{\eta} \, e^{iy \tilde{\eta}} g_{\varepsilon}^{(0)}(y,\tilde{\eta}) F_{y \rightarrow \tilde{\eta}} u \bigr)(\eta´) \bigr\|_{{\cal K}^{s,\gamma}}^2 \dbar \eta \\   
 & = & \int [\eta]^{2s} \bigl\| \kappa^{-1}_{[\eta]} \bigl( F_{y\rightarrow \eta}
 \int \dbar \tilde{\eta} \, e^{iy \tilde{\eta}} g_{\varepsilon}^{(0)}(y,\tilde{\eta}) \bigl(r[\tilde{\eta}]\bigr)^{-N}  
 \bigl(r[\tilde{\eta}]\bigr)^{N} F_{y \rightarrow \tilde{\eta}} u \bigr)(\eta) \bigr\|_{{\cal K}^{s,\gamma}}^2 \dbar \eta \\
 & \lesssim & \| G_{\varepsilon}^{(-N)} \|_{{\cal L}({\cal W}^{s,\gamma},{\cal W}^{s,\gamma})}^{2} \int [\eta]^{2s+2N} \| \kappa^{-1}_{[\eta]} \bigl( F_{y\rightarrow \eta} 
 r^{N} u \bigr)(\eta) \bigr\|_{{\cal K}^{s,\gamma}}^2 \dbar \eta \\  
\end{eqnarray*}
with respect to the modified Green operator
\[
 G_{\varepsilon}^{(-N)}u = \int \dbar \eta \, e^{iy \eta} g_{\varepsilon}^{(0)}(y,\eta) \bigl(r[\eta]\bigr)^{-N} F_{y \rightarrow \eta} u .
\]
For $u \in {\cal W}^{\infty}_{\mbox{\footnotesize comp}}(\mathbb{R}^q, {\cal K}^{s, \gamma}(X^\wedge))$,
the Fourier integral
\[
 \int [\eta]^{2s+2N} \| \kappa^{-1}_{[\eta]} \bigl( F_{y\rightarrow \eta} 
 r^{N} u \bigr)(\eta) \bigr\|_{{\cal K}^{s,\gamma}}^2 \dbar \eta
 = \| r^{N} u \|_{{\cal W}^{s+N,\gamma}}^{2} 
\]
is obviously finite and it remains to show that the operator norm
of the modified Green operator vanishes for $\varepsilon \rightarrow 0$. 
According to \cite[Theorem 1.3.59]{Schulze98}, its
operator norm can be estimated by
\begin{equation}
 \| G_{\varepsilon}^{(-N)} \|_{{\cal L}({\cal W}^{s,\gamma},{\cal W}^{s,\gamma})} \lesssim \sup_{y \in Y} \sup_{\eta \in \mathbb{R}^{3}} \sup_{|\alpha| \leq 2l} \| \kappa^{-1}_{[\eta]} \partial^{\alpha}_{y} g_{\varepsilon}^{(0)}(y,\eta) \bigl(r[\eta]\bigr)^{-N} \kappa_{[\eta]} \|_{{\cal L}({\cal K}^{s,\gamma},{\cal K}^{s,\gamma})} 
\label{GNop}
\end{equation}
with $s$ dependent constant $l$. In our particular applications, we
consider local edges $Y$ which correspond to bounded open subsets
of $\mathbb{R}^{3}$. Taking furthermore into account the smoothness
properties of the symbol $g_{\varepsilon}^{(0)}$ it turns out to be sufficient to consider the norm estimate
(\ref{GNop}) pointwise with respect to the edge variables $y$.
In order to estimate the operator
norm, let us next take into account twisted homogeneity,i.e.,
\[
 \| \kappa^{-1}_{[\eta]} \partial^{\alpha}_{y} g_{\varepsilon}^{(0)}(y,\eta) \bigl(r[\eta]\bigr)^{-N} \kappa_{[\eta]} \|_{{\cal L}({\cal K}^{s,\gamma},{\cal K}^{s,\gamma})} = [\eta]^{-m}
 \| \partial^{\alpha}_{y} g_{\varepsilon}^{(0)}(y,\eta/[\eta]) r^{-N} \|_{{\cal L}({\cal K}^{s,\gamma},{\cal K}^{s,\gamma})} .
\]
Therefore, it is sufficient to consider the norm
\[
 \| \partial^{\alpha}_{y} g_{\varepsilon}^{(0)}(y,\eta/[\eta]) r^{-N}  \|_{{\cal L}({\cal K}^{s,\gamma},{\cal K}^{s,\gamma})} =
 \sup_{v \in {\cal K}^{s,\gamma}} \frac{\| \partial^{\alpha}_{y} g_{\varepsilon}^{(0)}(y,\eta/[\eta]) r^{-N} v \|_{{\cal K}^{s,\gamma}}}{\| v \|_{{\cal K}^{s,\gamma}}} .  
\] 
A simple calculation shows
\[
 \| \partial^{\alpha}_{y} g_{\varepsilon}^{(0)}(y,\eta/[\eta]) r^{-N} v \|_{{\cal K}^{s,\gamma}} = \varepsilon^{-2-m+N} \| \partial^{\alpha}_{y} g_{1}^{(0)}(y,\eta/[\eta]) r^{-N} v \bigl( \varepsilon^{-1} \, \cdot \, \bigr) \|_{{\cal K}^{s,\gamma}} .
\]
To proceed, we need to estimate $v \bigl( \varepsilon^{-1} \, \cdot \, \bigr)$ in the ${\cal K}^{s,\gamma}$-norm. 
 
\begin{proposition} 
For $v \in {\cal K}^{s,\gamma}(X^{\wedge})$, we get the estimate
\[
 \| v \bigl( \varepsilon^{-1} \, \cdot \, \bigr) \|_{{\cal K}^{s,\gamma}} \lesssim \varepsilon^{-\gamma-s+\tfrac{3}{2}}
 \| v \|_{{\cal K}^{s,\gamma}} .
\]
\end{proposition}

\begin{proof} 
According to the definition, we have
\begin{equation}
 \| v \bigl( \varepsilon^{-1} \, \cdot \, \bigr) \|_{{\cal K}^{s,\gamma}} = \| v \bigl( \varepsilon^{-1} \, \cdot \, \bigr) \|_{\sigma{\cal H}^{s,\gamma}} + \| v \bigl( \varepsilon^{-1} \, \cdot \, \bigr) \|_{(1-\sigma)H^{s}} .
\label{Knorm}
\end{equation}
Let us first consider
\begin{eqnarray*}
 \| v \bigl( \varepsilon^{-1} \, \cdot \, \bigr) \|_{\sigma{\cal H}^{s,\gamma}}^{2} & = & \sum_{|\beta| \leq s} \int \sigma \bigl| |x|^{-\gamma+|\beta|} \partial_x^{\beta} v \bigl( \varepsilon^{-1} \, \cdot \, \bigr) \bigr|^{2} dx \\ 
 & = & \sum_{|\beta| \leq s} \int \sigma \bigl( \varepsilon \, \cdot \, \bigr) \bigl| \varepsilon^{-\gamma} |x|^{-\gamma+|\beta|} \partial_x^{\beta} v \bigr|^{2} \varepsilon^{3} dx \\ 
 & = & \sum_{|\beta| \leq s} \varepsilon^{-2\gamma +3} \left[ \int \sigma \bigl| |x|^{-\gamma+|\beta|} \partial_x^{\beta} v \bigr|^{2} dx 
 + \int (1-\sigma) \sigma \bigl( \varepsilon \, \cdot \, \bigr) \bigl| |x|^{-\gamma+|\beta|} \partial_x^{\beta} v \bigr|^{2} dx \right] ,  
\end{eqnarray*}
where w.l.o.g.~we assume in the last line $\sigma \prec \sigma \bigl( \varepsilon \, \cdot \, \bigr)$ for $\varepsilon$ sufficiently small.
In order to estimate the second term, let us first observe
\[
 (1-\sigma)|x|^{-2\gamma} \leq C_1 \quad \mbox{for} \ \gamma > 0
\]
and 
\[
 \sigma \bigl( \varepsilon \, \cdot \, \bigr) \bigl| \varepsilon |x| \bigr|^{2|\beta|} \leq C_2 .
\]
Therefore, we get the estimate
\[
 \int (1-\sigma) \sigma \bigl( \varepsilon \, \cdot \, \bigr) \bigl| |x|^{-\gamma+|\beta|} \partial_x^{\beta} v \bigr|^{2} dx 
 \lesssim \varepsilon^{-2|\beta|} \int (1-\sigma) \bigl|  \partial_x^{\beta} v \bigr|^{2} dx  
\]
which yields 
\begin{eqnarray*}
 \| v \bigl( \varepsilon^{-1} \, \cdot \, \bigr) \|_{\sigma{\cal H}^{s,\gamma}}^{2} & \lesssim & \varepsilon^{-2\gamma +3} \| v \|_{\sigma{\cal H}^{s,\gamma}} + \varepsilon^{-2\gamma-2s+3}
 \| v \|_{(1-\sigma)H^{s}} \\
 & \lesssim & \varepsilon^{-2\gamma-2s+3} \| v \|_{{\cal K}^{s,\gamma}} . 
\end{eqnarray*}
For the remaining part of (\ref{Knorm}), we get the estimate
\begin{eqnarray*}
 \| v \bigl( \varepsilon^{-1} \, \cdot \, \bigr) \|_{(1-\sigma)H^{s}}
 & = & \sum_{|\beta| \leq s} \int (1-\sigma) \bigl| \partial_x^{\beta} v \bigr|^{2} dx \\
 & \lesssim & \sum_{|\beta| \leq s} \varepsilon^{-2|\beta|+3}
 \int \left( 1-\sigma \bigl( \varepsilon \, \cdot \, \bigr) \right) \bigl| \partial_x^{\beta} v \bigr|^{2} dx \\
 & \lesssim & \varepsilon^{-2s+3} \| v \|_{(1-\sigma)H^{s}} . 
\end{eqnarray*}
Putting all together, we, finally, get the desired estimate.
\end{proof}

We can now estimate the operator norm in the following manner
\begin{eqnarray*}
 \| \partial^{\alpha}_{y} g_{\varepsilon}^{(0)}(y,\eta/[\eta]) r^{-N}  \|_{{\cal L}({\cal K}^{s,\gamma},{\cal K}^{s,\gamma})} & \lesssim &
 \varepsilon^{-\gamma-s-\tfrac{1}{2}-m+N} \sup_{v \in {\cal K}^{s,\gamma}} \frac{
 \| \partial^{\alpha}_{y} g_{1}^{(0)}(y,\eta/[\eta]) r^{-N} v \bigl( \varepsilon^{-1} \, \cdot \, \bigr) \|_{{\cal K}^{s,\gamma}}}
 {\| v \bigl( \varepsilon^{-1} \, \cdot \, \bigr) \|_{{\cal K}^{s,\gamma}}} \\
 & \lesssim & \varepsilon^{-\gamma-s-\tfrac{1}{2}-m+N}
 \| \partial^{\alpha}_{y} g_{1}^{(0)}(y,\eta/[\eta]) r^{-N}  \|_{{\cal L}({\cal K}^{s,\gamma},{\cal K}^{s,\gamma})} 
\end{eqnarray*}
and, therefore, for $N > \gamma+s+\tfrac{1}{2}+m$, we get
\[
 \lim_{\varepsilon \rightarrow 0} \| \partial^{\alpha}_{y} g_{\varepsilon}^{(0)}(y,\eta/[\eta]) r^{-N}  \|_{{\cal L}({\cal K}^{s,\gamma},{\cal K}^{s,\gamma})} =0 .
\]
\end{proof}

\section{Some simple calculations in the noninteracting case}
\label{noninteracting}
On a first glance, the asymptotic expression for the Green operator (\ref{Gasymp}) looks rather awkward
concerning the explicit formulas of the coefficients ${\cal Q}_{l,n}$ given in terms of Mellin transforms.
In this appendix we present some explicit calculations for wavefunctions of two noninteracting electrons, i.e., we consider the equation
\[
\bigg(-\frac{1}{2}\big(\Delta_1+\Delta_2\big)-\frac{Z}{|x_1|}-\frac{Z}{|x_2|}\bigg)\Psi(x_1,x_2)=E\Psi(x_1,x_2).
\]
In particular, it provides an independent check for the correctness of our asymptotic Green operator.
Obviously it is only the $e-n$ cusp which can be studied for different angular momenta $l=0,1,2$.
Let us start with $l=0$, where the exact ground state wavefunction\footnote{Here and in the following 
we disregard proper normalization constants.}
is given in hyperspherical coordinates
\[
 u_0(r,t)= e^{-Zt( \sin r + \cos r)} .
\]
The corresponding asymptotic expansion is given by
\[
 u_0(r,t) \sim \bigl( 1 -Ztr + \tfrac{1}{2} \big(Zt +  (Zt)^2\big)r^2 \cdots \bigr) e^{-Zt} .
\]
We want to calculate
\[
 Gu_0(r,y) = \int \dbar \eta \, e^{iy \eta} g \hat{u}_0(y,\eta) .
\]
Let us first note that
\[
 \int \dbar \eta \, e^{iy \eta} {\cal Q}_{0,1}(\hat{u}_0)
 = \int \dbar \eta \, e^{iy \eta} M \bigl( \opm_M^{\gamma-1}(h_1^{(0)}) \tilde{\sigma} \hat{u}_0 \bigr) (0) ,
\]
because 
\[
 \int \dbar \eta \, e^{iy \eta} M \bigl( \tilde{\sigma}' \sigma \opm_M^{\gamma-1}(a) \tilde{\sigma} \hat{u}_0 \bigr) (0) =0 ,
\]
which is due to the fact that the Mellin operator can be expressed as a local Hamiltonian operator, 
i.e., $\tilde{\sigma}' \sigma r^2(H-E_0) \tilde{\sigma}$
with cut-off functions $\tilde{\sigma}' \sigma \prec \tilde{\sigma}$ such that $Hu_0=E_0 u_0$ is satisfied on 
the support of $\tilde{\sigma}' \sigma$.
The remaining term becomes
\begin{eqnarray*}  
 \int \dbar \eta \, e^{iy \eta} M \bigl( \opm_M^{\gamma-1}(h_1^{(0)}) \tilde{\sigma} \hat{u}_0 \bigr) (0) & = &
 M \bigl( \opm_M^{\gamma-1}(\tfrac{1}{2t^2}(w^2-w)) \tilde{\sigma} u_0 \bigr) (0) \\
 & = & \tfrac{1}{2t^2}(w^2-w) M (\tilde{\sigma} u_0) |_{w=0} \\
 & = & \tfrac{1}{2t^2} \int_0^\infty dr r^{w-1} \biggl[ (-r \partial_r )^2 (\tilde{\sigma} u_0) - (-r \partial_r ) (\tilde{\sigma} u_0) \biggr] \bigg|_{w=0} \\
 & = & \tfrac{1}{2t^2} \int_0^\infty dr \biggl[ \partial_r \bigl( r \partial_r (\tilde{\sigma} u_0) \bigr) + \partial_r (\tilde{\sigma} u_0) \biggr] \\
 & = & - \tfrac{1}{2t^2} u_0(0) \\
 & = & - \tfrac{1}{2t^2} e^{-Zt} .
\end{eqnarray*}
In order to calculate
\[
 \int \dbar \eta \, e^{iy \eta} {\cal Q}_{0,2}(\hat{u}_0) 
\]
let us replace the product of pseudo-differential operators by a product
of the corresponding differential operators and using the same arguments as before
we obtain
\begin{eqnarray*}  
 \int \dbar \eta \, e^{iy \eta} {\cal Q}_{0,2}(\hat{u}_0) & = & 
 \bigl( -t^2 \partial_t^2 - 9t \partial_t \bigr) \int \dbar \eta \, e^{iy \eta} 
 M \bigl( \opm_M^{\gamma-1}(h_1^{(0)}) \tilde{\sigma} \hat{u}_0 \bigr) (0) \\
 & = & \bigl( -t^2 \partial_t^2 - 9t \partial_t \bigr) \bigl( - \tfrac{1}{2t^2} \bigr) e^{-Zt} \\ 
 & = & \bigl( 12+5Zt-Z^2t^2 \bigr) \bigl( - \tfrac{1}{2t^2} \bigr) e^{-Zt} .
\end{eqnarray*}
Putting things together one recovers the required asymptotic identity
\begin{eqnarray*}  
 Gu_0(r,y) & = & \int \dbar \eta \, e^{iy \eta} g \hat{u}_0(y,\eta) \\
 & \sim & 2t^2 \biggl[ \biggl( 1 +rtZ_1 +r^2 \bigl( -2 +\tfrac{1}{3}(tZ_1)^2
 + \tfrac{1}{3} tZ_2 \bigr) \biggr) \int \dbar \eta \, e^{iy \eta} {\cal Q}_{0,1}(\hat{u}_0)(y,\eta) \\
 & & +\tfrac{1}{6} r^2 \int \dbar \eta \, e^{iy \eta} {\cal Q}_{0,2}(\hat{u}_0)(y,\eta) + \cdots \biggr] \\
 & \sim & - \bigl( 1 -rtZ -2r^2 +\tfrac{2}{3}Z^2 (tr)^2 - \tfrac{1}{3} Zt r^2 \bigr) e^{-Zt} 
 - \tfrac{1}{6} r^2 \bigl( 12+5Zt-Z^2t^2 \bigr) e^{-Zt} \\
 &\sim & - \bigl( 1 -rtZ + \tfrac{1}{2} tZ r^2 + \tfrac{1}{2} Z^2 (tr)^2 \bigr) e^{-Zt} \\
 & \sim & -u_0 ,
\end{eqnarray*}
where we used $Z_1=-Z$ and $Z_2 = -tE_0 -Z= t-Z$ in the noninteracting case.

For $l=1$ let us consider the noninteracting wavefunction
\[
 u_1(r,t)= t \sin r e^{-\tfrac{Z}{2}t\sin r} e^{-Zt\cos r} Y_{1,m}(\theta_1,\phi_1) ,
\]
with asymptotic expansion 
\[
 u_1(r,t) \sim tr \bigl( 1 -Ztr \bigr) e^{-Zt} Y_{1,m}(\theta_1,\phi_1) .
\]
Once again we want to calculate
\[
 Gu_1(r,y) = \int \dbar \eta \, e^{iy \eta} g \hat{u}_1(y,\eta)
\]
in the asymptotic limit $r \rightarrow 0$.
Following the same line of arguments as before we get
\begin{eqnarray*}
\int \dbar \eta \, e^{iy \eta} {\cal Q}_{1,1}(\hat{u}_1) & = & 
 \int \dbar \eta \, e^{iy \eta} M \bigl( \opm_M^{\gamma-1}(h_2^{(1)}) \tilde{\sigma} \hat{u}_1 \bigr) (-1)
  -tZ_1 \int \dbar \eta \, e^{iy \eta} M \bigl( \opm_M^{\gamma-1}(h_1^{(1)}) \tilde{\sigma} \hat{u}_1 \bigr) (0) \\
 & = & \tfrac{1}{2t^2}(w^2-w-2) M (\tilde{\sigma} u_1) |_{w=-1} - \tfrac{Z_1}{2t}(w^2-w) M (\tilde{\sigma} u_1) |_{w=0}\\
 & = & \tfrac{1}{2t^2} \int_0^\infty dr r^{w-1} \bigl[ r^2 \partial_r^2 (\tilde{\sigma} u_1) +2r \partial_r (\tilde{\sigma} u_1)
  -2  \tilde{\sigma} u_1 \bigr] |_{w=-1} + \tfrac{Z_1}{2t} \int_0^\infty dr \partial_r (\tilde{\sigma} u_1) .
\end{eqnarray*}
Let us now substitute $u_1= r \tilde{u}_1$, from which it follows that the last integral vanishes and one gets
\begin{eqnarray*}
\int \dbar \eta \, e^{iy \eta} {\cal Q}_{1,1}(\hat{u}_1) 
 & = & \tfrac{1}{2t^2} \int_0^\infty dr \bigl[ \partial_r^2 (\tilde{\sigma} r \tilde{u}_1) +2\partial_r (\tilde{\sigma} \tilde{u}_1) \bigr] \\
 & = & \tfrac{1}{2t^2} \bigl[ - \partial_r (r \tilde{u}_1) (0) -2 \tilde{u}_1 (0) \bigr] \\
 & = & -\tfrac{3}{2t^2} \tilde{u}_1 (0) \\
 & = & -\tfrac{3}{2t} e^{-Zt} Y_{1,m}(\theta_1,\phi_1) .
\end{eqnarray*}
The asymptotic expansion becomes
\[
 Gu_1(r,y) = \int \dbar \eta \, e^{iy \eta} g \hat{u}_1(y,\eta) 
 \sim  -\tfrac{1}{2t^2} \bigl( rt - \tfrac{1}{2} Z (tr)^2 \bigr) e^{-Zt} Y_{1,m}(\theta_1,\phi_1) \sim - u_1 .
\]

Finally, for $l=2$ let us consider the noninteracting wavefunction
\[
 u_2(r,t)= t^2 \sin^2 r e^{-\tfrac{Z}{3}t\sin r} e^{-Zt \cos r} Y_{2,m}(\theta_1,\phi_1) 
 \sim (tr)^2 e^{-Zt} Y_{2,m}(\theta_1,\phi_1) .
\]
In this case the corresponding coefficients become
\begin{eqnarray*}
 \int \dbar \eta \, e^{iy \eta} {\cal Q}_{2,1}(\hat{u}_2) & = & 
 \int \dbar \eta \, e^{iy \eta} M \bigl( \opm_M^{\gamma-1}(h_3^{(2)}) \tilde{\sigma} \hat{u}_2 \bigr) (-2)
 -\tfrac{1}{2} tZ_1 \int \dbar \eta \, e^{iy \eta} M \bigl( \opm_M^{\gamma-1}(h_2^{(2)}) \tilde{\sigma} \hat{u}_2 \bigr) (-1) \\
 & & + \tfrac{1}{6} \bigl( 10 +(tZ_1)^2 -2tZ_2 \bigr) \int \dbar \eta \, e^{iy \eta}
 M \bigl( \opm_M^{\gamma-1}(h_1^{(2)}) \tilde{\sigma} \hat{u}_2 \bigr) (0) ,\\
 \int \dbar \eta \, e^{iy \eta} {\cal Q}_{2,2}(\hat{u}_2) & = &
 \int \dbar \eta \, e^{iy \eta}
 M \bigl( \opm_M^{\gamma-1} \bigl( \sum_\alpha \tfrac{1}{\alpha !} \partial_\eta^\alpha
 \tilde{C}_1 D_y^\alpha h_1^{(2)} \bigr) \tilde{\sigma} \hat{u}_2 \bigr) (0) .
\end{eqnarray*}
Setting $u_2=r^2 \tilde{u}_2$, explicit calculations yield
\begin{eqnarray*}
 \int \dbar \eta \, e^{iy \eta} M \bigl( \opm_M^{\gamma-1}(h_1^{(2)}) \tilde{\sigma} \hat{u}_2 \bigr) (0)
 & = & -\tfrac{1}{2t^2} \int_0^\infty dr 6r \tilde{u} ,\\
 \int \dbar \eta \, e^{iy \eta} M \bigl( \opm_M^{\gamma-1}(h_2^{(2)}) \tilde{\sigma} \hat{u}_2 \bigr) (-1)
 & = & -\tfrac{1}{2t^2} \int_0^\infty dr (4+2trZ_1) \tilde{u} ,\\
 \int \dbar \eta \, e^{iy \eta} M \bigl( \opm_M^{\gamma-1}(h_3^{(2)}) \tilde{\sigma} \hat{u}_2 \bigr) (-2)
 & = & \tfrac{1}{2t^2} \biggl( w^2-w-6 + \tfrac{8}{3} r^2 w 
 +(rt)^2 \partial_t^2 +5r^2t \partial_t -2rt Z_1 2r^2 -2r^2t Z_2 \biggr) \\
 & & \times M (\tilde{\sigma} u_2) |_{w=-2} \\
 & = & \tfrac{1}{2t^2} \int_0^\infty dr \biggl( 5 \partial_r \tilde{u} +rt^2 \partial_t^2 \tilde{u}
 +5rt \partial_t \tilde{u} -2tZ_1 \tilde{u} -2r \tilde{u} -2tr Z_2 \tilde{u} \biggr) .
\end{eqnarray*}
Summing up, one gets
\[
 \int \dbar \eta \, e^{iy \eta} {\cal Q}_{1,2}(\hat{u}_2) =
 \tfrac{1}{2t^2} \int_0^\infty dr \biggl( 5 \partial_r \tilde{u} +rt^2 \partial_t^2 \tilde{u} 
 +5rt \partial_t \tilde{u} -12 r \tilde{u} \biggr) ,
\]
\begin{eqnarray*}  
 \int \dbar \eta \, e^{iy \eta} {\cal Q}_{2,2}(\hat{u}_2) & = & 
 \bigl( -t^2 \partial_t^2 - 9t \partial_t \bigr) \int \dbar \eta \, e^{iy \eta} 
 M \bigl( \opm_M^{\gamma-1}(h_1^{(2)}) \tilde{\sigma} \hat{u}_2 \bigr) (0) \\
 & = & \bigl( -t^2 \partial_t^2 - 9t \partial_t \bigr) \bigl( - \tfrac{1}{2t^2} \bigr) \int_0^\infty dr 6r \tilde{u} \\
 & = & \tfrac{1}{2t^2} \int_0^\infty dr \biggl( 6rt^2 \partial_t^2 \tilde{u} +30 rt \partial_t \tilde{u} -72 r \tilde{u} \biggr)
\end{eqnarray*}
and, finally, 
\begin{eqnarray*}
 Gu_2(r,y) & = & \int \dbar \eta \, e^{iy \eta} g \hat{u}_2(y,\eta) \\
 & \sim & 2t^2 \biggl[ \tfrac{1}{5} r^2 \int \dbar \eta \, e^{iy \eta} {\cal Q}_{1,2}(\hat{u}_2) - 
 \tfrac{1}{30} r^2 \int \dbar \eta \, e^{iy \eta} {\cal Q}_{2,2}(\hat{u}_2) \cdots \biggr] \\
 & \sim & r^2 \int_0^\infty dr \partial_r \tilde{u} \\
 & \sim & - r^2 \tilde{u}_2(0) \\
 & \sim & - (tr)^2 e^{-Zt} Y_{2,m}(\theta_1,\phi_1) \\
 & \sim & -u_2 .
\end{eqnarray*} 

\section{Location and multiplicity of poles of Mellin type symbols}
\label{poles}
In Section \ref{nolog} we have discussed the absence of logarithmic terms in the edge asymptotic behaviour of eigenfunctions of the Hamiltonian. Within our approach, this
follows from the multiplicity of poles of Mellin type symbols
of the parametrix. The present work did not attempt to give
a complete description of the location and multiplicity of these poles. Intead we merely want to summarize in this Appendix 
our findings from the previous calculations. 

Let us start with the symbol $a_0^{-1}$, where our calculations
can be subsumed in the following remark which shows that all
of its poles are simple.

\begin{remark}
The poles of the meromorphic operator valued symbols $q_{0,n}$, $n=0,1,\ldots$, are determined by 
the recursive relation ${\mathrm(\ref{bn})}.$ 
 A simple calculation shows that 
\[
 \bigl( h_0-2n(2w-5-2n) \bigr)^{-1}
\]
is a meromorphic operator valued symbol with simple poles at $w=2n+3+l$ and $w=2n+2-l$ with $l=0,1,\ldots$,
which follows from the spectral decomposition of the Laplace-Beltrami operator on $S^2$, i.e.,
$\Delta_{S^2} = -\sum_{l=0}^\infty l(l+1) P_l$.
Furthermore, it follows from $2m+3+l \neq 2n+2-l$, i.e., $2(n-m) \neq 2l+1$, that $d_{j,n}$, $j,n=0,1,\ldots$,
are meromorphic operator valued symbols with simple poles at $w \in \mathbb{Z}$.
\end{remark}

At next let us consider the operator valued symbol $a_1^{-1}$, where a new type of denominator appears in $d^{(1)}_{0,0}$ which is of the form
\[
 h_0 \bigl( h_0-(2w-6) \bigr) .
\]
Resolving $\Delta_{S^2}$ like in the previous remark, we get
\[
 {\cal P}_l h_0 = (w-2)^{2} - (w-2) - l(l+1) = (w-3-l)(w-2+l)
\]
with simple zeros at $w_1=3+l$ and $w_2=2-l$, as well as
\[
 {\cal P}_l \bigl( h_0-(2w-6) \bigr) = (w-4-l)(w-3+l)
\]
with simple zeros at $w_3=4+l$ and $w_4=3-l$, respectively. 
Combining both factors it turns out that only for $l=0$
a multiple zero appears at $w_1=w_4=3$. 

Closing our discussion, we consider the operator valued symbol   $a_2^{-1}$, where the denominator of $d^{(2)}_{0,0}$ is of the form
\[
 h_0 \bigl( h_0-(2w-6) \bigr) \bigl( h_0-2(2w-7) \bigr) .
\]
The additional factor $\bigl( h_0-2(2w-7) \bigr)$ has
simple poles at $w_5=5+l$ and $w_6=4-l$, respectively.
Therefore, multiple zeros appear for $l=0$ at $w_3=w_6=4$
and $w_1=w_4=3$, respectively.

\section{Calculation of the residues}
\label{Appendixresidues}
In this appendix, we provide the necessary residues of the meromorphic operator valued symbols of the asymptotic parametrix
up to second order. Let us first list
the angular momentum resolved shifted symbols. The two shifted symbols contributing to zero and first order are given by
\begin{eqnarray*}
 &&T^n d^{(0)}_{0,0}(w) =  \sum_l \frac{{\cal  P}_l}{\bigl( w -(l+3-n) \bigr) \bigl( w -(2-l-n) \bigr)}\\
 &\mathrm{and}\\
&& T^n d^{(1)}_{0,0}(w) =  2tZ_1 \sum_l 
 \frac{{\cal  P}_l}{\bigl( w -(l+3-n) \bigr) \bigl( w -(2-l-n) \bigr) 
 \bigl( w -(l+4-n) \bigr) \bigl( w -(3-l-n) \bigr)} ,
\end{eqnarray*}
respectively. For the symbols contributing in second order it
is convenient to define the meromorphic function
\[
 h_n(w) := \sum_l 
\frac{{\cal  P}_l}{\bigl( w -(l+3-n) \bigr) \bigl( w -(2-l-n) \bigr) 
 \bigl( w -(l+5-n) \bigr) \bigl( w -(4-l-n) \bigr)} ,
\]
with it these symbols become
\[
 \partial_\eta^\alpha T^n d^{(0)}_{0,2}(w) = r^2 \partial_\eta^\alpha C_0 h_n(w), \quad
 \partial_\eta^\alpha T^n d^{(0)}_{1,1}(w) = -r \partial_\eta^\alpha (4it\tau) h_n(w), \quad
 \partial_\eta^\alpha T^n d^{(0)}_{2,0}(w) = -2 h_n(w) ,
\]
\[
 T^n d^{(1)}_{1,0}(w) = -10 h_n(w), \quad
 \partial_\eta^\alpha T^n d^{(1)}_{0,1}(w) = ir \partial_\eta^\alpha C_1 h_n(w) ,
\]
\[
 T^n d^{(2)}_{0,0}(w) = \sum_l \biggl[
 \tfrac{(2tZ_1)^2}
 {\bigl( w -(l+4-n) \bigr) \bigl( w -(3-l-n) \bigr)}
 -\tfrac{1}{3} \bigl( 8(w-2+n) -l(l+1) -6tZ_2 \bigr) \biggr] h_n(w) . 
\]
The corresponding residues which are required in our calculations
for the entire $a$-type Green operator, defined in Section
(\ref{sumatypeG}), are given by
\[
\begin{split}
 g_a: & \quad \quad \res(T^2 d^{(0)}_{0,0}, -l) = -\tfrac{1}{2l+1} {\cal  P}_l,\\
     & \quad \quad \res( T^2 d^{(1)}_{0,0}, -l) {\cal  P}_l = - \tfrac{2tZ_1}{2(2l+1) (l+1)} {\cal  P}_l,
 \quad (l \geq 0), \\ & \quad \quad 
 \res( T^2 d^{(1)}_{0,0}, 1-l) {\cal  P}_l = \tfrac{2tZ_1}{2l(2l+1)}{\cal  P}_l \quad (l \geq 1), \\
 & \quad \quad \res( T^2 d^{(1)}_{0,0}, -m) = - \tfrac{2tZ_1}{2(2m+1) (m+1)} {\cal  P}_m 
 + \tfrac{2tZ_1}{2(m+1)(2m+3)}{\cal  P}_{1+m} \quad ( m \geq 0) ,\\
 & \quad \quad \res(T^2 d^{(2)}_{0,0}, 1-l) {\cal  P}_l = \tfrac{(2tZ_1)^2}
 {4l(2l+1)(l+1)}{\cal  P}_l \quad (l \geq 1), \\
 & \quad \quad \res(T^2 d^{(2)}_{0,0},-l) {\cal  P}_l = -\biggl[
 \tfrac{(2tZ_1)^2}{4(l+1)(2l+1)(2l+3)}
 + \tfrac{l(l+9) +6tZ_2}{6(2l+1)(2l+3)} \biggr]{\cal  P}_l \quad (l \geq 0), \\
 & \quad \quad \res(T^2 d^{(2)}_{0,0},2-l){\cal  P}_l = \biggl[
 -\tfrac{(2tZ_1)^2}{4l(2l-1)(2l+1)}
 + \tfrac{l(l+9)-16 +6tZ_2}{6(2l-1)(2l+1)} \biggr]{\cal  P}_l \quad (l \geq 2), \\
 & \quad \quad \res(T^2 d^{(2)}_{0,0},-m) = -\biggl[ \tfrac{(2tZ_1)^2}{4(m+1)(2m+1)(2m+3)}
 + \tfrac{m(m+9) +6tZ_2}{6(2m+1)(2m+3)} \biggr] {\cal  P}_m
 + \tfrac{(2tZ_1)^2}{4(m+1)(2m+3)(m+2)}{\cal  P}_{1+m} \\
 & \quad \quad \hspace{3.5cm} 
 + \biggl[ -\tfrac{(2tZ_1)^2}{4(m+2)(2m+3)(2m+5)} + \tfrac{(m+2)(m+11)-16 +6tZ_2}{6(2m+3)(2m+5)} \biggr] {\cal  P}_{2+m}
 \quad (m \geq 0) ,\\
  & \quad \quad \res(\partial_\eta^\alpha T^2 d^{(0)}_{0,2}, -l){\cal  P}_l = - r^2 \partial_\eta^\alpha C_0 \tfrac{1}{2(2l+1) (2l+3)} P_l,
 \quad (l \geq 0), \\ & \quad \quad 
 \res(\partial_\eta^\alpha T^2 d^{(0)}_{0,2}, 2-l){\cal  P}_l = r^2 \partial_\eta^\alpha C_0 \tfrac{1}{2(2l+1) (2l-1)}{\cal  P}_l \quad (l \geq 2), \\
 & \quad \quad \res(\partial_\eta^\alpha T^2 d^{(0)}_{0,2}, -m) = r^2 \partial_\eta^\alpha C_0 \bigl(
 -\tfrac{1}{2(2m+1) (2m+3)} {\cal  P}_m +\tfrac{1}{2(2m+5) (2m+3)} {\cal  P}_{2+m} \bigr) \quad (m \geq 0) ,
 \\
  & \quad \quad \res(\partial_\eta^\alpha T^2 d^{(1)}_{0,1}, -l){\cal  P}_l = -ir \partial_\eta^\alpha C_1 \tfrac{1}{2(2l+1) (2l+3)} {\cal  P}_l,
 \quad (l \geq 0), \\ & \quad \quad 
 \res(\partial_\eta^\alpha T^2 d^{(1)}_{0,1}, 2-l) {\cal  P}_l = ir \partial_\eta^\alpha C_1 \tfrac{1}{2(2l+1) (2l-1)} {\cal  P}_l \quad (l \geq 2), \\
 & \quad \quad \res(\partial_\eta^\alpha T^2 d^{(1)}_{0,1}, -m) = ir \partial_\eta^\alpha C_1 \bigl(
 -\tfrac{1}{2(2m+1) (2m+3)}{\cal  P}_m +\tfrac{1}{2(2m+5) (2m+3)} {\cal  P}_{2+m} \bigr) \quad (m \geq 0) .
\end{split}
\]
Finally, the remaining residues corresponding to various $b$-type Green operators are given by
\[
\begin{split}
 g_{0,2}: & \quad \quad \res(d^{(0)}_{0,0},1) = -\tfrac{1}{3}{\cal  P}_1, \quad \res(d^{(0)}_{0,0},2) = -{\cal  P}_0 ,
\\
 & \quad \quad \res(\partial_\eta^\alpha d^{(0)}_{0,2},1) = r^2 \partial_\eta^\alpha C_0 \biggl[
 -\tfrac{1}{30} {\cal  P}_1 +\tfrac{1}{70}{\cal  P}_3 \biggr], \\
 & \quad \quad \res(\partial_\eta^\alpha d^{(0)}_{0,2},2) = r^2 \partial_\eta^\alpha C_0 \biggl[
 -\tfrac{1}{6} {\cal  P}_0 +\tfrac{1}{30}{\cal  P}_2 \biggr] ,
\\
 g_{0,3}: & \quad \quad \res(T^{-1} d^{(0)}_{0,0},1) = -\tfrac{1}{5} {\cal  P}_2, \quad \res(T^{-1} d^{(0)}_{0,0},2) = -\tfrac{1}{3} {\cal  P}_1,
 \quad \res(T^{-1} d^{(0)}_{0,0},3) = -{\cal  P}_0 ,\\
  & \quad \quad \res(T^{-1} \partial_\eta^\alpha d^{(0)}_{0,2},1) = r^2 \partial_\eta^\alpha C_0 \biggl[ 
 -\tfrac{1}{70} {\cal  P}_2 +\tfrac{1}{126}{\cal  P}_4 \biggr], \\
 & \quad \quad \res(T^{-1} \partial_\eta^\alpha d^{(0)}_{0,2},2) = r^2 \partial_\eta^\alpha C_0 \biggl[
 -\tfrac{1}{30} {\cal  P}_1 +\tfrac{1}{70} {\cal  P}_3 \biggr], \\
 & \quad \quad \res(T^{-1} \partial_\eta^\alpha d^{(0)}_{0,2},3) = r^2 \partial_\eta^\alpha C_0 \biggl[
 -\tfrac{1}{6}{\cal  P}_0 +\tfrac{1}{30} {\cal  P}_2 \biggr] , 
\end{split}
\]
\[
\begin{split}
  g_{1,2}: & \quad \quad \res(T d^{(0)}_{0,0},1) = -{\cal  P}_0, \quad \res(d^{(0)}_{0,0},1) = -\tfrac{1}{3} {\cal  P}_1,
 \quad \res(d^{(0)}_{0,0},2) = -{\cal  P}_0 , \\
 & \quad \quad \res(T \partial_\eta^\alpha d^{(0)}_{0,2},1) = r^2 \partial_\eta^\alpha C_0 \biggl[
 -\tfrac{1}{6}{\cal  P}_0 +\tfrac{1}{30} {\cal  P}_2 \biggr], \\
 & \quad \quad \res(\partial_\eta^\alpha d^{(0)}_{0,2},1) = r^2 \partial_\eta^\alpha C_0 \biggl[
 -\tfrac{1}{30}{\cal  P}_1 +\tfrac{1}{70} {\cal  P}_3 \biggr], \\
 & \quad \quad \res(\partial_\eta^\alpha d^{(0)}_{0,2},2) = r^2 \partial_\eta^\alpha C_0 \biggl[
 -\tfrac{1}{6} {\cal  P}_0 +\tfrac{1}{30}{\cal  P}_2 \biggr] ,\\
 & \quad \quad \res( d^{(1)}_{0,0},1) = 2tZ_1 \biggl[
 -\tfrac{1}{12}{\cal  P}_1 +\tfrac{1}{20}{\cal  P}_2 \biggr], \\
 & \quad \quad \res( d^{(1)}_{0,0},2) = 2tZ_1 \biggl[
 -\tfrac{1}{2} {\cal  P}_0 +\tfrac{1}{6}{\cal  P}_1 \biggr] , \\
  & \quad \quad \res( \partial_\eta^\alpha d^{(1)}_{0,1},1) = ir \partial_\eta^\alpha C_1 \biggl[
 -\tfrac{1}{30}{\cal  P}_1 +\tfrac{1}{70}{\cal  P}_3 \biggr], \\
 & \quad \quad \res(\partial_\eta^\alpha d^{(1)}_{0,1},2) = ir \partial_\eta^\alpha C_1 \biggl[
 -\tfrac{1}{6} {\cal  P}_0 +\tfrac{1}{30} {\cal  P}_2 \biggr] ,\\
 g_{1,3}: & \quad \quad \res( d^{(1)}_{0,0},1) = 2tZ_1 \biggl[
 -\tfrac{1}{12}{\cal  P}_1 +\tfrac{1}{20}{\cal  P}_2 \biggr], \\
 & \quad \quad \res( d^{(1)}_{0,0},2) = 2tZ_1 \biggl[
 -\tfrac{1}{2}{\cal  P}_0 +\tfrac{1}{6}{\cal  P}_1 \biggr] , \\
  & \quad \quad \res( \partial_\eta^\alpha d^{(1)}_{0,1},1) = ir \partial_\eta^\alpha C_1 \biggl[
 -\tfrac{1}{30} {\cal  P}_1 +\tfrac{1}{70} {\cal  P}_3 \biggr], \\
 & \quad \quad \res(\partial_\eta^\alpha d^{(1)}_{0,1},2) = ir \partial_\eta^\alpha C_1 \biggl[
 -\tfrac{1}{6} {\cal  P}_0 +\tfrac{1}{30} {\cal  P}_2 \biggr] ,
\end{split}
\]
\[
\begin{split}
 g_{2,2}:
& \quad \quad \res(d^{(0)}_{0,0},1) = -\tfrac{1}{3} {\cal  P}_1, \quad \res(d^{(0)}_{0,0},2) = - {\cal  P}_0 , 
 \quad \res(T^{-2} d^{(0)}_{0,0},1) = -\tfrac{1}{7} {\cal  P}_3, \\
 & \quad \quad \res(T^{-2} d^{(0)}_{0,0},2) = -\tfrac{1}{5} {\cal  P}_2,
 \quad \res(T^{-2} d^{(0)}_{0,0},3) = -\tfrac{1}{3} {\cal  P}_1, \quad \res(T^{-2} d^{(0)}_{0,0},4) = -{\cal  P}_0 ,\\
& \quad \quad \res( d^{(2)}_{0,0},1) = - \biggl[ \tfrac{1}{30} (tZ_1)^2 
 + \tfrac{1}{45} (5+3tZ_2) \biggr]{\cal  P}_1 
 + \tfrac{1}{30} (tZ_1)^2 {\cal  P}_2 \\ & \hspace{3cm} + \biggl[ -\tfrac{1}{105} (tZ_1)^2
 + \tfrac{1}{105} (10+3tZ_2) \biggr]{\cal  P}_3 \\
 & \quad \quad \res( d^{(2)}_{0,0},2) = - \biggl[ \tfrac{1}{3} (tZ_1)^2
 + \tfrac{1}{3} tZ_2 \biggr] {\cal  P}_0
 + \tfrac{1}{6} (tZ_1)^2 {\cal  P}_1 \\ & \hspace{3cm} + \biggl[ -\tfrac{1}{30} (tZ_1)^2
 + \tfrac{1}{15} (1+tZ_2) \biggr]{\cal  P}_2\\
  & \quad \quad \res( T d^{(1)}_{0,0},1) = 2tZ_1 \biggl[
 -\tfrac{1}{2} {\cal  P}_0 +\tfrac{1}{6} {\cal  P}_1 \biggr], \\
 & \quad \quad \res( d^{(1)}_{0,0},1) = 2tZ_1 \biggl[
 -\tfrac{1}{12} {\cal  P}_1 +\tfrac{1}{20}{\cal  P}_2 \biggr], \\
 & \quad \quad \res( d^{(1)}_{0,0},2) = 2tZ_1 \biggl[
 -\tfrac{1}{2}{\cal  P}_0 +\tfrac{1}{6}{\cal  P}_1 \biggr] ,\\
& \quad \quad \res(\partial_\eta^\alpha d^{(0)}_{0,2},1) = r^2 \partial_\eta^\alpha C_0 \biggl[
 -\tfrac{1}{30} {\cal  P}_1 +\tfrac{1}{70} {\cal  P}_3 \biggr], \\
 & \quad \quad \res(\partial_\eta^\alpha d^{(0)}_{0,2},2) = r^2 \partial_\eta^\alpha C_0 \biggl[
 -\tfrac{1}{6} {\cal  P}_0 +\tfrac{1}{30} {\cal  P}_2 \biggr], \\
 & \quad \quad \res(T^{-2} \partial_\eta^\alpha d^{(0)}_{0,2},1) = r^2 \partial_\eta^\alpha C_0 \biggl[
 -\tfrac{1}{126}{\cal  P}_3 +\tfrac{1}{198}{\cal  P}_5 \biggr], \\
 & \quad \quad \res(T^{-2} \partial_\eta^\alpha d^{(0)}_{0,2},2) = r^2 \partial_\eta^\alpha C_0 \biggl[
 -\tfrac{1}{70} {\cal  P}_2 +\tfrac{1}{126} {\cal  P}_4 \biggr], \\
 & \quad \quad \res(T^{-2} \partial_\eta^\alpha d^{(0)}_{0,2},3) = r^2 \partial_\eta^\alpha C_0 \biggl[
 -\tfrac{1}{30} {\cal  P}_1 +\tfrac{1}{70}{\cal  P}_3 \biggr], \\
 & \quad \quad \res(T^{-2} \partial_\eta^\alpha d^{(0)}_{0,2},4) = r^2 \partial_\eta^\alpha C_0 \biggl[
 -\tfrac{1}{6} {\cal  P}_0 +\tfrac{1}{30} {\cal  P}_2 \biggr] , \\
& \quad \quad \res(T \partial_\eta^\alpha d^{(1)}_{0,1},1) = ir \partial_\eta^\alpha C_1 \biggl[
 -\tfrac{1}{6}{\cal  P}_0 +\tfrac{1}{30}{\cal  P}_2 \biggr], \\
 & \quad \quad \res(\partial_\eta^\alpha d^{(1)}_{0,1},1) = ir \partial_\eta^\alpha C_1 \biggl[
 -\tfrac{1}{30} {\cal  P}_1 +\tfrac{1}{70} {\cal  P}_3 \biggr], \\
 & \quad \quad \res(\partial_\eta^\alpha d^{(1)}_{0,1},2) = ir \partial_\eta^\alpha C_1 \biggl[
 -\tfrac{1}{6} {\cal  P}_0 +\tfrac{1}{30} {\cal  P}_2 \biggr] ,
\\
 g_{2,3}: & \quad \quad \res( T d^{(2)}_{0,0},1) = - \biggl[ \tfrac{1}{3} (tZ_1)^2 
 + \tfrac{1}{3} tZ_2 \biggr] {\cal  P}_0 
 + \tfrac{1}{6} (tZ_1)^2{\cal  P}_1 \\ & \hspace{3cm} + \biggl[ -\tfrac{1}{30} (tZ_1)^2
 + \tfrac{1}{15} (1+tZ_2) \biggr]{\cal  P}_2.
\end{split}
\]


\clearpage
\begin{thebibliography}{999}
\bibitem{AM87} Abbott  P.~C. and  Maslen E.~N., {\em Coordinate systems and analytic expansions for three-body
               atomic wavefunctions: I. Partial summation for the Fock expansion in hyperspherical 
               coordinates}. J. Phys. A: Math. Gen. {\bf 20} (1987), 2043-2075.  
\bibitem{ACM14}  Akutagawa K., Carron G. and Mazzeo R., {\em H\"older regularity of solutions for Schr\"odinger operators on stratified spaces}. ArXiv:1409.0154v1 (2014).
\bibitem{ACN10} Ammann B.,  Carvalho C. and Nistor V., {\em Regularity for eigenfunctions of Schr\"odinger operators}.
                Lett. Math. Phys.,{\bf 101} (2012), 49-98.
\bibitem{Demkov} Demkov Y.N. and Ermolaev A.M., Zh. Eksp. Teor. Fiz. {\bf 36}, 896 (1959) [Sov.Phys.JETP {\bf 36}, {\bf 633} (1959)].
\bibitem{Erm1} Ermolaev A.M., {\em Expansion of many-electron wave functions in the Fock series}. Vestn. Leningr. Univ. 14, No. {\bf 22}, {\bf 46} (1958).
\bibitem{Erm2} Ermolaev A.M., Vestn. Leningr. Univ. 14, No. {\bf 16}, {\bf 19} (1961).
\bibitem{ES97} Egorov Y.V., Schulze B.-W., {\em Pseudo-Differential Operators, Singularities,
               Applications}. Birkh\"auser: Basel; 1997.
\bibitem{FH10}  Flad H.-J. and  Harutyunyan G., \textit{Ellipticity of quantum mechanical Hamiltonians in the edge algebra}.
               Discrete and continuous dynamical systems, Supplement 2011, 420-429.
\bibitem{FHSS10} Flad H.-J., Harutyunyan G., Schneider R. and Schulze B.-W., {\em Explicit Green operators for 
                 quantum mechanical Hamiltonians.~I.~The hydrogen atom}. Manuscripta Mathematika, Vol. {\bf 135}, No. {\bf 3-4} (2011), 497-519. 
\bibitem{FHS15} Flad H.-J., Harutyunyan G. and Schulze B.-W., {\em 
                Singular analysis and coupled cluster theory}. Phys. Chem. Chem. Phys., Vol. {\bf 17} (2015), 31530-31541.
\bibitem{FHS16} Flad H.-J., Harutyunyan G. and Schulze B.-W., {\em Asymptotic parametrices of elliptic edge operators}.
                Journal of Pseudo-Differential Operators and Applications, Vol. {\bf 7}, Issue {\bf 3} (2016), 321-363.
\bibitem{FSS08}  Flad H.-J.,  Schneider R. and Schulze B.-W., {\em Regularity of solutions 
                of Hartree-Fock equations with Coulomb potential}. Math. Methods Appl. Sci. {\bf 31} (2008), 2172-2201.
\bibitem{Fock54}  Fock V.~A., {\em On the Schr\"odinger equation of the helium atom}. In
                 {\em Fock V.~A. - Selected Works: Quantum Mechanics and Quantum Field Theory}, 
                 Faddeev  L.~D.,  Khalfin L.~A. and  Komarov I.~V.,
                 Eds., (Chapman \& Hall/CRC, Boca Raton, 2004), 525-538.
\bibitem{FHO2S05} Fournais S.,  Hoffmann-Ostenhof M.,  Hoffmann-Ostenhof T., and {\O}stergaard S{\o}rensen T.,
                {\em Sharp regularity results for Coulombic many-electron wave functions}.
                Commun. Math. Phys. {\bf 255} (2005), 183-227.
\bibitem{FHO2S09} Fournais S.,  Hoffmann-Ostenhof M., Hoffmann-Ostenhof T., {\O}stergaard S{\o}rensen T., 
                  {\em Analytic structure of many-body Coulombic wave functions}.
                  Comm. Math. Phys. {\bf 289} (2009), 291-310.
\bibitem{FP66} Frankowski K. and  Pekeris C.~L., {\em Logarithmic terms in the wave functions of the ground state 
               of two-electron atoms}. Phys. Rev. {\bf 146} (1966), 46-49.
\bibitem{FHM84}  Freund D.~E.,  Huxtable B.~D. and  Morgan J.~D. III, {\em Variational calculations on the helium isoelectronic sequence}.
                Phys. Rev. A {\bf 29} (1984), 980-982.
\bibitem{GSS00}  Gil J.~B.,  Schulze B.-W.and  Seiler J., {\em Cone pseudo-differential operators in the
                edge symbolic calculus}. Osaka J. Math. {\bf 37} (2000), 221-260. 
\bibitem{GAM87}  Gottschalk J.~E.,  Abbott P.~C. and  Maslen E.~N., {\em Coordinate systems and analytic expansions for 
               three-body atomic wavefunctions: II. Closed form wavefunction to second order in $r$}.
               J. Phys. A: Math. Gen. {\bf 20} (1987), 2077-2104.  
\bibitem{GM87}  Gottschalk J.~E. and  Maslen E.~N., {\em Coordinate systems and analytic expansions for 
               three-body atomic wavefunctions: III. Derivative continuity via solutions to Laplace's equation}.
               J. Phys. A: Math. Gen. {\bf 20} (1987), 2781-2803.  
\bibitem{G63}  Granzow K.~D., {\em N-dimensional total orbital angular-monentum operator}.
              J. Math. Phys. {\bf 4} (1963), 897--900.
\bibitem{HS08} Harutyunyan G., Schulze B.-W., {\em Elliptic Mixed, Transmission and Singular Crack Problems}.
               EMS Tracts in Mathematics Vol. {\bf 4}, European Math. Soc: Z\"urich; 2008.
\bibitem{Hill85} Hill R.~N. , {\em Rates of convergence and error estimation formulas for the Rayleigh-Ritz
                 variational method}. J. Chem. Phys. {\bf 83} (1985), 1173-1196.
\bibitem{HOS81}  Hoffmann-Ostenhof M., and  Seiler R., {\em Cusp conditions for eigenfunctions of n-electron systems}.
             Phys. Rev. A {\bf 23} (1981), 21-23.
\bibitem{HO292}  Hoffmann-Ostenhof H. and Hoffmann-Ostenhof T., {\em Local properties of solutions
              of Schr\"odinger equations}. Commun. Partial Diff. Eq. {\bf 17} (1992), 491-522.
\bibitem{HO2S94} Hoffmann-Ostenhof M., Hoffmann-Ostenhof T., and Stremnitzer H., {\em Local properties of Coulombic
              wave functions}. Commun. Math. Phys. {\bf 163} (1994), 185-215.
\bibitem{HNS08} Hunsicker E., Nistor V. and  Sofo J.~O., {\em Analysis of periodic Schr\"odinger operators:
                regularity and approximation of eigenfunctions}. J. Math. Phys. {\bf 49} (2008) 083501 (21 pages).
\bibitem{Hyll29}  Hylleraas E.~A., {\em Neue Berechnung der Energie des Heliums im Grundzustande, sowie des tiefsten
                 Terms von Ortho-Helium}. Z. Phys. {\bf 54} (1929), 347-366.
\bibitem{Kato57}  Kato T., {\em On the eigenfunctions of many-particle systems in quantum mechanics}.
               Commun. Pure Appl. Math. {\bf 10} (1957), 151-177.
\bibitem{Leray1} Leray J., {\em Application a l'\'equation de Schr\"odinger atomique dune extension du theor\`{e}me de Fuchs}.
                Actes du 6\`{e}me Congres du Groupement des Mathematiciens d'Expression Latine,
                pp.~179-187, Gauthier-Villars (1982).
\bibitem{Leray2}  Leray J., {\em Sur les solutions de l'equation de Schr\"odinger atomique et les cas particulier
                 de deux electrons}.
                 In: Lecture Notes in Physics {\bf 195} (Springer, Berlin, 1984), 235-247.
\bibitem{Leray3}  Leray J., {\em La fonction de Green de la sphere $S^n$ et l'application effective a l'equation
                 de Schr\"odinger atomique dune extension du theor\`{e}me de Fuchs}.
                 Methods of functional analysis and theory of elliptic operators. Naples 1982: 
		 Universita di Napoli.
\bibitem{Luehr77}  L\"uhrmann K.~H., {\em Equations for subsystems}. Ann. Phys. {\bf 103} (1977), 253--288.
\bibitem{Morg86}  Morgan J.~D. III, {\em Convergence properties of Fock's expansion for S-state eigenfunctions
                of the helium atom}. Theor. Chim. Acta {\bf 69} (1986), 181-223.
\bibitem{MF53} P.~M.~Morse and H.~Feshbach, {\em Methods of Theoretical Physics}. (McGraw-Hill, New York, 1953).
\bibitem{MM65}  Mott N.~F. and  Massey H.~S.~W., {\em The Theory of Atomic Collisions}. 3rd. ed.
               (Oxford University Press, Oxford, 1965).
\bibitem{Pek58}  Pekeris C.~L., {\em Ground state of two-electron atoms}. Phys. Rev. {\bf 112} (1958), 1649-1658. 
\bibitem{RS4} Reed M. and Simon B., {\em Methods of Modern Mathematical Physics. Volume IV: Analysis of Operators}.
(Academic Press, 1978).                
\bibitem{Schulze98} Schulze B.-W., {\em Boundary Value Problems and Singular Pseudo-Differential Operators}.
                    Wiley: New York; 1998.
\bibitem{Yser04} Yserentant H., \emph{On the regularity of the electronic
                 Schr\"{o}dinger equation in Hilbert spaces of mixed derivatives}. Numer. Math.
                 \textbf{98} (2004), 731-759.
\bibitem{Yser05} Yserentant H., \emph{Sparse grid spaces for the numerical
                 solution of the electronic Schr\"{o}dinger equation}. Numer. Math.
                 \textbf{101} (2005), 381-389.
\end{thebibliography}
\end{document}